\documentclass{article}
\usepackage[linesnumbered,ruled,vlined]{algorithm2e}
\usepackage{amsmath,amsfonts, amsthm, amssymb}
\usepackage{dsfont}
\usepackage{cases}
\usepackage{tablefootnote,graphicx}
\usepackage{natbib}

\usepackage{refcount}
\usepackage{geometry}
\usepackage{booktabs}
\usepackage{xcolor,colortbl}
\usepackage{makecell}
\usepackage{multirow}
\usepackage[unicode=true,
 bookmarks=false,
 breaklinks=false,pdfborder={0 0 1},colorlinks=false]
 {hyperref}
\hypersetup{colorlinks,citecolor=blue,filecolor=blue,linkcolor=blue,urlcolor=blue}
\newtheorem{lemma}{Lemma}
\newtheorem{remark}{Remark}
\newtheorem{theorem}{Theorem}
\newtheorem{definition}{Definition}
\newcounter{tmp}

\usepackage{amssymb}
\usepackage{pifont}
\newcommand{\cmark}{\ding{51}}%
\newcommand{\xmark}{\ding{55}}%

\definecolor{LightCyan}{rgb}{0.88,1,1}

\definecolor{dacong}{RGB}{10,103,68}

\definecolor{yc}{RGB}{200,3,68}

\allowdisplaybreaks

\title{Faster Last-iterate Convergence of Policy Optimization in Zero-Sum Markov Games\footnotetext{Authors are listed alphabetically.}}
\author{Shicong Cen\textsuperscript{1}\thanks{Department of Electrical and Computer Engineering, Carnegie Mellon University; email: \texttt{shicongc@andrew.cmu.edu}. }
		\qquad  Yuejie Chi\textsuperscript{1}\thanks{Department of Electrical and Computer Engineering, Carnegie Mellon University; email: \texttt{yuejiechi@cmu.edu}. }
	\qquad Simon S. Du\textsuperscript{2,3}\thanks{Paul G. Allen School of Computer Science and Engineering, University of Washington; email: \texttt{ssdu@cs.washington.edu}.}
	\qquad Lin Xiao\textsuperscript{3}\thanks{Meta AI Research; email: \texttt{linx@fb.com}.} \\[1ex]
	\textsuperscript{1}Carnegie Mellon University \\ \textsuperscript{2}University of Washington\\ \textsuperscript{3}Meta AI Research \\
	\\}

\newcommand{\ex}[2]{\mathbb{E}_{#1}\left[#2\right]}
\newcommand{\exlim}[2]{\mathop\mathbb{E}\limits_{#1}\left[#2\right]}
\newcommand{\s}{s} 
\newcommand{\prn}[1]{\left({#1}\right)} 
\newcommand{\prnbig}[1]{\big({#1}\big)} 
\newcommand{\brk}[1]{\left[{#1}\right]} 
\newcommand{\norm}[1]{\left\|{#1}\right\|} 
\newcommand{\normbig}[1]{\big\|{#1}\big\|} 
\newcommand{\innprod}[1]{\big\langle{#1} \big\rangle} 

\newcommand{\cS}{\mathcal{S}}
\newcommand{\cA}{\mathcal{A}}
\newcommand{\cB}{\mathcal{B}}
\newcommand{\cO}{\mathcal{O}}
\newcommand{\cC}{\mathcal{C}}

\newcommand{\cH}{\mathcal{H}}

\newcommand{\KL}[2]{\mathsf{KL}\prnbig{{#1}\,\|\,{#2}}}
\newcommand{\KLs}[2]{\mathsf{KL}_s\prnbig{{#1}\,\|\,{#2}}}
\newcommand{\KLrho}[2]{\mathsf{KL}_\rho\prnbig{{#1}\,\|\,{#2}}}
\newcommand{\KLbig}[2]{\mathsf{KL}\prnbig{{#1}\,\|\,{#2}}}

\newcommand{\one}{\mathbf{1}}
\newcommand{\ent}[1]{\cH\big(#1\big)}
\newcommand{\best}[1]{#1_\tau^\star}
\newcommand{\besth}[1]{#1_{h,\tau}^\star}

\newcommand{\Q}[1]{Q^{{({#1})}}}
\newcommand{\V}[1]{V^{{({#1})}}}
\newcommand{\bmut}{\bar{\mu}^{(t)}}
\newcommand{\bnut}{\bar{\nu}^{(t)}}
\newcommand{\bzt}{\bar{\zeta}^{(t)}}
\newcommand{\bmutp}{\bar{\mu}^{(t+1)}}
\newcommand{\bnutp}{\bar{\nu}^{(t+1)}}
\newcommand{\bztp}{\bar{\zeta}^{(t+1)}}

\newcommand{\mut}{\mu^{(t)}}
\newcommand{\nut}{\nu^{(t)}}
\newcommand{\zt}{\zeta^{(t)}}
\newcommand{\mutp}{\mu^{(t+1)}}
\newcommand{\nutp}{\nu^{(t+1)}}
\newcommand{\ztp}{\zeta^{(t+1)}}
\newcommand{\mutm}{\mu^{(t-1)}}
\newcommand{\nutm}{\nu^{(t-1)}}

\begin{document}

\maketitle

\begin{abstract}

Multi-Agent Reinforcement Learning (MARL)---where multiple agents learn to interact in a shared dynamic environment---permeates across a wide range of critical applications. While there has been substantial progress on understanding the global convergence of policy optimization methods in single-agent RL, designing and analysis of efficient policy optimization algorithms in the MARL setting present significant challenges, which unfortunately, remain highly inadequately addressed by existing theory. In this paper, we focus on the most basic setting of competitive multi-agent RL, namely two-player zero-sum Markov games, and study equilibrium finding algorithms in both the infinite-horizon discounted setting and the finite-horizon episodic setting. 
We propose a single-loop policy optimization method with symmetric updates from both agents, where the policy is updated via the entropy-regularized optimistic multiplicative weights update (OMWU) method and the value is updated on a slower timescale. We show that, in the full-information tabular setting, the proposed method achieves a finite-time last-iterate linear convergence to the quantal response equilibrium of the regularized problem, which translates to a sublinear last-iterate convergence to the Nash equilibrium by controlling the amount of regularization. Our convergence results improve upon the best known iteration complexities, and lead to a better understanding of policy optimization in competitive Markov games.

\end{abstract} 


\tableofcontents

\section{Introduction}

Policy optimization methods \citep{williams1992simple,sutton2000policy,kakade2002natural,peters2008natural,konda2000actor}, which cast sequential decision making  as value maximization problems with regards to (parameterized) policies, have been instrumental in enabling recent successes of reinforcement learning (RL). 
See e.g., \citet{schulman2015trust,schulman2017proximal,silver2016mastering}. Despite its empirical popularity, the theoretical underpinnings of policy optimization methods remain elusive until very recently. For single-agent RL problems, a flurry of recent works has made substantial progress on understanding the global convergence of policy optimization methods under the framework of Markov Decision Processes (MDP) \citep{agarwal2019optimality,bhandari2019global,mei2020global,cen2020fast,lan2021policy,bhandari2020note,zhan2021policy,khodadadian2021linear,xiao2022convergence}. Despite the nonconcave nature of value maximization, (natural) policy gradient methods are shown to achieve global convergence at a sublinear rate \citep{agarwal2019optimality,mei2020global} or even a linear rate in the presence of regularization \citep{mei2020global,cen2020fast,lan2021policy,zhan2021policy} when the learning rate is constant.

Moving beyond single-agent RL, Multi-Agent Reinforcement Learning (MARL) is the next frontier---where multiple agents learn to interact in a shared dynamic environment---permeating across critical applications such as multi-agent networked systems, autonomous vehicles, robotics, and so on. Designing and analysis of efficient policy optimization algorithms in the MARL setting present significant challenges and new desiderata, which unfortunately, remain highly inadequately addressed by existing theory.

\subsection{Policy optimization for competitive RL}

In this work, we focus on one of the most basic settings of competitive multi-agent RL, namely two-player zero-sum Markov games \citep{shapley1953stochastic}, and study equilibrium finding algorithms in both the infinite-horizon discounted setting and the finite-horizon episodic setting. In particular, our designs gravitate around algorithms that are \textit{single-loop}, \textit{symmetric}, with {\em finite-time} \textit{last-iterate} convergence to the Nash Equilibrium (NE) or Quantal Response Equilibrium (QRE) under bounded rationality, two prevalent solution concepts in game theory. These design principles naturally come up as a result of pursuing simple yet efficient algorithms: \textit{single-loop} updates preclude sophisticated interleaving of rounds between agents; \textit{symmetric} updates ensure no agent will compromise its rewards in the learning process, which can be otherwise exploited by a faster-updating opponent; in addition, asymmetric updates typically lead to one-sided convergence, i.e., only one of the agents is guaranteed to converge to the minimax equilibrium in a non-asymptotic manner, which is less desirable; moreover, \textit{last-iterate convergence} guarantee absolves the need for agents to switch between learning and deployment; last but not least, it is desirable to converge as fast as possible, where the iteration complexities are {\em non-asymptotic} with clear dependence on salient problem parameters.

Substantial algorithmic developments have been made for finding equilibria in two-player zero-sum Markov games, where Dynamical Programming (DP) techniques have long been used as a fundamental building block, leading to prototypical iterative schemes such as Value Iteration (VI) \citep{shapley1953stochastic} and Policy Iteration (PI) \citep{van1978discounted,patek1999stochastic}. Different from their single-agent counterparts, these methods require solving a two-player zero-sum matrix game for every state per iteration. A considerable number of recent works \citep{zhao2021provably,alacaoglu2022natural,cen2021fast,chen2021sample} are based on these DP iterations, by plugging in various (gradient-based) solvers of two-player zero-sum matrix games. However, these methods are inherently nested-loop, which are less convenient to implement. In addition, PI-based methods are asymmetric and come with only one-sided convergence guarantees \citep{patek1999stochastic,zhao2021provably,alacaoglu2022natural}.

Going beyond nested-loop algorithms, single-loop policy gradient methods have been proposed recently for solving two-player zero-sum Markov games. Here, we are interested in finding an $\epsilon$-optimal NE or QRE in terms of the duality gap, i.e. the difference in the value functions when either of the agents deviates from the solution policy. 
\begin{itemize}
\item For the infinite-horizon discounted setting, \citet{daskalakis2020independent} demonstrated that the independent policy gradient method, with direct parameterization and asymmetric learning rates, finds an $\epsilon$-optimal NE within a polynomial number of iterations. \citet{zeng2022regularized} improved over this rate using an entropy-regularized policy gradient method with softmax parameterization and asymmetric learning rates. On the other end, \citet{wei2021last} proposed an optimistic gradient descent ascent (OGDA) method \citep{rakhlin2013optimization} with direct parameterization and symmetric learning rates,\footnote{To be precise, \citet{wei2021last} proved the average-iterate convergence of the duality gap, as well as the last-iterate convergence of the policy in terms of the Euclidean distance to the set of NEs, where it is possible to translate the latter last-iterate convergence to the duality gap (see Appendix~\ref{sec:translation}). The resulting iteration complexity, however, is much worse than that of the average-iterate convergence in terms of the duality gap, with a problem-dependent constant that can scale pessimistically with salient problem parameters.} which achieves a last-iterate convergence at a rather pessimistic iteration complexity.

\item For the finite-horizon episodic setting, \citet{zhang2022policy,yang2022oftrl} showed that the weighted average-iterate of the optimistic Follow-The-Regularized-Leader (FTRL) method, when combined with slow critic updates, finds an $\epsilon$-optimal NE in a polynomial number of iterations. 
\end{itemize}

A more complete summary of prior results can be found in Table \ref{table:watermelon} and Table \ref{table:firemelon}. In brief, while there have been encouraging progresses in developing computationally efficient policy gradient methods for solving zero-sum Markov games, achieving fast finite-time last-iterate convergence with single-loop and symmetric update rules remains a challenging goal.

\subsection{Our contributions}

Motivated by the positive role of entropy regularization in enabling faster convergence of policy optimization in single-agent RL \citep{cen2020fast,lan2021policy} and two-player zero-sum games \citep{cen2021fast}, we propose a single-loop policy optimization algorithm for two-player zero-sum Markov games in both the infinite-horizon and finite-horizon settings. The proposed algorithm follows the style of actor-critic \citep{konda2000actor}, with the actor updating the policy via the entropy-regularized optimistic multiplicative weights update (OMWU) method \citep{cen2021fast} and the critic updating the value function on a slower timescale. Both agents execute multiplicative and symmetric policy updates, where the learning rates are carefully selected to ensure a fast last-iterate convergence. In both the infinite-horizon and finite-horizon settings, we prove that the last iterate of the proposed method learns the optimal value function and converges at a linear rate to the unique QRE of the entropy-regularized Markov game, which can be further translated into finding the NE by setting the regularization sufficiently small. 
\begin{itemize}  
	\item For the infinite-horizon discounted setting, the last iterate of our method takes at most 
	$$\widetilde{\cO}\left(\frac{|\cS|}{(1-\gamma)^4\tau}\log\frac{1}{\epsilon}\right)$$ iterations for finding an $\epsilon$-optimal QRE under entropy regularization, where $\widetilde{\cO}(\cdot)$ hides logarithmic dependencies. Here, $|\cS|$ is the size of the state space, $\gamma$ is the discount factor, and $\tau$ is the regularization parameter. Moreover, this implies the last-iterate convergence with an iteration complexity of 
	$$\widetilde{\cO}\left(\frac{|\cS|}{(1-\gamma)^5\epsilon}\right)$$ for finding an $\epsilon$-optimal NE.
	\item For the finite-horizon episodic setting, the last iterate of our method takes at most 
	$$\widetilde{\cO}\left(\frac{H^2}{\tau}\log \frac{1}{\epsilon}\right)$$ 
	iterations for finding an $\epsilon$-optimal QRE under entropy regularization, where $H$ is the horizon length. Similarly, this implies the last-iterate convergence with an iteration complexity of
	$$\widetilde{\cO}\left(\frac{H^3}{\epsilon}\right)$$ 
for finding an $\epsilon$-optimal NE.
\end{itemize}

Detailed comparisons between the proposed method and prior arts are provided in Table \ref{table:watermelon} and Table \ref{table:firemelon}. To the best of our knowledge, this work presents the first method that is simultaneously single-loop, symmetric, and achieves fast finite-time last-iterate convergence in terms of the duality gap in both infinite-horizon and finite-horizon settings. From a technical perspective, the infinite-horizon discounted setting is in particular challenging, where ours is the first single-loop algorithm that guarantees an iteration complexity of $\widetilde{\cO}(1/\epsilon)$ for last-iterate convergence in terms of the duality gap, with clear and improved dependencies on other problem parameters in the meantime. In contrast, several existing works introduce additional problem-dependent constants \citep{daskalakis2020independent,wei2021last,zeng2022regularized} in the iteration complexity, which can scale rather pessimistically---sometimes even exponentially---with problem dimensions \citep{li2021softmax}.

Our technical developments require novel ingredients that deviate from prior tools such as error propagation analysis for Bellman operators \citep{perolat2015approximate,patek1999stochastic} from a dynamic programming perspective, as well as the gradient dominance condition \citep{daskalakis2020independent,zeng2022regularized} from a policy optimization perspective. Importantly, at the core of our analysis lies a carefully-designed one-step error contraction bound for policy learning, together with a set of recursive error bounds for value learning,  all of which tailored to the non-Euclidean OMWU update rules that have not been well studied in the setting of Markov games.

\begin{table}[t]
\centering
\begin{tabular}{cccccc}
\toprule
\makecell{Solution\\type}& Reference  & \makecell{Iteration\\ complexity} & \makecell{Single\\ loop} & \makecell{Symmetric}  & \makecell{Last-iterate\\ convergence} \\
\midrule
\multirow{14}{*}{$\epsilon$-NE} & \makecell{\textbf{PI-based Methods}\\ \textbf{\citet{zhao2021provably}}\\\textbf{\citet{alacaoglu2022natural}}} & $\widetilde{\cO}\Big(\frac{\|1/\rho\|_\infty}{(1-\gamma)^3\epsilon}\Big)^*$ & \xmark  & \xmark   & {\color{blue} \cmark}\\
\cmidrule{2-6}
& \makecell{\textbf{VI-based Methods}\\ \textbf{\citet{cen2021fast}}\\ \textbf{\citet{chen2021sample}}} & $\widetilde{\cO}\Big(\frac{1}{(1-\gamma)^3\epsilon}\Big)$ & \xmark  & {\color{blue}\cmark}   & {\color{blue}\cmark}\\
\cmidrule{2-6}
& \textbf{\citet{daskalakis2020independent}} & \makecell{Polynomial$^*$} & {\color{blue}\cmark} & \xmark &   \xmark\\
\cmidrule{2-6}
& \textbf{\citet{zeng2022regularized}} & $\widetilde{\cO}\Big(\frac{|\cS|^{2}\|1/\rho\|_\infty^{5}}{(1-\gamma)^{14}c^4\epsilon^3}\Big)^*$ & {\color{blue}\cmark} &  \xmark& {\color{blue} \cmark}\\
\cmidrule{2-6}
& \multirow{3}{*}{\textbf{\citet{wei2021last}}} & $\widetilde{\cO} \Big(\frac{|\cS|^3}{(1-\gamma)^{9} \epsilon^2}\Big)$& {\color{blue}\cmark} & {\color{blue}\cmark}   & \xmark  \\
\cmidrule{3-6}
&  & $\widetilde{\cO} \Big(\frac{|\cS|^5(|\cA|+|\cB|)^{1/2}}{(1-\gamma)^{16} c^4 \epsilon^2}\Big)$& {\color{blue}\cmark} & {\color{blue}\cmark}   & {\color{blue}\cmark}  \\
\cmidrule{2-6}
& \textbf{{\color{blue}This Work}} \cellcolor{LightCyan}& $\widetilde{\cO}\Big(\frac{|\cS|}{(1-\gamma)^5\epsilon}\Big)$ \cellcolor{LightCyan} & {\color{blue}\cmark} \cellcolor{LightCyan} & {\color{blue}\cmark} \cellcolor{LightCyan}   & {\color{blue}\cmark} \cellcolor{LightCyan} \\
\midrule
\multirow{6}{*}{$\epsilon$-QRE} & \makecell{\textbf{VI-based Methods}\\ \textbf{\citet{cen2021fast}}} & $\widetilde{\cO}\Big(\frac{1}{(1-\gamma)^3}\log^2\frac{1}{\epsilon}\Big)$ & \xmark  & {\color{blue}\cmark}   & {\color{blue}\cmark}\\
\cmidrule{2-6}
 & \textbf{\citet{zeng2022regularized}} & $\widetilde{\cO}\Big(\frac{|\cS|^2\|1/\rho\|^5_\infty}{(1-\gamma)^{11}c^4\tau^3}\log\frac{1}{\epsilon}\Big)^*$ & {\color{blue}\cmark} & \xmark   & {\color{blue} \cmark}\\
\cmidrule{2-6}
& \textbf{{\color{blue}This Work}} \cellcolor{LightCyan} & $\widetilde{\cO}\Big(\frac{|\cS|}{(1-\gamma)^4\tau}\log\frac{1}{\epsilon}\Big)$ \cellcolor{LightCyan} & {\color{blue}\cmark} \cellcolor{LightCyan} & {\color{blue}\cmark} \cellcolor{LightCyan}   & {\color{blue}\cmark} \cellcolor{LightCyan}\\
\bottomrule
\end{tabular}
\caption{Comparison of policy optimization methods for finding an $\epsilon$-optimal NE (resp. QRE) of two-player zero-sum discounted Markov games in terms of the duality gap, i.e., a policy pair $(\mu, \nu)$ satisfying $\max_{\mu', \nu'}(V^{\mu', \nu}(\rho) - V^{\mu, \nu'}(\rho)) \le \epsilon$ (resp. $\max_{\mu', \nu'}(V_\tau^{\mu', \nu}(\rho) - V_\tau^{\mu, \nu'}(\rho)) \le \epsilon$). Note that $*$ implies one-sided convergence, i.e., only one of the agents is guaranteed to achieve finite-time convergence to the equilibrium. Here, $c>0$ refers to some problem-dependent constant. For simplicity and a fair comparison, we replace various notions of concentrability coefficient and distribution mismatch coefficient with a crude upper bound $\|1/\rho\|_\infty$, where $\rho$ is the initial state distribution. }
\label{table:watermelon}
\end{table}

\begin{table}[ht]
\centering
\begin{tabular}{cccccc}
\toprule
\makecell{Solution\\type}& Reference  & \makecell{Iteration\\ complexity} & \makecell{Single\\ loop} & \makecell{Symmetric}   & \makecell{Last-iterate\\ convergence} \\\midrule
\multirow{8}{*}{$\epsilon$-NE} & \makecell{\textbf{\citet{zhang2022policy}}\\\textbf{OFTRL}} & $\widetilde{\cO}\big(\frac{H^{28/5}}{\epsilon^{6/5}}\big)$ & {\color{blue}\cmark}& {\color{blue}\cmark} & \xmark\\
\cmidrule{2-6}
& \makecell{\textbf{\citet{zhang2022policy}}\\\textbf{modified OFTRL}} & $\widetilde{\cO}\big(\frac{H^4}{\epsilon}\big)$ & {\color{blue}\cmark} & {\color{blue}\cmark}& \xmark\\
\cmidrule{2-6}
& \makecell{\textbf{\citet{yang2022oftrl}}\\\textbf{OFTRL}} & $\widetilde{\cO}\big(\frac{H^5}{\epsilon}\big)$ & {\color{blue}\cmark} & {\color{blue}\cmark}& \xmark\\
\cmidrule{2-6}
& {\textbf{{\color{blue}This Work}}} \cellcolor{LightCyan} & $\widetilde{\cO}\big(\frac{H^3}{\epsilon}\big)$ \cellcolor{LightCyan} & {\color{blue}\cmark} \cellcolor{LightCyan} & {\color{blue}\cmark} \cellcolor{LightCyan} &  {\color{blue}\cmark} \cellcolor{LightCyan}\\
\midrule
\multirow{1}{*}{$\epsilon$-QRE} & {\textbf{{\color{blue}This Work}}} \cellcolor{LightCyan} & $\widetilde{\cO}\big(\frac{H^2}{\tau}\log \frac{1}{\epsilon}\big)$ \cellcolor{LightCyan} & {\color{blue}\cmark} \cellcolor{LightCyan} & {\color{blue}\cmark} \cellcolor{LightCyan}   & {\color{blue}\cmark} \cellcolor{LightCyan}\\
\bottomrule
\end{tabular}
\caption{Comparison of policy optimization methods for finding an $\epsilon$-optimal NE or QRE of two-player zero-sum episodic Markov games in terms of the duality gap.  }
\label{table:firemelon}
\end{table}

\subsection{Related works}

\paragraph{Learning in two-player zero-sum matrix games.}{\ }\citet{freund1999adaptive} showed that the average iterate of Multiplicative Weight Update (MWU) method converges to NE at a rate of $\cO(1/\sqrt{T})$, which in principle holds for many other no-regret algorithms as well. \citet{daskalakis2011near} deployed the excessive gap technique of Nesterov and improved the convergence rate to $\cO(1/T)$, which is achieved later by \citep{rakhlin2013optimization} with a simple modification of MWU method, named Optimistic Mirror Descent (OMD) or more commonly, OMWU. 
Moving beyond average-iterate convergence, \citet{bailey2018multiplicative} demonstrated that MWU updates, despite converging in an ergodic manner, diverge from the equilibrium. \citet{daskalakis2018last,wei2021last} explored the last-iterate convergence guarantee of OMWU, assuming uniqueness of NE. \citet{cen2021fast} established linear last-iterate convergence of entropy-regularized OMWU without uniqueness assumption. \citet{sokota2022unified} showed that optimistic update is not necessary for achieving linear last-iterate convergence in the presence of regularization, albeit with a more strict restriction on the step size.

\paragraph{Learning in two-player zero-sum Markov games.} In addition to the aforementioned works on policy optimization methods (policy-based methods) for two-player zero-sum Markov games (cf.~Table \ref{table:watermelon} and Table \ref{table:firemelon}), a growing body of works have developed model-based methods \citep{liu2021sharp,zhang2020model,li2022minimax} and value-based methods \citep{bai2020provable,bai2020near,chen2021almost,jin2021v,sayin2021decentralized,xie2020learning}, with a primary focus on learning NE in a sample-efficient manner. Our work, together with prior literatures on policy optimization, focuses instead on learning NE in a computation-efficient manner assuming full-information. 

\paragraph{Entropy regularization in RL and games.} Entropy regularization is a popular algorithmic idea in RL \citep{williams1991function} that promotes exploration of the policy. A recent line of works \citep{mei2020global,cen2020fast,lan2021policy,zhan2021policy} demonstrated that incorporating entropy regularization provably accelerates policy optimization in single-agent MDPs by enabling fast linear convergence. While the positive role of entropy regularization is also verified in various game-theoretic settings, e.g.,  two-player zero-sum matrix games \citep{cen2021fast}, zero-sum polymatrix games \citep{leonardos2021exploration}, and potential games \citep{cen2022independent}, it remains highly unexplored the interplay between entropy regularization and policy optimization in Markov games with only a few exceptions \citep{zeng2022regularized}.

\subsection{Notations}
We denote the probability simplex over a set $\cA$ by $\Delta(\cA)$. We use bracket with subscript to index the entries of a vector or matrix, e.g., $[x]_{a}$ for $a$-th element of a vector $x$, or simply $x(a)$ when it is clear from the context. Given two distributions $x, y \in \Delta(\cA)$, the Kullback-Leibler (KL) divergence from $y$ to $x$ is denoted by $\KLbig{x}{y} = \sum_{a\in\cA}x(a)(\log x(a) - \log y(a))$. Finally, we denote by $\norm{A}_\infty$ the maximum entrywise absolute value of a matrix $A$, i.e., $\norm{A}_\infty = \max_{i,j}|A_{i,j}|$.

%

\section{Algorithm and theory: the infinite-horizon setting}
\label{sec:infinite_horizon}

\subsection{Problem formulation}
\label{sec:prob_formulation}

\paragraph{Two-player zero-sum discounted Markov game.} A two-player zero-sum discounted Markov game is defined by a tuple $\mathcal{M} = (\cS, \cA, \cB, P, r, \gamma)$, with finite state space $\cS$, finite action spaces of the two players $\cA$ and $\cB$, reward function $r: \cS \times \cA \times \cB \to [0,1]$, transition probability kernel $P:\cS\times\cA\times\cB \to \Delta(\cS)$ and discount factor $0 \le \gamma < 1$. The action selection rule of the max player (resp. the min player) is represented by $\mu: \cS \to \Delta(\cA)$ (resp. $\nu: \cS \to \Delta(\cB)$), where the probability of selecting action $a \in \cA$ (resp. $b\in\cB$) in state $s\in\cS$ is specified by $\mu(a|s)$ (resp. $\nu(b|s)$). The probability of transitioning from state $s$ to a new state $s'$ upon selecting the action pair $(a,b) \in \cA, \cB$ is given by $P(s'|s,a,b)$. 

\paragraph{Value function and Q-function.}
For a given policy pair $\mu, \nu$, the state value of $s\in\cS$ is evaluated by the expected discounted sum of rewards with initial state $s_0 = s$:
\begin{equation}
	\forall s\in\cS: \qquad V^{\mu, \nu}(s) = \ex{}{\sum_{t=0}^\infty \gamma^t r(s_t, a_t, b_t) \big| s_0 = s},
	\label{eq:V_def}
\end{equation}
the quantity the max player seeks to maximize while the min player seeks to minimize.
Here, the trajectory $(s_0, a_0, b_0, s_1, \cdots)$ is generated according to $a_t \sim \mu(\cdot|s_t)$, $b_t \sim \nu(\cdot|s_t)$ and $s_{t+1}\sim P(\cdot|s_t, a_t, b_t)$. Similarly, the $Q$-function $Q^{\mu, \nu}(s,a,b)$ evaluates the expected discounted cumulative reward with initial state $s$ and initial action pair $(a, b)$:
\begin{equation}
	\forall (s,a,b)\in\cS\times\cA\times\cB: \qquad Q^{\mu, \nu}(s,a,b) = \ex{}{\sum_{t=0}^\infty \gamma^t r(s_t, a_t, b_t) \big| s_0 = s, a_0 = a, b_0 = b}.
\end{equation}
For notation simplicity, we denote by $Q^{\mu, \nu}(s)\in\mathbb{R}^{|\cA|\times|\cB|}$ the matrix $[Q^{\mu,\nu}(s,a,b)]_{(a,b)\in\cA\times\cB}$, so that 
\[
	\forall s\in\cS: \qquad V^{\mu, \nu}(s) = \mu(s)^\top Q^{\mu, \nu}(s) \nu(s).
\]
\citet{shapley1953stochastic} proved the existence of a policy pair $(\mu^{\star}, \nu^{\star})$ that solves the min-max problem 
$$\max_{\mu} \min_{\nu} V^{\mu, \nu}(s)$$ 
for all $s\in\cS$ simultaneously, and that the mini-max value is unique. A set of such optimal policy pair $(\mu^{\star}, \nu^{\star})$ is called the Nash equilibrium (NE) to the Markov game.

\paragraph{Entropy regularized two-player zero-sum Markov game.} Entropy regularization is shown to provably accelerate convergence in single-agent RL \citep{geist2019theory, mei2020global, cen2020fast} and facilitate the analysis in two-player zero-sum matrix games \citep{cen2021fast} as well as Markov games \citep{cen2021fast,zeng2022regularized}. The entropy-regularized value function $V_\tau^{\mu, \nu}(s)$ is defined as
\begin{equation}
	\forall s\in\cS: \qquad V_\tau^{\mu, \nu}(s) = \ex{}{\sum_{t=0}^\infty \gamma^t {\Big(}r(s_t, a_t, b_t) - \tau \log \mu(a_t|s_t) + \tau \log \nu(b_t|s_t) {\Big)} \Big|s_0 = s},
	\label{eq:V_tau_def}
\end{equation}
where $\tau \ge 0$ is the regularization parameter. Similarly, the regularized $Q$-function $Q_\tau^{\mu, \nu}$ is given by
\begin{equation}
	\forall (s,a,b)\in\cS\times\cA\times\cB: \qquad Q_\tau^{\mu, \nu}(s) = r(s,a,b) + \gamma\ex{s'\sim P(\cdot|s,a,b)}{V_\tau^{\mu,\nu}(s')}.	\label{eq:Q_tau_def}
\end{equation}
It is known that \citep{cen2021fast} there exists a unique pair of policy $(\best{\mu}, \best{\nu})$ that solves the min-max entropy-regularized problem 
\begin{subequations}
\begin{equation}
	\max_{\mu}\min_{\nu} \; V_\tau^{\mu, \nu} (s),
	\label{eq:min_max_markov}
\end{equation}
or equivalently
\begin{equation}
	\max_{\mu}\min_{\nu} \; \mu(s)^\top Q_\tau^{\mu, \nu} (s) \nu(s) + \tau \ent{\mu(s)} - \tau \ent{\nu(s)}
	\label{eq:min_max_Q}
\end{equation}
\end{subequations}
for all $s\in \cS$, and we call $(\best{\mu}, \best{\nu})$ the quantal response equilibrium (QRE) \citep{mckelvey1995quantal} to the entropy-regularized Markov game. We denote the associated regularized value function and Q-function by
\[
V_\tau^\star (s) = V_\tau^{\best{\mu},\best{\nu}}(s)\quad \text{and}\quad Q_\tau^\star (s,a,b) = Q_\tau^{\best{\mu},\best{\nu}}(s,a,b).
\] 

\paragraph{Goal.} We seek to find an $\epsilon$-optimal QRE or $\epsilon$-QRE (resp. $\epsilon$-optimal NE or $\epsilon$-NE) $\zeta = (\mu, \nu)$ which satisfies 
\begin{equation}
\max_{s\in\cS, \mu', \nu'}\Big(V_\tau^{\mu',\nu}(s) - V_\tau^{\mu,\nu'}(s)\Big) \le \epsilon
\label{eq:eps_QRE_max_s}
\end{equation}
(resp. $\max_{s\in\cS, \mu', \nu'}\Big(V^{\mu',\nu}(s) - V^{\mu,\nu'}(s)\Big) \le \epsilon$) in a computationally efficient manner. In truth, the solution concept of $\epsilon$-QRE provides an approximation of $\epsilon$-NE with appropriate choice of the regularization parameter $\tau$. Basic calculations tell us that	
\begin{align*}
V^{\mu',\nu}(s) - V^{\mu,\nu'}(s) &= \big(V_\tau^{\mu',\nu}(s) - V_\tau^{\mu,\nu'}(s)\big) + \big(V^{\mu', \nu}(s) - V_\tau^{\mu', \nu }(s)\big) - \big(V^{\mu,\nu'}(s) - V_\tau^{\mu,\nu'}(s)\big)\\
		&\le V_\tau^{\mu',\nu}(s) - V_\tau^{\mu,\nu'}(s) + \frac{\tau(\log|\cA| + \log|\cB|)}{1-\gamma},
	\end{align*}
which guarantees that an $\epsilon/2$-QRE is an $\epsilon$-NE as long as $\tau \leq  \frac{(1-\gamma)\epsilon}{2(\log|\cA| + \log|\cB|) }$. For technical convenience, we assume 
\begin{equation}\label{eq:pomelo}
\tau \le \frac{1}{\max\{1, \log|\cA| + \log|\cB|\}}
\end{equation} 
throughout the paper. In addition, one might instead be interested in the expected (entropy-regularized) value function when the initial state is sampled from a distribution $\rho \in \Delta(\cS)$ over $\cS$, which are given by
$$ V_{\tau}^{\mu,\nu}(\rho) := \exlim{s\sim \rho}{V_{\tau}^{\mu,\nu}(s)}, \qquad \mbox{and} \qquad V^{\mu,\nu}(\rho) := \exlim{s\sim \rho}{V^{\mu,\nu}(s)}. $$
The $\epsilon$-QRE/NE can be defined analogously, which facilitates comparisons to a number of related works.


\paragraph{Additional notation.} For notation convenience, we denote by $\zeta$ the concatenation of a policy pair $\mu$ and $\nu$, i.e., $\zeta = (\mu,\nu)$. The QRE to the regularized problem is denoted by $\best{\zeta} = (\best{\mu}, \best{\nu})$. We use shorthand notation $\mu(s)$ and $\nu(s)$ to denote $\mu(\cdot|s)$ and $\nu(\cdot|s)$. In addition, we write $\KL{\mu(s)}{\mu'(s)}$ and $\KL{\nu(s)}{\nu'(s)}$ as $\KLs{\mu}{\mu'}$ and $\KLs{\nu}{\nu'}$, and let
\begin{equation*}
    \KLs{\zeta}{\zeta'} = \KLs{\mu}{\mu'} + \KLs{\nu}{\nu'}.
\end{equation*}

\subsection{Single-loop algorithm design}
\label{sec:result_discounted}

In this section, we propose a single-loop policy optimization algorithm for finding the QRE of the entropy-regularized Markov game, which is generalized from the entropy-regularized OMWU method \citep{cen2021fast} for solving entropy-regularized matrix games, with a careful orchestrating of the policy update and the value update.

\paragraph{Review: entropy-regularized OMWU for two-player zero-sum matrix games.} We briefly review the algorithm design of entropy-regularized OMWU method for two-player zero-sum matrix game \citep{cen2021fast}, which our method builds upon. The problem of interest can be described as
\begin{equation}
	\max_{\mu\in\Delta(\cA)} \min_{\nu\in\Delta(\cB)} \mu^\top A \nu + \tau \cH(\mu) - \tau \cH(\nu),
	\label{eq:min_max_matrix}
\end{equation}
where $A \in \mathbb{R}^{|\cA| \times |\cB|}$ is the payoff matrix of the game. The update rule of entropy-regularized OMWU with learning rate $\eta > 0$ is defined as follows: $\forall a\in \cA, b\in\cB$,
\begin{subequations}
\begin{align}
\begin{cases}
    \mut(a) \propto \mutm(a)^{1-\eta\tau} \exp(\eta [A\bnut]_a)\\
    \nut(b) \propto \nutm(b)^{1-\eta\tau} \exp(-\eta [A^\top\bmut]_b)\\
\end{cases}, \\
\begin{cases}
	\bmutp(a) \propto \mut(a)^{1-\eta\tau} \exp(\eta[A\bnut]_a)\\
	\bnutp(b) \propto \nut(b)^{1-\eta\tau} \exp(-\eta[A^\top\bmut]_b)
\end{cases}.
\end{align}
\end{subequations}
We remark that the update rule can be alternatively motivated from the perspective of natural policy gradient \citep{kakade2002natural,cen2020fast} or mirror descent \citep{lan2021policy,zhan2021policy} with optimistic updates. In particular, the midpoint $(\bmutp, \bnutp)$ serves as a prediction of $(\mutp,\nutp)$ by running one step of mirror descent. \citet{cen2021fast} established that the last iterate of entropy-regularized OMWU converges to the QRE of the matrix game \eqref{eq:min_max_matrix} at a linear rate $(1-\eta\tau)^t$, as long as the step size $\eta$ is no larger than $\min\left\{ \frac{1}{2\|A\|_\infty + 2\tau}, \frac{1}{4\|A\|_\infty}\right\}$.

\paragraph{Single-loop algorithm for two-player zero-sum Markov games.} In view of the similarity in the problem formulations of \eqref{eq:min_max_Q} and \eqref{eq:min_max_matrix}, it is tempting to apply the aforementioned method to the Markov game in a state-wise manner, where the $Q$-function assumes the role of the payoff matrix. It is worth noting, however, that $Q$-function depends on the policy pair $\zeta = (\mu, \nu)$ and is hence changing concurrently with the update of the policy pair. We take inspiration from \citet{wei2021last} and equip the entropy-regularized OMWU method with the following update rule that iteratively approximates the value function in an actor-critic fashion:
\begin{equation*}
		\Q{t+1}(s,a,b) = r(s,a,b) + \gamma \ex{s' \sim P(\cdot|s,a,b)}{\V{t}(s')},
\end{equation*}
where $\V{t+1}$ is updated as a convex combination of the previous $\V{t}$ and the regularized game value induced by $\Q{t+1}$ as well as the policy pair $\bztp = (\bmutp, \bnutp)$:
\begin{equation}
\begin{aligned}
	\V{t+1}(s) &= (1-\alpha_{t+1})\V{t}(s) \\
	&\hspace{-2ex}+\alpha_{t+1}\big[ \bmutp(s)^\top \Q{t+1}(s)\bnutp(s) + \tau \ent{\bmutp(s)} - \tau \ent{\bnutp(s)} \big].
\end{aligned}
\end{equation}
The update of $V$ becomes more conservative with a smaller learning rate $\alpha_t$, hence stabilizing the update of policies. However, setting $\alpha_t$ too small slows down the convergence of $V$ to $\best{V}$. A key novelty---suggested by our analysis---is the choice of the constant learning rates $\alpha:= \alpha_t = \eta\tau$ which updates at a slower timescale than the policy due to $\tau<1$. This is in sharp contrast to the vanishing sequence $\alpha_t = \frac{2/(1-\gamma) + 1}{2/(1-\gamma) + t}$ adopted in \citet{wei2021last}, which is essential in their analysis but inevitably leads to a much slower convergence. We summarize the detailed procedure in Algorithm \ref{alg:kaoya}. 
Last but not least, it is worth noting that the proposed method access
the reward via ``first-order information'', i.e., either agent can only update its policy with the marginalized value function $Q(s)\nu(s)$ or $Q(s)^\top \mu(s)$. Update rules of this kind are instrumental in breaking the curse of multi-agents in the sample complexity when adopting sample-based estimates in \eqref{eq:value_update_kaoya}, as we only need to estimate the marginalized Q-function rather than its full form \citep{li2022minimax,chen2021sample}.


\begin{algorithm}[t]
\label{alg:kaoya}
\caption{Entropy-regularized OMWU for Discounted Two-player Zero-sum Markov Game}
\textbf{Input:} Regularization parameter $\tau > 0$, learning rate for policy update $\eta > 0$, learning rate for value update $\{\alpha_t\}_{t=1}^{\infty}$.\\
\textbf{Initialization:} Set $\mu^{(0)}, \bar{\mu}^{(0)}$, $\nu^{(0)}$ and $\bar{\nu}^{(0)}$ as uniform policies; and set  
\begin{equation*}
\Q{0} = 0,\quad \V{0} = \tau(\log |\cA| - \log |\cB|).
\end{equation*}

\SetKwProg{ForP}{for}{ do in parallel}{end}
\For{$t = 0,1,\cdots$}{
\ForP{all $s\in \cS$}{

\begin{subequations} \label{eq:update_whole}
When $t \ge 1$, update policy pair $\zt(s)$ as:
\begin{equation}
\begin{cases}
    \mut(a|s) \propto \mutm(a|s)^{1-\eta\tau} \exp(\eta [\Q{t}(s)\bnut(s)]_a)\\
    \nut(b|s) \propto \nutm(b|s)^{1-\eta\tau} \exp(-\eta [\Q{t}(s)^\top\bmut(s)]_b)\\
\end{cases}.
\label{eq:update}
\end{equation}

Update policy pair $\bztp(s)$ as:
\begin{equation}
\begin{cases}
	\bmutp(a|s) \propto \mut(a|s)^{1-\eta\tau} \exp(\eta[\Q{t}(s)\bnut(s)]_a)\\
	\bnutp(b|s) \propto \nut(b|s)^{1-\eta\tau} \exp(-\eta[\Q{t}(s)^\top\bmut(s)]_b)
\end{cases}.
\label{eq:update_bar}
\end{equation}
\end{subequations}

Update $\Q{t+1}(s)$ and $\V{t+1}(s)$ as
\begin{equation} \label{eq:value_update_kaoya}
	\begin{cases}
		\Q{t+1}(s,a,b) &= r(s,a,b) + \gamma \ex{s' \sim P(\cdot|s,a,b)}{\V{t}(s')}\\
		\V{t+1}(s)\quad &= (1-\alpha_{t+1})\V{t}(s) \\
		&\hspace{-4ex}+\alpha_{t+1}\big[ \bmutp(s)^\top \Q{t+1}(s)\bnutp(s) + \tau \ent{\bmutp(s)} - \tau \ent{\bnutp(s)} \big]
	\end{cases}.
\end{equation}
}
}
\end{algorithm}


\subsection{Theoretical guarantees}
Below we present our main results concerning the last-iterate convergence of Algorithm \ref{alg:kaoya} for solving entropy-regularized two-player zero-sum Markov games in the infinite-horizon discounted setting. The proof is postponed to Appendix \ref{sec:analysis_discounted}.

\begin{theorem}
\label{thm:Prospero_A}
Setting $0 < \eta \le \frac{(1-\gamma)^3}{32000|\cS|}$ and $\alpha_t = \eta\tau$, it holds for all $t \ge 0$ that
\begin{subequations}
\begin{multline} \label{eq:Prospero_A}
	\max\left\{\frac{1}{|\cS|}\sum_{s\in\cS}\KLs{\best{\zeta}}{\zt},\,  \frac{1}{2|\cS|}\sum_{s\in\cS}\KLs{\best{\zeta}}{\bzt}, \,  \frac{3\eta}{|\cS|}\sum_{s\in\cS}\normbig{\Q{t}(s) - \best{Q}(s)}_\infty\right\} \\
  \le \frac{3000}{(1-\gamma)^2\tau}\Big(1-\frac{(1-\gamma)\eta\tau}{4}\Big)^{t};
\end{multline}
and 
\begin{align}
	&\max_{s\in \cS, \mu, \nu}\Big(V_\tau^{\mu,\bar{\nu}^{(t)}}(s) - V_\tau^{\bmut,\nu}(s)\Big)\le \frac{6000|\cS|}{(1-\gamma)^3\tau}\max\left\{\frac{8}{(1-\gamma)^2\tau}, \frac{1}{\eta}\right\}\Big(1-\frac{(1-\gamma)\eta\tau}{4}\Big)^{t}.
	\label{eq:Miranda_A}
\end{align}
\end{subequations}
\end{theorem}

Theorem \ref{thm:Prospero_A} demonstrates that as long as the learning rate $\eta$ is small enough, the last iterate of Algorithm \ref{alg:kaoya} converges at a linear rate for the entropy-regularized Markov game. Compared with prior literatures investigating on policy optimization, our analysis focuses on the last-iterate convergence of non-Euclidean updates in the presence of entropy regularization, which appears to be the first of its kind. Several remarks are in order, with detailed comparisons in Table~\ref{table:watermelon}.

\begin{itemize}
	\item \textbf{Linear convergence to the QRE.} Theorem \ref{thm:Prospero_A}  demonstrates that the last iterate of Algorithm \ref{alg:kaoya} takes at most $\widetilde{\cO}\left(\frac{1}{(1-\gamma)\eta\tau}\log\frac{1}{\epsilon}\right)$ iterations to yield an $\epsilon$-optimal policy in terms of the KL divergence to the QRE $\max\limits_{s\in\cS}\KLs{\best{\zeta}}{\bzt} \le \epsilon$, the entrywise error of the regularized Q-function $\normbig{\Q{t} - \best{Q}}_\infty \le \epsilon$, as well as the duality gap $\max\limits_{s\in\cS, \mu, \nu}\big(V_\tau^{\mu,\bar{\nu}^{(t)}}(s) - V_\tau^{\bmut,\nu}(s)\big) \le \epsilon$ at once. Minimizing the bound over the learning rate $\eta$, the proposed method is guaranteed to find an $\epsilon$-QRE within 
	\[
	\widetilde{\cO}\left(\frac{|\cS|}{(1-\gamma)^4\tau}\log\frac{1}{\epsilon}\right)
	\]
	iterations, which significantly improves upon the one-side convergence rate of \citet{zeng2022regularized}.
	\item \textbf{Last-iterate convergence to $\epsilon$-optimal NE.} By setting $\tau =\frac{(1-\gamma)\epsilon}{2(\log|\cA| + \log|\cB|)}$, this immediately leads to provable last-iterate convergence to an $\epsilon$-NE, with an iteration complexity of
	\[
	\widetilde{\cO}\left(\frac{|\cS|}{(1-\gamma)^5\epsilon}\right),
	\]
	which again outperforms the convergence rate of \citet{wei2021last}. 
\end{itemize}

\begin{remark}
The learning rate $\eta$ is constrained to be inverse proportional to $|\cS|$, which is for the worst case and can be potentially loosened for problems with a small concentrability coefficient. We refer interested readers to Appendix \ref{sec:analysis_discounted} for details.
\end{remark}

\section{Algorithm and theory: the finite-horizon setting}

\label{sec:result_episodic}


\paragraph{Episodic two-player zero-sum Markov game.}

An episodic two-player zero-sum Markov game is defined by a tuple $\{\cS, \cA, \cB, H, \{P_h\}_{h=1}^H, \{r_h\}_{h=1}^H\}$, with $\cS$ being a finite state space, $\cA$ and $\cB$ denoting finite action spaces of the two players, and $H > 0$ the horizon length. Every step $h \in [H]$ admits a transition probability kernel $P_h: \cS \times \cA \to \Delta(\cS)$ and reward function $r_h: \cS \times \cA \times \cB \to [0,1]$. Furthermore, $\mu = \{\mu_h\}_{h=1}^H$ and $\{\nu_h\}_{h=1}^H$ denote the policies of the two players, where the probability of the max player choosing $a\in\cA$ (resp. the min player choosing $b\in\cB$) at time $h$ is specified by $\mu_h(a|s)$ (resp. $\nu_h(a|s)$).

\paragraph{Entropy regularized value functions.}
 The value function and Q-function characterize the expected cumulative reward starting from step $h$ by following the policy pair $\mu, \nu$. For conciseness, we only present the definition of entropy-regularized value functions below and remark that the their un-regularized counterparts $V_h^{\mu, \nu}$ and $Q_{h}^{\mu,\nu}$ can be obtained by setting $\tau = 0$. We have
\[
    \begin{aligned}
    {V_{h, \tau}^{\mu, \nu}}(s) &= \ex{}{\sum_{h'=h}^{H}\brk{r_{h'}(s_{h'}, a_{h'}, b_{h'}) - \tau \log \mu_{h'}(a_{h'}|s_{h'}) + \tau \log \nu_{h'}(b_{h'}|s_{h'})}\;\Big|\;s_h = s};\\
    {Q_{h, \tau}^{\mu, \nu}}(s,a,b) &= r_h(s,a,b) + \ex{s'\sim P_h(\cdot|s,a,b)}{{V_{h+1, \tau}^{\mu, \nu}}(s')}.
    \end{aligned}
\]
The solution concept of NE and QRE are defined in a similar manner by focusing on the episodic versions of value functions. 
We again denote the unique QRE by $\zeta_\tau^{\star} =(\mu_\tau^\star, \nu_\tau^\star)$.
\paragraph{Proposed method and convergence guarantee}
It is straightforward to adapt Algorithm \ref{alg:kaoya} to the episodic setting with minimal modifications, with detailed procedure showcased in Algorithm \ref{alg:ramen}.
\begin{algorithm}[t]
\label{alg:ramen}
\caption{Entropy-regularized OMWU for Episodic Two-player Zero-sum Markov Game}
\textbf{Input:} Regularization parameter $\tau > 0$, learning rate for policy update $\eta > 0$, learning rate for value update $\{\alpha_t\}_{t=1}^{\infty}$.\\
\textbf{Initialization:} Set $\mu^{(0)}, \bar{\mu}^{(0)}$, $\nu^{(0)}$ and $\bar{\nu}^{(0)}$ as uniform policies; set
\begin{equation*}
\Q{0} = 0,\quad \V{0} = \tau(\log |\cA| - \log |\cB|).
\end{equation*}
\SetKwProg{ForP}{for}{ do in parallel}{end}
\For{$t = 0,1,\cdots$}{
\ForP{all $h \in [H]$, $s\in \cS$}{
\begin{subequations}
When $t \ge 1$, update policy pair $\zt_h(s)$ as:
\begin{equation}
\begin{cases}
    \mut_h(a|s) \propto \mutm_h(a|s)^{1-\eta\tau} \exp(\eta [\Q{t}_h(s)\bnut_h(s)]_a)\\
    \nut_h(b|s) \propto \nutm_h(b|s)^{1-\eta\tau} \exp(-\eta [\Q{t}_h(s)^\top\bmut_h(s)]_b)\\
\end{cases}.
\label{eq:update_epi}
\end{equation}

Update policy pair $\bztp_h(s)$ as:
\begin{equation}
\begin{cases}
	\bmutp_h(a|s) \propto \mut_h(a|s)^{1-\eta\tau} \exp(\eta[\Q{t}_h(s)\bnut_h(s)]_a)\\
	\bnutp_h(b|s) \propto \nut_h(b|s)^{1-\eta\tau} \exp(-\eta[\Q{t}_h(s)^\top\bmut_h(s)]_b)
\end{cases}.
\label{eq:update_bar_epi}
\end{equation}
\end{subequations}

Update $\Q{t+1}_h(s)$ and $\V{t+1}_h(s)$ as
\begin{equation}
	\begin{cases}
		\Q{t+1}_h(s,a,b) &= r_h(s,a,b) + \gamma \ex{s' \sim P_h(\cdot|s,a,b)}{\V{t}_{h+1}(s')}\\
		\V{t+1}_{h}(s)\quad &= (1-\alpha_{t+1})\V{t}_{h}(s) \\
		&\hspace{-4ex}+\alpha_{t+1}\Big[ \bmutp_{h}(s)^\top \Q{t+1}_{h}(s)\bnutp_{h}(s) + \tau \ent{\bmutp_{h}(s)} - \tau \ent{\bnutp_{h}(s)} \Big]
	\end{cases}.
\end{equation}
}
}
\end{algorithm}
The analysis, which substantially deviates from the discounted setting, exploits the structure of finite-horizon MDP and time-inhomogeneous policies, enabling a much larger range of learning rates as showed in the following theorem.
\begin{theorem}
\label{thm:Antonio}
	Setting $0 < \eta \le \frac{1}{8H}$ and $\alpha_t = \eta\tau$, it holds for all $h \in [H]$ and $t \ge T_h :=(H-h)T_{\mathsf{start}}$ with $T_{\mathsf{start}} = \lceil\frac{1}{\eta\tau}\log H\rceil$ that
	\begin{subequations}
	\begin{equation}
		\normbig{\besth{Q} - \Q{t}_{h}}_\infty \le (1-\eta\tau)^{t-T_h} t^{H-h};
		\label{eq:Q_conv_epi}
	\end{equation}
	\begin{equation}
		\max_{s\in\cS, \mu, \nu}\big(V_{h,\tau}^{\mu, \bnut}(s) - V_{h,\tau}^{\bmut, \nu}(s)\big) \le 4(1-\eta\tau)^{t-T_h}\max\left\{\frac{8H^2}{\tau}, \frac{1}{\eta}\right\} \Big(\frac{8H}{\tau} + 6\eta t^{H-h+1}\Big).
		\label{eq:gap_conv_epi}
	\end{equation}
	\end{subequations}
\end{theorem}

Theorem \ref{thm:Antonio} implies that the last iterate of Algorithm~\ref{alg:ramen} takes no more than $\widetilde{\cO}\big(H T_{\mathsf{start}} + \frac{H}{\eta\tau}\log\frac{1}{\epsilon}\big) = \widetilde{\cO}\big(\frac{H}{\eta\tau}\log\frac{1}{\epsilon}\big)$ iterations for finding an $\epsilon$-QRE. Minimizing the bound over the learning rate $\eta$, Algorithm \ref{alg:ramen} is guaranteed to find an $\epsilon$-QRE in
\[
\widetilde{\cO}\left(\frac{H^2}{\tau}\log\frac{1}{\epsilon}\right)
\]
iterations, which translates into an iteration complexity of $\widetilde{\cO}\left(\frac{H^3}{\epsilon}\right)$ for finding an $\epsilon$-NE in terms of the duality gap, i.e., 
$
\max_{s\in\cS, h\in[H], \mu, \nu}\Big(V_{h}^{\mu,\bar{\nu}^{(t)}}(s) - V_{h}^{\bmut,\nu}(s)\Big) \le \epsilon$,
by setting $\tau = \cO\Big(\frac{\epsilon}{H(\log|\cA| + \log|\cB|)}\Big)$.


\section{Discussion}
\label{sec:discussion}
This work develops policy optimization methods for zero-sum Markov games that feature single-loop and symmetric updates with provable last-iterate convergence guarantees. Our approach yields better iteration complexities in both infinite-horizon and finite-horizon settings, by adopting entropy regularization and non-Euclidean policy update. Important future directions include investigating whether larger learning rates are possible without knowing problem-dependent information a priori, extending the framework to allow function approximation, and designing sample-efficient implementations of the proposed method.
Last but not least, the introduction of entropy regularization requires each agent to reveal the entropy of their current policy to each other, which prevents the proposed method from being fully \textit{decentralized}. Can we bypass this by dropping the entropy information in value learning? We leave the answers to future work.

\section*{Acknowledgments}

The authors would like to thank Gen Li and Zeyuan Allen-Zhu for valuable discussions.
Part of this work was completed while S. Cen was an intern at Meta AI Research. S.~Cen and Y. Chi are supported in part by the grants ONR N00014-19-1-2404, NSF CCF-1901199, CCF-2106778 and CNS-2148212. S.~Cen is also gratefully supported by Wei Shen and Xuehong Zhang Presidential Fellowship, and Nicholas Minnici Dean's Graduate Fellowship in Electrical and Computer Engineering at Carnegie Mellon University.
S. S. Du acknowledges funding from NSF Awards CCF-2212261, IIS-2143493, DMS-2134106, CCF-2019844 and IIS-2110170.
 
\bibliography{bibfileGame,bibfileRL}

\begin{thebibliography}{51}
\providecommand{\natexlab}[1]{#1}
\providecommand{\url}[1]{\texttt{#1}}
\expandafter\ifx\csname urlstyle\endcsname\relax
  \providecommand{\doi}[1]{doi: #1}\else
  \providecommand{\doi}{doi: \begingroup \urlstyle{rm}\Url}\fi

\bibitem[Agarwal et~al.(2020)Agarwal, Kakade, Lee, and
  Mahajan]{agarwal2019optimality}
A.~Agarwal, S.~M. Kakade, J.~D. Lee, and G.~Mahajan.
\newblock Optimality and approximation with policy gradient methods in {M}arkov
  decision processes.
\newblock In \emph{Conference on Learning Theory}, pages 64--66. PMLR, 2020.

\bibitem[Alacaoglu et~al.(2022)Alacaoglu, Viano, He, and
  Cevher]{alacaoglu2022natural}
A.~Alacaoglu, L.~Viano, N.~He, and V.~Cevher.
\newblock A natural actor-critic framework for zero-sum {M}arkov games.
\newblock In \emph{International Conference on Machine Learning}, pages
  307--366. PMLR, 2022.

\bibitem[Bai and Jin(2020)]{bai2020provable}
Y.~Bai and C.~Jin.
\newblock Provable self-play algorithms for competitive reinforcement learning.
\newblock In \emph{International Conference on Machine Learning}, pages
  551--560. PMLR, 2020.

\bibitem[Bai et~al.(2020)Bai, Jin, and Yu]{bai2020near}
Y.~Bai, C.~Jin, and T.~Yu.
\newblock Near-optimal reinforcement learning with self-play.
\newblock \emph{Advances in neural information processing systems},
  33:\penalty0 2159--2170, 2020.

\bibitem[Bailey and Piliouras(2018)]{bailey2018multiplicative}
J.~P. Bailey and G.~Piliouras.
\newblock Multiplicative weights update in zero-sum games.
\newblock In \emph{Proceedings of the 2018 ACM Conference on Economics and
  Computation}, pages 321--338, 2018.

\bibitem[Bhandari and Russo(2019)]{bhandari2019global}
J.~Bhandari and D.~Russo.
\newblock Global optimality guarantees for policy gradient methods.
\newblock \emph{arXiv preprint arXiv:1906.01786}, 2019.

\bibitem[Bhandari and Russo(2020)]{bhandari2020note}
J.~Bhandari and D.~Russo.
\newblock A note on the linear convergence of policy gradient methods.
\newblock \emph{arXiv preprint arXiv:2007.11120}, 2020.

\bibitem[Cen et~al.(2021{\natexlab{a}})Cen, Cheng, Chen, Wei, and
  Chi]{cen2020fast}
S.~Cen, C.~Cheng, Y.~Chen, Y.~Wei, and Y.~Chi.
\newblock Fast global convergence of natural policy gradient methods with
  entropy regularization.
\newblock \emph{Operations Research}, 2021{\natexlab{a}}.

\bibitem[Cen et~al.(2021{\natexlab{b}})Cen, Wei, and Chi]{cen2021fast}
S.~Cen, Y.~Wei, and Y.~Chi.
\newblock Fast policy extragradient methods for competitive games with entropy
  regularization.
\newblock \emph{Advances in Neural Information Processing Systems},
  34:\penalty0 27952--27964, 2021{\natexlab{b}}.

\bibitem[Cen et~al.(2022)Cen, Chen, and Chi]{cen2022independent}
S.~Cen, F.~Chen, and Y.~Chi.
\newblock Independent natural policy gradient methods for potential games:
  Finite-time global convergence with entropy regularization.
\newblock In \emph{2022 IEEE 61th Conference on Decision and Control (CDC)}.
  IEEE, 2022.

\bibitem[Chen et~al.(2021{\natexlab{a}})Chen, Ma, and Zhou]{chen2021sample}
Z.~Chen, S.~Ma, and Y.~Zhou.
\newblock Sample efficient stochastic policy extragradient algorithm for
  zero-sum markov game.
\newblock In \emph{International Conference on Learning Representations},
  2021{\natexlab{a}}.

\bibitem[Chen et~al.(2021{\natexlab{b}})Chen, Zhou, and Gu]{chen2021almost}
Z.~Chen, D.~Zhou, and Q.~Gu.
\newblock Almost optimal algorithms for two-player {M}arkov games with linear
  function approximation.
\newblock \emph{arXiv preprint arXiv:2102.07404}, 2021{\natexlab{b}}.

\bibitem[Daskalakis and Panageas(2018)]{daskalakis2018last}
C.~Daskalakis and I.~Panageas.
\newblock Last-iterate convergence: Zero-sum games and constrained min-max
  optimization.
\newblock \emph{arXiv preprint arXiv:1807.04252}, 2018.

\bibitem[Daskalakis et~al.(2011)Daskalakis, Deckelbaum, and
  Kim]{daskalakis2011near}
C.~Daskalakis, A.~Deckelbaum, and A.~Kim.
\newblock Near-optimal no-regret algorithms for zero-sum games.
\newblock In \emph{Proceedings of the twenty-second annual ACM-SIAM symposium
  on Discrete Algorithms}, pages 235--254. SIAM, 2011.

\bibitem[Daskalakis et~al.(2020)Daskalakis, Foster, and
  Golowich]{daskalakis2020independent}
C.~Daskalakis, D.~J. Foster, and N.~Golowich.
\newblock Independent policy gradient methods for competitive reinforcement
  learning.
\newblock In \emph{Advances in Neural Information Processing Systems},
  volume~33, pages 5527--5540, 2020.

\bibitem[Freund and Schapire(1999)]{freund1999adaptive}
Y.~Freund and R.~E. Schapire.
\newblock Adaptive game playing using multiplicative weights.
\newblock \emph{Games and Economic Behavior}, 29\penalty0 (1-2):\penalty0
  79--103, 1999.

\bibitem[Geist et~al.(2019)Geist, Scherrer, and Pietquin]{geist2019theory}
M.~Geist, B.~Scherrer, and O.~Pietquin.
\newblock A theory of regularized {M}arkov decision processes.
\newblock In \emph{International Conference on Machine Learning}, pages
  2160--2169, 2019.

\bibitem[Jin et~al.(2021)Jin, Liu, Wang, and Yu]{jin2021v}
C.~Jin, Q.~Liu, Y.~Wang, and T.~Yu.
\newblock V-learning--a simple, efficient, decentralized algorithm for
  multiagent rl.
\newblock \emph{arXiv preprint arXiv:2110.14555}, 2021.

\bibitem[Kakade(2002)]{kakade2002natural}
S.~M. Kakade.
\newblock A natural policy gradient.
\newblock In \emph{Advances in neural information processing systems}, pages
  1531--1538, 2002.

\bibitem[Khodadadian et~al.(2021)Khodadadian, Jhunjhunwala, Varma, and
  Maguluri]{khodadadian2021linear}
S.~Khodadadian, P.~R. Jhunjhunwala, S.~M. Varma, and S.~T. Maguluri.
\newblock On the linear convergence of natural policy gradient algorithm.
\newblock In \emph{2021 60th IEEE Conference on Decision and Control (CDC)},
  pages 3794--3799. IEEE, 2021.

\bibitem[Konda and Tsitsiklis(2000)]{konda2000actor}
V.~R. Konda and J.~N. Tsitsiklis.
\newblock Actor-critic algorithms.
\newblock In \emph{Advances in neural information processing systems}, pages
  1008--1014. Citeseer, 2000.

\bibitem[Lan(2022)]{lan2021policy}
G.~Lan.
\newblock Policy mirror descent for reinforcement learning: Linear convergence,
  new sampling complexity, and generalized problem classes.
\newblock \emph{Mathematical programming}, pages 1--48, 2022.

\bibitem[Leonardos et~al.(2021)Leonardos, Piliouras, and
  Spendlove]{leonardos2021exploration}
S.~Leonardos, G.~Piliouras, and K.~Spendlove.
\newblock Exploration-exploitation in multi-agent competition: convergence with
  bounded rationality.
\newblock \emph{Advances in Neural Information Processing Systems},
  34:\penalty0 26318--26331, 2021.

\bibitem[Li et~al.(2021)Li, Wei, Chi, Gu, and Chen]{li2021softmax}
G.~Li, Y.~Wei, Y.~Chi, Y.~Gu, and Y.~Chen.
\newblock Softmax policy gradient methods can take exponential time to
  converge.
\newblock In \emph{Conference on Learning Theory}, pages 3107--3110. PMLR,
  2021.

\bibitem[Li et~al.(2022)Li, Chi, Wei, and Chen]{li2022minimax}
G.~Li, Y.~Chi, Y.~Wei, and Y.~Chen.
\newblock Minimax-optimal multi-agent {RL} in zero-sum {M}arkov games with a
  generative model.
\newblock \emph{arXiv preprint arXiv:2208.10458}, 2022.

\bibitem[Liu et~al.(2021)Liu, Yu, Bai, and Jin]{liu2021sharp}
Q.~Liu, T.~Yu, Y.~Bai, and C.~Jin.
\newblock A sharp analysis of model-based reinforcement learning with
  self-play.
\newblock In \emph{International Conference on Machine Learning}, pages
  7001--7010. PMLR, 2021.

\bibitem[McKelvey and Palfrey(1995)]{mckelvey1995quantal}
R.~D. McKelvey and T.~R. Palfrey.
\newblock Quantal response equilibria for normal form games.
\newblock \emph{Games and economic behavior}, 10\penalty0 (1):\penalty0 6--38,
  1995.

\bibitem[Mei et~al.(2020)Mei, Xiao, Szepesvari, and Schuurmans]{mei2020global}
J.~Mei, C.~Xiao, C.~Szepesvari, and D.~Schuurmans.
\newblock On the global convergence rates of softmax policy gradient methods.
\newblock In \emph{International Conference on Machine Learning}, pages
  6820--6829. PMLR, 2020.

\bibitem[Patek and Bertsekas(1999)]{patek1999stochastic}
S.~D. Patek and D.~P. Bertsekas.
\newblock Stochastic shortest path games.
\newblock \emph{SIAM Journal on Control and Optimization}, 37\penalty0
  (3):\penalty0 804--824, 1999.

\bibitem[Perolat et~al.(2015)Perolat, Scherrer, Piot, and
  Pietquin]{perolat2015approximate}
J.~Perolat, B.~Scherrer, B.~Piot, and O.~Pietquin.
\newblock Approximate dynamic programming for two-player zero-sum {M}arkov
  games.
\newblock In \emph{International Conference on Machine Learning}, pages
  1321--1329. PMLR, 2015.

\bibitem[Peters and Schaal(2008)]{peters2008natural}
J.~Peters and S.~Schaal.
\newblock Natural actor-critic.
\newblock \emph{Neurocomputing}, 71\penalty0 (7-9):\penalty0 1180--1190, 2008.

\bibitem[Rakhlin and Sridharan(2013)]{rakhlin2013optimization}
A.~Rakhlin and K.~Sridharan.
\newblock Optimization, learning, and games with predictable sequences.
\newblock \emph{arXiv preprint arXiv:1311.1869}, 2013.

\bibitem[Sayin et~al.(2021)Sayin, Zhang, Leslie, Basar, and
  Ozdaglar]{sayin2021decentralized}
M.~Sayin, K.~Zhang, D.~Leslie, T.~Basar, and A.~Ozdaglar.
\newblock Decentralized {Q}-learning in zero-sum {M}arkov games.
\newblock \emph{Advances in Neural Information Processing Systems},
  34:\penalty0 18320--18334, 2021.

\bibitem[Schulman et~al.(2015)Schulman, Levine, Abbeel, Jordan, and
  Moritz]{schulman2015trust}
J.~Schulman, S.~Levine, P.~Abbeel, M.~Jordan, and P.~Moritz.
\newblock Trust region policy optimization.
\newblock In \emph{International conference on machine learning}, pages
  1889--1897, 2015.

\bibitem[Schulman et~al.(2017)Schulman, Wolski, Dhariwal, Radford, and
  Klimov]{schulman2017proximal}
J.~Schulman, F.~Wolski, P.~Dhariwal, A.~Radford, and O.~Klimov.
\newblock Proximal policy optimization algorithms.
\newblock \emph{arXiv preprint arXiv:1707.06347}, 2017.

\bibitem[Shapley(1953)]{shapley1953stochastic}
L.~S. Shapley.
\newblock Stochastic games.
\newblock \emph{Proceedings of the National Academy of Sciences}, 39\penalty0
  (10):\penalty0 1095--1100, 1953.

\bibitem[Silver et~al.(2016)Silver, Huang, Maddison, Guez, Sifre, Van
  Den~Driessche, Schrittwieser, Antonoglou, Panneershelvam, Lanctot,
  et~al.]{silver2016mastering}
D.~Silver, A.~Huang, C.~J. Maddison, A.~Guez, L.~Sifre, G.~Van Den~Driessche,
  J.~Schrittwieser, I.~Antonoglou, V.~Panneershelvam, M.~Lanctot, et~al.
\newblock Mastering the game of {G}o with deep neural networks and tree search.
\newblock \emph{nature}, 529\penalty0 (7587):\penalty0 484--489, 2016.

\bibitem[Sokota et~al.(2022)Sokota, D'Orazio, Kolter, Loizou, Lanctot,
  Mitliagkas, Brown, and Kroer]{sokota2022unified}
S.~Sokota, R.~D'Orazio, J.~Z. Kolter, N.~Loizou, M.~Lanctot, I.~Mitliagkas,
  N.~Brown, and C.~Kroer.
\newblock A unified approach to reinforcement learning, quantal response
  equilibria, and two-player zero-sum games.
\newblock \emph{arXiv preprint arXiv:2206.05825}, 2022.

\bibitem[Sutton et~al.(2000)Sutton, McAllester, Singh, and
  Mansour]{sutton2000policy}
R.~S. Sutton, D.~A. McAllester, S.~P. Singh, and Y.~Mansour.
\newblock Policy gradient methods for reinforcement learning with function
  approximation.
\newblock In \emph{Advances in neural information processing systems}, pages
  1057--1063, 2000.

\bibitem[Van Der~Wal(1978)]{van1978discounted}
J.~Van Der~Wal.
\newblock Discounted markov games: Generalized policy iteration method.
\newblock \emph{Journal of Optimization Theory and Applications}, 25\penalty0
  (1):\penalty0 125--138, 1978.

\bibitem[Wei et~al.(2021)Wei, Lee, Zhang, and Luo]{wei2021last}
C.-Y. Wei, C.-W. Lee, M.~Zhang, and H.~Luo.
\newblock Last-iterate convergence of decentralized optimistic gradient
  descent/ascent in infinite-horizon competitive markov games.
\newblock In \emph{Conference on learning theory}, pages 4259--4299. PMLR,
  2021.

\bibitem[Williams(1992)]{williams1992simple}
R.~J. Williams.
\newblock Simple statistical gradient-following algorithms for connectionist
  reinforcement learning.
\newblock \emph{Machine learning}, 8\penalty0 (3-4):\penalty0 229--256, 1992.

\bibitem[Williams and Peng(1991)]{williams1991function}
R.~J. Williams and J.~Peng.
\newblock Function optimization using connectionist reinforcement learning
  algorithms.
\newblock \emph{Connection Science}, 3\penalty0 (3):\penalty0 241--268, 1991.

\bibitem[Xiao(2022)]{xiao2022convergence}
L.~Xiao.
\newblock On the convergence rates of policy gradient methods.
\newblock \emph{arXiv preprint arXiv:2201.07443}, 2022.

\bibitem[Xie et~al.(2020)Xie, Chen, Wang, and Yang]{xie2020learning}
Q.~Xie, Y.~Chen, Z.~Wang, and Z.~Yang.
\newblock Learning zero-sum simultaneous-move markov games using function
  approximation and correlated equilibrium.
\newblock In \emph{Conference on learning theory}, pages 3674--3682. PMLR,
  2020.

\bibitem[Yang and Ma(2022)]{yang2022oftrl}
Y.~Yang and C.~Ma.
\newblock {$O(T^{-1})$} convergence of optimistic-follow-the-regularized-leader
  in two-player zero-sum markov games.
\newblock \emph{arXiv preprint arXiv:2209.12430}, 2022.

\bibitem[Zeng et~al.(2022)Zeng, Doan, and Romberg]{zeng2022regularized}
S.~Zeng, T.~T. Doan, and J.~Romberg.
\newblock Regularized gradient descent ascent for two-player zero-sum {M}arkov
  games.
\newblock \emph{Advances in Neural Information Processing Systems}, 35, 2022.

\bibitem[Zhan et~al.(2021)Zhan, Cen, Huang, Chen, Lee, and Chi]{zhan2021policy}
W.~Zhan, S.~Cen, B.~Huang, Y.~Chen, J.~D. Lee, and Y.~Chi.
\newblock Policy mirror descent for regularized reinforcement learning: A
  generalized framework with linear convergence.
\newblock \emph{arXiv preprint arXiv:2105.11066}, 2021.

\bibitem[Zhang et~al.(2020)Zhang, Kakade, Basar, and Yang]{zhang2020model}
K.~Zhang, S.~Kakade, T.~Basar, and L.~Yang.
\newblock Model-based multi-agent {RL} in zero-sum {M}arkov games with
  near-optimal sample complexity.
\newblock \emph{Advances in Neural Information Processing Systems}, 33, 2020.

\bibitem[Zhang et~al.(2022)Zhang, Liu, Wang, Xiong, Li, and
  Bai]{zhang2022policy}
R.~Zhang, Q.~Liu, H.~Wang, C.~Xiong, N.~Li, and Y.~Bai.
\newblock Policy optimization for {M}arkov games: Unified framework and faster
  convergence.
\newblock \emph{Advances in Neural Information Processing Systems}, 35, 2022.

\bibitem[Zhao et~al.(2022)Zhao, Tian, Lee, and Du]{zhao2021provably}
Y.~Zhao, Y.~Tian, J.~Lee, and S.~Du.
\newblock Provably efficient policy optimization for two-player zero-sum markov
  games.
\newblock In \emph{International Conference on Artificial Intelligence and
  Statistics}, pages 2736--2761. PMLR, 2022.

\end{thebibliography}
\bibliographystyle{abbrvnat}

\appendix

\section{Analysis for the infinite-horizon setting}
\label{sec:analysis_discounted}

We begin with the definitions of a certain concentrability coefficient, as well as the regularized minimax mismatch coefficient,  which allow us to present general theorems that take into account the problem structure in a more refined manner, from which Theorem \ref{thm:Prospero_A} follow directly.

\begin{definition}
	Given $\rho\in \Delta(\cS)$ with $\rho(s) > 0, \forall s\in S$, the concentrability coefficient $c_\rho(t)$ is defined as
	\begin{equation*}
	c_\rho(t) = \sup_{\substack{x^{(l)}\in \cA^{\cS},1\le l \le t,\\y^{(l)}\in \cB^{\cS},1\le l \le t}}	\Big\|\frac{\rho P_{x^{(1)},y^{(1)}} \cdots P_{x^{(t)},y^{(t)}}}{\rho}\Big\|_\infty,
	\end{equation*}
	where $P_{x^{(l)}, y^{(l)}} \in \mathbb{R}^{|\cS|\times|\cS|}$ is the state transition matrix induced by a pair of deterministic policy $x^{(l)}, y^{(l)}$:
	\[
	[P_{x^{(l)}, y^{(l)}}]_{s,s'} = P(s'|s,x^{(l)}(s),y^{(l)}(s)).
	\]
	Let $\cC_\rho$ be the maximum value of $c_\rho(t)$ over $t\ge 0$:
	\begin{equation*}
		\cC_\rho = \sup_{t\ge0} c_\rho(t).
	\end{equation*}
	In addition, let $\Gamma(\rho)$ be the set of all possible distribution over $\cS$ induced by an initial state distribution $\rho$ and deterministic policy sequences, i.e., 
	\begin{equation*}
		\Gamma(\rho) = \bigcup_{t=0}^\infty \big\{\rho P_{x^{(1)},y^{(1)}} \cdots P_{x^{(t)},y^{(t)}}: x^{(l)}\in \cA^{\cS},y^{(l)}\in \cB^{\cS}, \forall l \in [t]\big\}.
	\end{equation*}
\end{definition}

The following definition of the the regularized minimax mismatch coefficient parallels that of the unregularized one in \citep{daskalakis2020independent}. 

\begin{definition}
	We define the regularized minimax mismatch coefficient by
	\begin{equation*}
		\cC_{\rho,\tau}^\dagger = \max\bigg\{\max_{\mu} \bigg\|\frac{d_\rho^{\mu, \nu_\tau^\dagger(\mu)}}{\rho}\bigg\|_\infty,\, \max_{\nu} \bigg\|\frac{d_\rho^{\mu_\tau^\dagger(\nu), \nu}}{\rho}\bigg\|_\infty\bigg\}.
	\end{equation*}
	Here, $\nu_\tau^\dagger(\mu)$ denotes the optimal policy of the min player when the max player adopts policy $\mu$:
	\[
		\nu_\tau^\dagger(\mu) = \arg \min_{\nu} V_\tau^{\mu, \nu}(\rho),
	\]
	and $\mu_\tau^\dagger(\nu)$ is defined in a symmetric way. The discounted state visitation distribution $d_\rho^{\mu,\nu}$ is defined as
	\[
	d_\rho^{\mu, \nu}(s) = (1-\gamma)\exlim{s_0 \sim \rho}{\sum_{t=0}^{\infty} \gamma^t P(s_t = s | s_0)}.
	\]
\end{definition}

We make note that Theorem \ref{thm:Prospero_A} is the direct corollary of the following theorems, by setting $\rho$ to the uniform distribution over $\cS$, where $\cC_\rho$ and $\|1/\rho\|_\infty$ admit a trivial upper bound $|\cS|$. By a slight abuse of notation, let $\KLrho{\zeta}{\zeta'}$ denote $\ex{s\sim \rho}{\KLs{\zeta}{\zeta'}}$ for $\rho \in \Delta(\cS)$.

\begin{theorem}
\label{thm:Prospero}
With $0 < \eta \le \frac{(1-\gamma)^3}{32000\cC_\rho}$, and $\alpha_i = \eta\tau$, we have
\begin{align*}
	\max\Big\{\KLrho{\best{\zeta}}{\zt}, \frac{1}{2}\KLrho{\best{\zeta}}{\bzt}, 3\eta \exlim{s\sim \rho}{\normbig{\Q{t}(s) - \best{Q}(s)}_\infty}\Big\}\le \frac{3000}{(1-\gamma)^2\tau}\Big(1-\frac{(1-\gamma)\eta\tau}{4}\Big)^{t}.
\end{align*}
\end{theorem}

\begin{theorem}
\label{thm:Miranda}
With $0 < \eta \le \frac{(1-\gamma)^3}{32000\cC_\rho}$, and $\alpha_i = \eta\tau$, we have
\begin{align*}
	&\max_{s\in\cS, \mu, \nu}\Big(V_\tau^{\mu,\bar{\nu}^{(t)}}(s) - V_\tau^{\bmut,\nu}(s)\Big)\le \frac{6000\|1/\rho\|_\infty}{(1-\gamma)^3\tau}\max\Big\{\frac{8}{(1-\gamma)^2\tau}, \frac{1}{\eta}\Big\}\Big(1-\frac{(1-\gamma)\eta\tau}{4}\Big)^{t},
\end{align*}
and
\begin{align*}
	&\max_{\mu, \nu}\big(V_\tau^{\mu,\bnut}(\rho) - V_\tau^{\bmut,\nu}(\rho)\big)\le \frac{6000\cC_{\rho,\tau}^\dagger}{(1-\gamma)^3\tau}\max\Big\{\frac{8}{(1-\gamma)^2\tau}, \frac{1}{\eta}\Big\}\Big(1-\frac{(1-\gamma)\eta\tau}{4}\Big)^{t}.
\end{align*}
\end{theorem}

\paragraph{Key lemmas.}
While Theorem \ref{thm:Prospero} and \ref{thm:Miranda} focus on the case where $\alpha_i = \eta\tau, \forall i \ge 1$, we assume in the following lemmas that the sequence $\{\alpha_i\}$ is non-increasing and bounded by $\eta\tau$ for generality.
For notational simplicity, we set $\Q{-1} = 0$, $\bar{\zeta}^{(-1)} = \bar{\zeta}^{(0)}$ and $\alpha_0 = 1$. It follows from the update rule \eqref{eq:update} that $\bar{\zeta}^{(1)} = \zeta^{(0)} = \bar{\zeta}^{(0)}$. Let us introduce  
\begin{equation} \label{def:arr_step}
\alpha_{l, t} = 
 \alpha_l \prod_{i=l+1}^{t}(1-\alpha_i), 
\end{equation}
and
\begin{equation} \label{def_lambda_lt}
	\lambda_{l,t} = \alpha_l \prod_{i=l+1}^{t}\Big(1-\frac{1-\gamma}{4}\cdot \alpha_i\Big).
\end{equation}
It follows straightforwardly that 
\[
\sum_{l=0}^t \alpha_{l,t} = \alpha_0 = 1.
\]
We start with the following lemma. 
\begin{lemma}
\label{lemma:core}
Suppose $0 < \eta \le 1/\tau$. It holds for all $t \ge 0$ that
\begin{align} 
			&\KLrho{\best{\zeta}}{\ztp} - (1-\eta\tau)\KLrho{\best{\zeta}}{\zt} +\Big(1-\eta\tau - \frac{4\eta}{1-\gamma}\Big)\KLrho{\bztp}{\bzt}  + \eta\tau \KLrho{\bztp}{\best{\zeta}}  \nonumber \\
		&\qquad+ \Big(1-\frac{2\eta}{1-\gamma}\Big)\KLrho{\ztp}{\bztp} + (1-\eta\tau)\KLrho{\bzt}{\zt}  - \frac{2\eta}{1-\gamma}\KLrho{\bzt}{\bar{\zeta}^{(t-1)}} \nonumber\\
		&\le \exlim{s\sim \rho}{2\eta \normbig{\Q{t+1}(s) - \best{Q}(s)}_\infty + \frac{4\eta^2}{1-\gamma} \normbig{\Q{t}(s) - \Q{t+1}(s)}_\infty + \frac{12\eta^2}{1-\gamma}\normbig{\Q{t-1}(s) - \Q{t}(s)}_\infty}. 	\label{eq:core}
	\end{align}
\end{lemma}
\begin{proof}
	See Appendix \ref{sec:pf_lemma_core}.
\end{proof}
We continue to bound the terms on the right hand side of \eqref{eq:core}. By a slight abuse of notation, we denote
\[
\normbig{\Q{t+1} - \best{Q}}_{\Gamma(\rho)} = \sup_{\chi\in\Gamma(\rho)}\exlim{s\sim\chi}{\normbig{\Q{t+1}(s) - \best{Q}(s)}_\infty},
\]
and
\[
\normbig{\Q{t+1} - \Q{t}}_{\Gamma(\rho)} = \sup_{\chi\in\Gamma(\rho)}\exlim{s\sim\chi}{\normbig{\Q{t+1}(s) - \Q{t}(s)}_\infty}.
\]
The following two lemmas establish a set of recursive bounds that relate $\big\{\normbig{\Q{l+1} - \best{Q}}_{\Gamma(\rho)}\big\}_{0\leq l \leq t}$ and $\big\{\normbig{\Q{l+1} - \Q{l}}_{\Gamma(\rho)}\big\}_{0\leq l \leq t}$ with $\big\{\KLrho{\bar{\zeta}^{(l+1)}}{\bar{\zeta}^{(l)}}\big\}_{0\leq l \leq t-1}$. 

\begin{lemma}
\label{lemma:Laertes}
	Suppose that $0 < \eta \le \min\{{(1-\gamma)}/{180}, {(1-\gamma)^2}/{48}\}$. It holds for all $t \ge 1$ that
\begin{equation}
	\normbig{\Q{t+1} - \Q{t}}_{\Gamma(\rho)} \le \frac{1+\gamma}{2} \sum_{l=1}^{t} \alpha_{l, t}\normbig{\Q{l} - \Q{l-1}}_{\Gamma(\rho)} + \frac{4\cC_\rho}{\eta} \cdot \sum_{l=1}^{t} \alpha_{l, t}\KLrho{\bar{\zeta}^{(l)}}{\bar{\zeta}^{(l-1)}}.
    \label{eq:Laertes}
\end{equation} 
When $t = 0$, we have $\normbig{\Q{1} - \Q{0}}_{\Gamma(\rho)} \le 2$.
\end{lemma}
\begin{proof}
	See Appendix \ref{sec:pf_lemma_Laertes}.
\end{proof}

\begin{lemma}
\label{lemma:Ophelia}
	Suppose that $0 < \eta \le (1-\gamma)^2/16$. It holds for all $t \ge 1$ that
\begin{align}
	\normbig{\Q{t+1} - \best{Q}}_{\Gamma(\rho)}  
	&\le \frac{1+\gamma}{2}\cdot{\sum_{l=0}^t\alpha_{l,t}}\Big(\normbig{\Q{l} - \best{Q}}_{\Gamma(\rho)} +  \frac{2\eta}{1-\gamma}\normbig{\Q{l} - \Q{l-1}}_{\Gamma(\rho)}\Big) +  2\alpha_{0,t}.
	\label{eq:Ophelia}
\end{align}
When $t = 0$, we have $\normbig{\Q{1} - \best{Q}}_{\Gamma(\rho)} \le \frac{2\gamma}{1-\gamma}.$
\end{lemma}
\begin{proof}
	See Appendix \ref{sec:pf_lemma_Ophelia}.
\end{proof}
The following lemma further demystifies the complicated recursive bounds showed in Lemmas \ref{lemma:Laertes}-\ref{lemma:Ophelia}. 
\begin{lemma}
	\label{lemma:Claudius}
	Under the assumption of Lemma \ref{lemma:Laertes} and \ref{lemma:Ophelia}, it holds for all $t \ge 0$ that
	\begin{align*}
	&\sum_{l=0}^{t}\lambda_{l+1, t+1} \Big[\eta\normbig{\best{Q} - \Q{l+1}}_{\Gamma(\rho)} + \frac{12\eta^2}{(1-\gamma)^2}\normbig{\Q{l+1} - \Q{l}}_{\Gamma(\rho)}\Big]\\
	&\le \frac{6250\eta\cC_\rho}{(1-\gamma)^3}\sum_{l=0}^{t-1}\lambda_{l+1,t+1}\KLrho{\bar{\zeta}^{(l+1)}}{\bar{\zeta}^{(l)}}  + \frac{550\eta}{(1-\gamma)^2}\lambda_{0, t+1}
	\end{align*}
\end{lemma}
\begin{proof}
	See Appendix \ref{sec:pf_lemma_Claudius}.
\end{proof}

\paragraph{Proof of Theorem~\ref{thm:Prospero}.} We are now ready to prove our main results. Starting with Lemma~\ref{lemma:core}, averaging \eqref{eq:core} with the weights $\lambda_{l,t}$ gives
\begin{align*}		
	&\sum_{l=0}^{t}\lambda_{l+1,t+1}\bigg[\KLrho{\best{\zeta}}{\zeta^{(l+1)}} - (1-\eta\tau)\KLrho{\best{\zeta}}{\zeta^{(l)}}\\
	&\qquad +\Big(1-\frac{2\eta}{1-\gamma}\Big)\KLrho{\zeta^{(l+1)}}{\bar{\zeta}^{(l+1)}} + 3\eta \exlim{s\sim \rho}{\normbig{\Q{l+1}(s) - \best{Q}(s)}_\infty}\\
	&\qquad +\Big(1-\eta\tau - \frac{4\eta}{1-\gamma}\Big)\KLrho{\bar{\zeta}^{(l+1)}}{\bar{\zeta}^{(l)}} - \frac{2\eta}{1-\gamma}\KLrho{\bar{\zeta}^{(l)}}{\bar{\zeta}^{(l-1)}}\bigg]\\
	&\le \sum_{l=0}^{t} \lambda_{l+1,t+1} \exlim{s\sim\rho}{ 5\eta \normbig{\Q{l+1}(s) - \best{Q}(s)}_\infty + \frac{4\eta^2}{1-\gamma} \normbig{\Q{l+1}(s) - \Q{l}(s)}_\infty+ \frac{13\eta^2}{1-\gamma}\normbig{\Q{l-1}(s) - \Q{l}(s)}_\infty}\\
	&\le 5\sum_{l=0}^{t} \lambda_{l+1,t+1}\Big[\eta\normbig{\best{Q} - \Q{l+1}}_{\Gamma(\rho)} + \frac{12\eta^2}{(1-\gamma)^2}\normbig{\Q{l+1} - \Q{l}}_{\Gamma(\rho)}\Big] \\
	&\le \frac{31250\eta \cC_\rho}{(1-\gamma)^3}\sum_{l=0}^{t-1}\lambda_{l+1,t+1}\KLrho{\bar{\zeta}^{(l+1)}}{\bar{\zeta}^{(l)}}  + \frac{2750\eta}{(1-\gamma)^2}\lambda_{0, t+1}
\end{align*}
for all $t \ge 0$, where the last line follows from Lemma~\ref{lemma:Claudius}. Rearranging terms, we have
\begin{align*}
	& \alpha_{t+1} \Big[\KLrho{\best{\zeta}}{\ztp} +\Big(1-\frac{2\eta}{1-\gamma}\Big)\KLrho{\ztp}{\bztp} + 3\eta\exlim{s\sim \rho}{\normbig{\Q{t+1}(s) - \best{Q}(s)}_\infty}\Big]\\
	&\qquad + \sum_{l=1}^{t} (\lambda_{l,t+1} - (1-\eta\tau)\lambda_{l+1,t+1})\KLrho{\best{\zeta}}{\zeta^{(l)}}\\
	&\qquad + \sum_{l=0}^{t-1}\bigg[\lambda_{l+1,t+1}\Big(1-\eta\tau - \frac{4\eta}{1-\gamma} - \frac{31250\eta\cC_\rho}{(1-\gamma)^3}\Big)-\lambda_{l+2,t+1}\frac{2\eta}{1-\gamma}\bigg]\KL{\bar{\zeta}^{(l+1)}}{\bar{\zeta}^{(l)}}\\
	&\le \frac{2750\eta}{(1-\gamma)^2}\lambda_{0, t+1} + (1-\eta\tau)\lambda_{1,t+1}\KLrho{\best{\zeta}}{\zeta^{(0)}} \le \Big(\frac{2750\eta}{(1-\gamma)^2} + \eta\Big)\lambda_{0, t+1}.
\end{align*}
Here, the last step results from 
\begin{align*}
(1-\eta\tau)\lambda_{1,t+1}\KLrho{\best{\zeta}}{\zeta^{(0)}} 
&= \alpha_1\cdot \frac{1-\eta\tau}{1-(1-\gamma)\alpha_1/4}\lambda_{0,t+1}\KLrho{\best{\zeta}}{\zeta^{(0)}}\\
&\le \eta \tau \lambda_{0,t+1}\KLrho{\best{\zeta}}{\zeta^{(0)}} \le \eta\tau (\log|\cA| + \log|\cB|)\lambda_{0,t+1} \le \eta \lambda_{0,t+1}.
\end{align*}
where we use the fact that $\alpha_1 = \eta\tau$ and the assumption on $\tau$ \eqref{eq:pomelo}.
With $0 < \eta \le \frac{(1-\gamma)^3}{32000\cC_\rho}$, and $\alpha_i = \eta\tau$, we have $\lambda_{l,t+1} - (1-\eta\tau)\lambda_{l+1,t+1} \ge 0$ (cf. \eqref{eq:we_love_entropy}), and 
\begin{align*}
	 &\lambda_{l+1,t+1}\Big(1-\eta\tau - \frac{4\eta}{1-\gamma} - \frac{31250\eta\cC_\rho}{(1-\gamma)^3}\Big)-\lambda_{l+2,t+1}\frac{2\eta}{1-\gamma}\\
	 &= \eta\tau \prod_{j=l+3}^{t+1} \Big(1-\frac{1-\gamma}{4}\alpha_j\Big)\Big[(1-\frac{1-\gamma}{4}\eta\tau)\Big(1-\eta\tau - \frac{4\eta}{1-\gamma} - \frac{31250\eta\cC_\rho}{(1-\gamma)^3}\Big) - \frac{2\eta}{1-\gamma}\Big] \ge 0.
\end{align*}
It follows that
\begin{align} 
	&\KLrho{\best{\zeta}}{\ztp} +\Big(1-\frac{2\eta}{1-\gamma}\Big)\KLrho{\ztp}{\bztp} + 3\eta \exlim{s\sim\rho}{\normbig{\Q{t+1}(s) - \best{Q}(s)}_\infty} \nonumber \\
	&\le \Big(\frac{2750}{(1-\gamma)^2\tau} + \frac{1}{\tau}\Big)\Big(1-\frac{(1-\gamma)\eta\tau}{4}\Big)^{t+1} < \frac{3000}{(1-\gamma)^2\tau}\Big(1-\frac{(1-\gamma)\eta\tau}{4}\Big)^{t+1}.
\label{eq:zeta_conv}
\end{align}
This proves the bound of $\KLrho{\best{\zeta}}{\ztp}$ and $3\eta \exlim{s\sim\rho}{\normbig{\Q{t+1}(s) - \best{Q}(s)}_\infty}$ in Theorem \ref{thm:Prospero}. Note that the bound holds trivially for $\KLrho{\best{\zeta}}{\zeta^{(0)}}$ and $3\eta \exlim{s\sim\rho}{\normbig{\Q{0}(s) - \best{Q}(s)}_\infty}$.  It remains to bound $\KL{\best{\zeta}}{\bztp}$, which we make use of the following lemma. 
\begin{lemma}
\label{lemma:bar_zeta_2_zeta}
With $0 < \eta \le {(1-\gamma)}/{8}$, we have
\begin{align*}
	&\frac{1}{2}\KLs{\best{\zeta}}{\bztp} + \eta\tau \KLs{\bztp}{\best{\zeta}}\\
	&\le (1-\eta\tau)\KLs{\best{\zeta}}{\zt} + \frac{2\eta}{1-\gamma}\KLs{\zt}{\bzt} + 2\eta \normbig{\Q{t}(s) - \best{Q}(s)}_\infty.
\end{align*}
\end{lemma}
\begin{proof}
	See Appendix \ref{sec:pf_lemma_bar_zeta_2_zeta}.
\end{proof}
Combining Lemma~\ref{lemma:bar_zeta_2_zeta} with \eqref{eq:zeta_conv} gives
\begin{align}
    &\frac{1}{2}\KLrho{\best{\zeta}}{\bztp} + \eta\tau \KLrho{\bztp}{\best{\zeta}} \nonumber \\
    &\le (1-\eta\tau)\Big(\KLrho{\best{\zeta}}{\zt} + \Big(1-\frac{2\eta}{1-\gamma}\Big)\KLrho{\zt}{\bzt} + 3\eta \exlim{s\sim\rho}{\normbig{\Q{t}(s) - \best{Q}(s)}_\infty}\Big) \nonumber\\
    &\le \frac{3000}{(1-\gamma)^2\tau}\Big(1-\frac{(1-\gamma)\eta\tau}{4}\Big)^{t+1}, \label{eq:zeta_zzz}
\end{align}
which concludes the proof of Theorem~\ref{thm:Prospero}. 

\paragraph{Proof of Theorem~\ref{thm:Miranda}.} We are now ready to bound the duality gap in Theorem~\ref{thm:Miranda}. Before proceeding, we introduce the following two lemmas. 

\begin{lemma}
\label{lemma:dualgap_markov_2_matrix}
It holds for any policy pair $(\mu, \nu)$ that
\begin{equation}
	\max_{\mu', \nu'}\Big(V_\tau^{\mu',\nu}(\rho) - V_\tau^{\mu,\nu'}(\rho)\Big) \le \frac{2\cC_{\rho,\tau}^\dagger}{1-\gamma} \exlim{s\sim\rho}{\max_{\mu', \nu'}\Big(f_{s}(\best{Q}, \mu', \nu) - f_{s}(\best{Q}, \mu, \nu')\Big)}
	\label{eq:dualgap_markov_2_matrix_rho}
\end{equation}
and
\begin{equation}
	\max_{s\in\cS, \mu', \nu'}\Big(V_\tau^{\mu',\nu}(s) - V_\tau^{\mu,\nu'}(s)\Big) \le \frac{2\|1/\rho\|_\infty}{1-\gamma}\exlim{s\sim\rho}{\max_{\mu', \nu'}\Big(f_{s}(\best{Q}, \mu', \nu) - f_{s}(\best{Q}, \mu, \nu')\Big)}.
	\label{eq:dualgap_markov_2_matrix_max}
\end{equation}
Here, $f_s(Q, \mu, \nu)$ is the one-step entropy-regularized game value at state $s$, i.e., 
\begin{equation} \label{eq:def_fs}
	f_s(Q, \mu, \nu) = \mu(s)^\top Q(s) \nu(s) + \tau \cH(\mu(s)) - \tau \cH(\nu(s)).
\end{equation}
\end{lemma}
\begin{proof}
	Note that \eqref{eq:dualgap_markov_2_matrix_max} is a slight generalization of \cite[Lemma 32]{wei2021last}. The proof can be found in Appendix \ref{sec:pf_lemma_dualgap_markov_2_matrix}.
\end{proof}
\begin{lemma}[\mbox{\cite[Lemma 4]{cen2021fast}}]
\label{lemma:dualgap_matrix}
It holds for all $s\in\cS$ and policy pair $\mu, \nu$ that
\begin{equation*}
	\max_{\mu',\nu'} \big( f_s(\best{Q}, \mu', \nu) - f_s(\best{Q}, \mu, \nu')\big) \le \frac{4}{(1-\gamma)^2\tau}\KLs{\best{\zeta}}{\zeta} + \tau \KLs{\zeta}{\best{\zeta}}.
\end{equation*}
\end{lemma}

Putting all pieces together, we arrive at
\begin{align*}
	\max_{\mu, \nu}\big(V_\tau^{\mu,\bnut}(\rho) - V_\tau^{\bmut,\nu}(\rho)\big)
	&\le \frac{2\cC_{\rho,\tau}^\dagger}{1-\gamma}\Big(\frac{4}{(1-\gamma)^2\tau}\KLrho{\best{\zeta}}{\bztp} + \tau \KLrho{\bztp}{\best{\zeta}} \Big)\\
	&\le \frac{2\cC_{\rho,\tau}^\dagger}{1-\gamma} \max\Big\{\frac{8}{(1-\gamma)^2\tau}, \frac{1}{\eta}\Big\}\Big(\frac{1}{2}\KLrho{\best{\zeta}}{\bztp} + \eta\tau \KLrho{\bztp}{\best{\zeta}} \Big)\\
    &\le \frac{6000\cC_{\rho,\tau}^\dagger}{(1-\gamma)^3\tau}\max\Big\{\frac{8}{(1-\gamma)^2\tau}, \frac{1}{\eta}\Big\}\Big(1-\frac{(1-\gamma)\eta\tau}{4}\Big)^{t},
\end{align*}
where the last line follows from \eqref{eq:zeta_zzz}. We omit the proof for $\max_{s\in\cS, \mu, \nu}\Big(V_\tau^{\mu,\bar{\nu}^{(t)}}(s) - V_\tau^{\bmut,\nu}(s)\Big)$ for brevity as it follows essentially from the same argument.

\section{Analysis for the finite-horizon setting}

Throughout the analysis, we restrict our choice of the step size for value update to $\alpha_t = \eta\tau$. We start with the following lemma which parallels Lemma \ref{lemma:one_step_policy_bound} in the infinite-horizon Markov game setting; for brevity we omit the proof.
\begin{lemma}
    \label{lemma:one_step_policy_bound_epi}
    With $0 < \eta \le 1/\tau$, it holds for all $s \in \cS$, $h \in [H]$ and $t \ge 0$ that
    \begin{equation}
            \max\big\{\normbig{\bmutp_h(s) - \mutp_h(s)}_1, \normbig{\bnutp_h(s) - \nutp_h(s)}_1\big\} \le 2\eta H.
    \end{equation}
    In addition, we have
\begin{equation}
\max\{\|\log \zt_h(s)\|_\infty, \|\log \bzt_h(s)\|_\infty, \|\log \besth{\zeta}(s)\|_\infty\} \le \frac{2H}{\tau}.
\label{eq:log_bound_epi}
\end{equation}

\end{lemma}

\begin{lemma}
\label{lemma:Hamlet}
With $0 < \eta \le \frac{1}{8H}$, it holds for all $0 \le t_1 \le t_2$, $h\in[H]$ and $s\in \cS$ that
\begin{align*}
	&\KLs{\besth{\zeta}}{\zeta_h^{(t_2)}} + (1-4\eta H)\KLs{\zeta_h^{(t_2)}}{\bar{\zeta}_h^{(t_2)}}\\ 
	&\le (1-\eta\tau)^{t_2-t_1}\Big(\KLs{\besth{\zeta}}{\zeta_h^{(t_1)}} + (1-4\eta H)\KLs{\zeta_h^{(t_1)}}{\bar{\zeta}_h^{(t_1)}}\Big)  +4\eta\sum_{l=t_1}^{t_2} (1-\eta\tau)^{t_2-l}  \normbig{\Q{l}_h(s) - \best{Q}(s)}_\infty.
\end{align*}
\end{lemma}
\begin{proof}
	See Appendix \ref{sec:pf_lemma_Hamlet}.
\end{proof}

\begin{lemma}
\label{lemma:Horatio}
With $0 < \eta \le \frac{1}{8H}$, it holds for all $0 < t_1 \le t_2$, $2\le h \le H$ and $s\in \cS$ that
\begin{align*}
	&\big|\Q{t_2}_{h-1}(s,a,b) - Q_{h-1,\tau}^\star(s,a,b)\big|\\
	 &\le 2(1-\eta\tau)^{t_2-t_1}H + 10\eta\tau\exlim{s'\sim P_{h-1}(\cdot|s,a,b)}{\sum_{l=t_1-1}^{t_2-1}(1-\eta\tau)^{t_2-1-l} \normbig{\Q{l}_h(s) - \besth{Q}(s)}_\infty }\\
	 &\qquad + \tau (1-\eta\tau)^{t_2-t_1}\exlim{s'\sim P_{h-1}(\cdot|s,a,b)}{\KLs{\besth{\zeta}}{\zeta^{(t_1-1)}_h}  + (1-4\eta H)\KLs{\zeta^{(t_1-1)}_h}{\bar{\zeta}^{(t_1-1)}_h} }.
\end{align*}
\end{lemma}
\begin{proof}
	See Appendix \ref{sec:pf_lemma_Horatio}.
\end{proof}
 
 \paragraph{Proof of Theorem~\ref{thm:Antonio}.} We prove Theorem \ref{thm:Antonio} by induction. By definition, we have
\[
	\normbig{Q_{H, \tau}^\star - \Q{0}_{H}}_\infty = \normbig{Q_{H, \tau}^\star}_\infty \le 1,
\]
and $\normbig{Q_{H, \tau}^\star - \Q{t}_{H}}_\infty = \normbig{r_H - r_H}_\infty = 0$ for $t > 0$. So \eqref{eq:Q_conv_epi} holds trivially for $h = H$.
When the statement holds for some $h$, we can invoke Lemma \ref{lemma:Horatio} with $t_1 = T_h + 1$ and $t_2 = t \ge T_{h-1}$, which yields
\begin{align*}
	 \normbig{\Q{t}_{h-1} - Q_{h-1,\tau}^\star}
	&\le 2(1-\eta\tau)^{t-T_h - 1}H + 10\eta\tau\exlim{s'\sim P(\cdot|s,a,b)}{\sum_{l=T_h}^{t-1}(1-\eta\tau)^{t-1-l} \normbig{\Q{l}_h(s) - \besth{Q}(s)}_\infty }\\
	&\qquad + \tau (1-\eta\tau)^{t-T_h - 1}\exlim{s'\sim P(\cdot|s,a,b)}{\KLs{\besth{\zeta}}{\zeta^{(T_h)}_h}  + (1-4\eta H)\KLs{\zeta^{(T_h)}_h}{\bar{\zeta}^{(T_h)}_h}}\\
		 &\le 2(1-\eta\tau)^{t-T_h - 1}H + 10\eta\tau\exlim{s'\sim P(\cdot|s,a,b)}{\sum_{l=T_h}^{t-1}(1-\eta\tau)^{t-T_h-1} l^{H-h} }\\
	&\qquad + \tau (1-\eta\tau)^{t-T_h - 1}\exlim{s'\sim P(\cdot|s,a,b)}{\KLs{\besth{\zeta}}{\zeta^{(T_h)}_h}  + (1-4\eta H)\KLs{\zeta^{(T_h)}_h}{\bar{\zeta}^{(T_h)}_h} }\\
	&\le (1-\eta\tau)^{t-T_{h-1}} (1-\eta\tau)^{T_{\mathsf{start}} -1}\Big[10H + 10\eta\tau t^{H-h+1} \Big],
\end{align*}
where the last step results from 
\begin{align*}
	&\tau\Big(\KLs{\besth{\zeta}}{\zeta^{(T_h)}_h}  + (1-4\eta H)\KLs{\zeta^{(T_h)}_h}{\bar{\zeta}^{(T_h)}_h}\Big)\\
	&\le \tau\Big(\normbig{\log \besth{\mu}(s) - \log\mu^{(T_h)}_h(s)}_\infty + \normbig{\log \besth{\nu}(s) - \log\nu^{(T_h)}_h(s)}_\infty\\
	&\qquad + \normbig{\log\mu^{(T_h)}_h(s) - \log\bar\mu^{(T_h)}_h(s)}_\infty + \normbig{\log\nu^{(T_h)}_h(s) - \log\bar\nu^{(T_h)}_h(s)}_\infty\Big)\\
	&\le \tau\Big(\max\big\{\normbig{\log \besth{\mu}(s)}_\infty \normbig{\log\mu^{(T_h)}_h(s)}_\infty\big\} + \max\big\{\normbig{\log \besth{\nu}(s)}_\infty, \normbig{\log\nu^{(T_h)}_h(s)}_\infty\big\}\\
	&\qquad + \max\big\{\normbig{\log\mu^{(T_h)}_h(s)}_\infty, \normbig{\log\bar\mu^{(T_h)}_h(s)}_\infty\big\} + \max\big\{\normbig{\log\nu^{(T_h)}_h(s)}_\infty, \normbig{\log\bar\nu^{(T_h)}_h(s)}_\infty\big\}\Big)\\
	&\le 8H,
\end{align*}
where the last step results from Lemma~\ref{lemma:one_step_policy_bound_epi} (cf. \eqref{eq:log_bound_epi}).
	Therefore, with $T_{\mathsf{start}} = \lceil\frac{1}{\eta\tau}\log H\rceil$ we can guarantee that 
		\begin{align*}
	\normbig{\Q{t}_{h-1} - Q_{h-1,\tau}^\star} &\le 10(1-\eta\tau)^{t-T_{h-1}} (1-\eta\tau)^{T_{\mathsf{start}} -1}\Big[H + \eta\tau t^{H-h+1} \Big]\\
	&\le (1-\eta\tau)^{t-T_{h-1}} t^{H-h+1}.
	\end{align*}
	This completes the proof for \eqref{eq:Q_conv_epi}. Regarding \eqref{eq:gap_conv_epi}, 
	we start by the following lemmas, which are simply Lemma \ref{lemma:bar_zeta_2_zeta} and Lemma \ref{lemma:dualgap_matrix} applied to the episodic setting.
\begingroup
\setcounter{tmp}{\value{lemma}}
\setcounterref{lemma}{lemma:bar_zeta_2_zeta} 
\addtocounter{lemma}{-1}
\renewcommand\thelemma{\arabic{lemma}A}

\begin{lemma}
\label{lemma:bar_zeta_2_zeta_epi}
With $0 < \eta \le \frac{1}{8H}$, we have
\begin{align*}
	&\frac{1}{2}\KLs{\besth{\zeta}}{\bztp_h} + \eta\tau \KLs{\bztp_h}{\besth{\zeta}}\\
	&\le (1-\eta\tau)\KLs{\besth{\zeta}}{\zt_h} + 2\eta H\KLs{\zt_h}{\bzt_h} + 2\eta \normbig{\Q{t}_h(s) - \besth{Q}(s)}_\infty.
\end{align*}
\end{lemma}

\setcounterref{lemma}{lemma:dualgap_matrix} 
\addtocounter{lemma}{-1}

\begin{lemma}
\label{lemma:dualgap_matrix_epi}
It holds for all $h\in[H]$, $s\in\cS$ and policy pair $\mu, \nu$ that
\begin{equation*}
	\max_{\mu',\nu'} \big( f_s(\besth{Q}, \mu'_h, \nu_h) - f_s(\best{Q}, \mu_h, \nu'_h)\big) \le \frac{4H^2}{\tau}\KLs{\besth{\zeta}}{\zeta_h} + \tau \KLs{\zeta_h}{\besth{\zeta}}.
\end{equation*}
\end{lemma}

\endgroup
\setcounter{lemma}{\thetmp}
We conclude that for $0\le t_1 \le t_2-1$,
\begin{align*}
	&\max_{\mu,\nu} \big( f_s(\besth{Q}, \mu_h, \bar{\nu}^{(t_2)}_h) - f_s(\best{Q}, \bar{\mu}^{(t_2)}_h, \nu_h)\big) \\
	&\overset{\mathrm{(i)}}{\le} \frac{4H^2}{\tau}\KLs{\besth{\zeta}}{\bar{\zeta}^{(t_2)}_h} + \tau \KLs{\bar{\zeta}^{(t_2)}_h}{\besth{\zeta}}\\
	&\le \max\Big\{\frac{8H^2}{\tau}, \frac{1}{\eta}\Big\} \Big(\frac{1}{2}\KLs{\besth{\zeta}}{\bar{\zeta}^{(t_2)}_h} + \eta\tau \KLs{\bar{\zeta}^{(t_2)}_h}{\besth{\zeta}}\Big)\\
	&\overset{\mathrm{(ii)}}{\le} \max\Big\{\frac{8H^2}{\tau}, \frac{1}{\eta}\Big\}\Big((1-\eta\tau)\KLs{\besth{\zeta}}{\zeta_h^{(t_2-1)}} + 2\eta H\KLs{\zeta_h^{(t_2-1)}}{\bar{\zeta}_h^{(t_2-1)}} + 2\eta \normbig{\Q{t_2-1}_h(s) - \besth{Q}(s)}_\infty\Big)\\
	&\overset{\mathrm{(iii)}}{\le} \max\Big\{\frac{8H^2}{\tau}, \frac{1}{\eta}\Big\} \Big((1-\eta\tau)^{t_2-t_1}\Big(\KLs{\besth{\zeta}}{\zeta_h^{(t_1)}} + (1-4\eta H)\KLs{\zeta_h^{(t_1)}}{\bar{\zeta}_h^{(t_1)}}\Big) \\
	&\qquad + 6\eta\sum_{l=t_1}^{t_2} (1-\eta\tau)^{t_2-l}  \normbig{\Q{l}_h(s) - \besth{Q}(s)}_\infty\Big),
\end{align*}
where (i) invokes Lemma \ref{lemma:dualgap_matrix_epi}, (ii) invokes Lemma \ref{lemma:bar_zeta_2_zeta_epi} and (iii) results from Lemma \ref{lemma:Hamlet}. 
It is straightforward to verify that the above inequality holds for $0 \le t_1 \le t_2$, by omitting the third step. Substitution of \eqref{eq:Q_conv_epi} into the above inequality yields
\begin{align}
	&\max_{\mu,\nu} \big( f_s(\besth{Q}, \mu_h, \bar{\nu}^{(t)}_h) - f_s(\best{Q}, \bar{\mu}^{(t)}_h, \nu_h)\big) \nonumber\\
	&\le \max\Big\{\frac{8H^2}{\tau}, \frac{1}{\eta}\Big\} \Big((1-\eta\tau)^{t-T_h}\Big(\KLs{\besth{\zeta}}{\zeta_h^{(T_h)}} + (1-4\eta H)\KLs{\zeta_h^{(T_h)}}{\bar{\zeta}_h^{(T_h)}}\Big) \nonumber\\
	&\qquad + 6\eta\sum_{l=T_h}^{t} (1-\eta\tau)^{t-l} (1-\eta\tau)^{l-T_h} l^{H-h}\Big)\nonumber\\
	&\le (1-\eta\tau)^{t-T_h}\max\Big\{\frac{8H^2}{\tau}, \frac{1}{\eta}\Big\} \Big(\frac{8H}{\tau} + 6\eta t^{H-h+1}\Big).
	\label{eq:dualgap_epi_matrix_step}
\end{align}
We prove the following results instead, where \eqref{eq:gap_conv_epi} is a direct consequence of \eqref{eq:gap_conv_epi_split} by summing up the two inequalities,
\begin{equation}
	\begin{cases}
		\max\limits_{s\in\cS,\mu}\Big(V_{h,\tau}^{\mu, \bnut}(s) - \besth{V}(s)\Big) \le 2(1-\eta\tau)^{t-T_h}\max\Big\{\frac{8H^2}{\tau}, \frac{1}{\eta}\Big\} \Big(\frac{8H}{\tau} + 6\eta t^{H-h+1}\Big)\\
		\max\limits_{s\in\cS,\mu}\Big(\besth{V}(s) - V_{h,\tau}^{\bmut, \nu}(s)\Big) \le 2(1-\eta\tau)^{t-T_h}\max\Big\{\frac{8H^2}{\tau}, \frac{1}{\eta}\Big\} \Big(\frac{8H}{\tau} + 6\eta t^{H-h+1}\Big)\\
	\end{cases} .
	\label{eq:gap_conv_epi_split}
\end{equation}
We prove by induction. Note that when $h = H$, we have $V_{H,\tau}^{\mu, \nu}(s) = f_s(r_H, \mu_H, \nu_H) = f_s(Q_{H,\tau}^\star, \mu_H, \nu_H)$ and the claim holds by invoking \eqref{eq:dualgap_epi_matrix_step}. When the claim holds for some $2\le h \le H$, we have
	\begin{align*}
		&V_{h-1,\tau}^{\mu, \bnut}(s) - V_{h-1,\tau}^\star(s) \\
		&= \mu_{h-1}(s)^\top Q_{h-1,\tau}^{\mu, \bnut}(s)\bnut_{h-1}(s) + \tau \cH\big(\mu_{h-1}(s)\big) - \tau \cH\big(\bnut_{h-1}(s)\big)\\
		&\qquad - \mu_{h-1,\tau}^\star(s)^\top Q_{h-1,\tau}^\star(s)\nu_{h-1,\tau}^\star(s) + \tau \cH\big(\mu_{h-1,\tau}^\star(s)\big) - \tau \cH\big(\nu_{h-1,\tau}^\star(s)\big)\\
		&= f_s(Q_{h-1,\tau}^\star, \mu_{h-1}, \bnut_{h-1}) - f_s(Q_{h-1,\tau}^\star, \mu_{h-1,\tau}^\star, \nu_{h-1,\tau}^\star) + \mu_{h-1}(s)^\top \big(Q_{h-1,\tau}^{\mu, \bnut}(s) - Q_{h-1,\tau}^\star(s)\big)\bnut_{h-1}(s)\\
		&\le f_s(Q_{h-1,\tau}^\star, \mu_{h-1}, \bnut_{h-1}) - f_s(Q_{h-1,\tau}^\star, \bmut_{h-1}, \nu_{h-1,\tau}^\star) + \max_{s'\in\cS}\Big[V_{h,\tau}^{\mu,\bnut}(s') - V_{h,\tau}^\star(s')\Big]\\
		&\le \max\limits_{\mu_{h-1}', \nu_{h-1}'}\Big(f_s(Q_{h-1,\tau}^\star, \mu_{h-1}', \bnut_{h-1}) - f_s(Q_{h-1,\tau}^\star, \bmut_{h-1}, \nu_{h-1}')\Big) + \max_{s'\in\cS}\Big[V_{h,\tau}^{\mu,\bnut}(s') - V_{h,\tau}^\star(s')\Big]\\
		&\le(1-\eta\tau)^{t-T_{h-1}}\max\Big\{\frac{8H^2}{\tau}, \frac{1}{\eta}\Big\} \Big(\frac{8H}{\tau} + 6\eta t^{H-h+2}\Big)\\
		&\qquad + 2(1-\eta\tau)^{t-T_{h}}\max\Big\{\frac{8H^2}{\tau}, \frac{1}{\eta}\Big\} \Big(\frac{8H}{\tau} + 6\eta t^{H-h+1}\Big)\\
		&\le2(1-\eta\tau)^{t-T_{h-1}}\max\Big\{\frac{8H^2}{\tau}, \frac{1}{\eta}\Big\} \Big(\frac{8H}{\tau} + 6\eta t^{H-h+2}\Big).
	\end{align*}
	Taking maximum over $\mu$ verifies the claim for $h-1$, thereby finishing the proof. The bound for $\max\limits_{s\in\cS,\mu}\Big(\besth{V}(s) - V_{h,\tau}^{\bmut, \nu}(s)\Big)$ can be established by following a similar argument and is therefore omitted.

\section{Proof of key lemmas for the infinite-horizon setting}

\subsection{Proof of Lemma \ref{lemma:core}}
\label{sec:pf_lemma_core}

Before proceeding, we shall introduce the following useful lemma that quantifies the distance between two consecutive updates, whose proof can be found in Appendix \ref{sec:pf_lemma_one_step_policy_bound}.
\begin{lemma}
    \label{lemma:one_step_policy_bound}
    For $0 < \eta \le 1/\tau$, it holds for all $s \in \cS$ and $t \ge 0$ that
    \begin{subequations}
    \begin{align}
            \max\big\{\normbig{\bmutp(s) - \mutp(s)}_1,\, \normbig{\bnutp(s) - \nu^{(t+1)}(s)}_1\big\} & \le \dfrac{2\eta}{1-\gamma} , \label{eq:diff_bar_unbar}\\
        \max\big\{\normbig{\bmutp(s) - \bmut(s)}_1,\, \normbig{\bnutp(s) - \bnut(s)}_1\big\} & \le \frac{6\eta}{1-\gamma}, \label{eq:diff_bar}
    \end{align}
    \end{subequations}
    and that
    \begin{equation}
	\max\big\{\normbig{\log \zt(s)}_\infty, \normbig{\log \bzt(s)}_\infty, \normbig{\log \best{\zeta}(s)}_\infty \big\}\le \frac{2}{(1-\gamma)\tau}.
	\label{eq:log_bound}
	\end{equation}
\end{lemma}
For notational simplicity, we use $x\overset{\mathbf{1}}{=}y$ to denote equivalence up to a global shift for two vectors $x, y$, i.e.
\begin{equation}\label{eq:equiv_one}
	x = y + c\cdot \one
\end{equation}
for some constant $c \in \mathbb{R}$.
Taking logarithm on the both sides of the update rule \eqref{eq:update}, we get
\begin{align}
	\begin{cases}
		\log \mutp(s) - (1-\eta\tau)\log \mut(s) &\overset{\mathbf{1}}{=} \eta \Q{t+1}(s)\bnutp(s)\\
		\log \nu^{(t+1)}(s) - (1-\eta\tau)\log \nut(s) &\overset{\mathbf{1}}{=} -\eta \Q{t+1}(s)^\top\bmutp(s)
	\end{cases}.
	\label{eq:log_update}
\end{align}
On the other hand, it holds for the QRE $(\best{\mu}, \best{\nu})$ that
\begin{align}
	\begin{cases}
		\eta\tau\log \best{\mu}(s) &\overset{\mathbf{1}}{=} \eta \best{Q}(s)\best{\nu}(s)\\
		\eta\tau\log \best{\nu}(s) &\overset{\mathbf{1}}{=} -\eta \best{Q}(s)^\top \best{\mu}(s)
	\end{cases}.
	\label{eq:log_opt}
\end{align}

Subtracting \eqref{eq:log_opt} from \eqref{eq:log_update} and taking inner product with $\bztp(s) - \best{\zeta}(s)$ gives
\begin{align} 
		&\innprod{\log \ztp(s) - (1-\eta\tau)\log \zt(s) - \eta\tau \log \best{\zeta}(s), \bztp(s) - \best{\zeta}(s)} \nonumber \\
		&= \eta\innprod{\bmutp(s) - \best{\mu}(s), \Q{t+1}(s)\bnutp(s) - \best{Q}(s)\best{\nu}(s)} \nonumber \\
		&\qquad - \eta \innprod{\bnutp(s) - \best{\nu}(s), \Q{t+1}(s)^\top\bmutp(s) - \best{Q}(s)^\top \best{\mu}(s)} \nonumber\\
		&= \eta \innprod{\bmutp(s) - \best{\mu}(s), (\Q{t+1}(s) - \best{Q}(s))\bnutp(s)} \nonumber\\
		&\qquad - \eta \innprod{\bnutp(s) - \best{\nu}(s), (\Q{t+1}(s) - \best{Q}(s))^\top\bmutp(s)} \nonumber\\
		&= -\eta\innprod{\best{\mu}(s), (\Q{t+1}(s) - \best{Q}(s))\bnutp(s)} + \eta \innprod{\best{\nu}(s), (\Q{t+1}(s) - \best{Q}(s))^\top\bmutp(s)} \nonumber\\
		&\le 2\eta \normbig{\Q{t+1}(s) - \best{Q}(s)}_\infty.
	\label{eq:core_step_1}
\end{align}

We continue to rewrite the LHS of \eqref{eq:core_step_1} as
\begin{align*}
	&\innprod{\log \ztp(s) - (1-\eta\tau)\log \zt(s) - \eta\tau \log \best{\zeta}(s), \bztp(s) - \best{\zeta}(s)}\\
	&=- \innprod{\log \ztp(s) - (1-\eta\tau)\log \zt(s) - \eta\tau \log \best{\zeta}(s), \best{\zeta}(s)}\\
	&\qquad + \innprod{\log \bztp(s) - (1-\eta\tau)\log \bzt(s) - \eta\tau \log \best{\zeta}(s), \bztp(s)}\\
	&\qquad + \innprod{\log \ztp(s) - \log \bztp(s), \bztp(s)}\\
	&\qquad - (1-\eta\tau)\innprod{\log \zt(s) - \log \bzt(s), \bztp(s)} \\
	&= \KLs{\best{\zeta}}{\ztp}- (1-\eta\tau)\KLs{\best{\zeta}}{\zt} \\
	&\qquad + (1-\eta\tau)\KLs{\bztp}{\bzt} + \eta\tau \KLs{\bztp}{\best{\zeta}}\\
	&\qquad + \KLs{\ztp}{\bztp} - \innprod{\log \bztp(s)-\log {\zeta}^{(t+1)}(s),  \bztp(s)- {\zeta}^{(t+1)}(s)}\\
	&\qquad + (1-\eta\tau)\KLs{\bzt}{\zt} - (1-\eta\tau)\innprod{\log \zt(s) - \log \bzt(s), \bztp(s) - \bzt(s)}.
\end{align*}
Rearranging terms, we have
\begin{align*}
	&\KLs{\best{\zeta}}{\ztp} - (1-\eta\tau)\KLs{\best{\zeta}}{\zt} +(1-\eta\tau)\KLs{\bztp}{\bzt} \\
	&\qquad + \eta\tau \KLs{\bztp}{\best{\zeta}} + \KLs{\ztp}{\bztp}+ (1-\eta\tau)\KLs{\bzt}{\zt}\\
	&\qquad - \innprod{\log \bztp(s)-\log {\zeta}^{(t+1)}(s),  \bztp(s)- {\zeta}^{(t+1)}(s)}\\
	&\qquad 	- (1-\eta\tau)\innprod{\log \zt(s) - \log \bzt(s), \bztp(s) - \bzt(s)}\\
	&\le 2\eta \normbig{\Q{t+1}(s) - \best{Q}(s)}_\infty.
\end{align*}
It remains to upper bound 
$$\innprod{ \log \bztp(s)-\log {\zeta}^{(t+1)}(s),  \bztp(s)- {\zeta}^{(t+1)}(s) }\quad\mbox{and}\quad \innprod{ \log \zt(s) - \log \bzt(s), \bztp(s) - \bzt(s) }. $$ For the first term, note that
\begin{align} 
		&\innprod{\log \bmutp(s)-\log \mutp(s),  \bmutp(s)- \mutp(s)} \nonumber \\
		&= \eta \innprod{\Q{t}(s)\bnut(s) - \Q{t+1}(s)\bnutp(s), \bmutp(s) - \mutp(s)} \nonumber\\
		&\le \eta \normbig{\Q{t}(s)\bnut(s) - \Q{t+1}(s)\bnutp(s)}_1\normbig{\bmutp(s) - \mutp(s)}_1.
	\label{eq:core_step_2}
\end{align}
Here, $\normbig{\Q{t}(s)\bnut(s) - \Q{t+1}(s)\bnutp(s)}_1$ can be bounded as
\begin{align*}
	&\normbig{\Q{t}(s)\bnut(s) - \Q{t+1}(s)\bnutp(s)}_1 \\
	&\le \normbig{\Q{t+1}(s)\big(\bnut(s) - \bnutp(s)\big)}_1 + \normbig{\big(\Q{t}(s)- \Q{t+1}(s)\big)\bnut(s)}_1\\
	&\le \frac{2}{1-\gamma}\normbig{\bnut(s) - \bnutp(s)}_1 + \normbig{\Q{t}(s)- \Q{t+1}(s)}_\infty.
\end{align*}
Plugging the above inequality into \eqref{eq:core_step_2} and invoking Young's inequality yields
\begin{align}
	&\innprod{\log \bmutp(s)-\log \mutp(s),  \bmutp(s)- \mutp(s)}\nonumber\\
	&\le \frac{\eta}{1-\gamma}\Big(\normbig{\bnutp(s) - \bar{\nu}^{(t)}(s)}_1^2 + \normbig{\bmutp(s) - \mutp(s)}_1^2\Big)\nonumber\\
	&\qquad + \eta \normbig{\Q{t}(s)- \Q{t+1}(s)}_\infty\normbig{\bmutp(s) - \mutp(s)}_1\nonumber\\
	&\le \frac{2\eta}{1-\gamma}\KLs{\bnutp}{\bar{\nu}^{(t)}} + \frac{2\eta}{1-\gamma}\KLs{\mutp}{\bmutp} + \frac{2\eta^2}{1-\gamma} \normbig{\Q{t}(s) - \Q{t+1}(s)}_\infty,
	\label{eq:pinsker_trick_example}
\end{align}
where the last step results from Pinsker's inequality and Lemma \ref{lemma:one_step_policy_bound}. Similarly, we have
\begin{equation*}
\begin{aligned}
	&\innprod{\log \bnutp(s)-\log {\nu}^{(t+1)}(s),  \bnutp(s)- {\nu}^{(t+1)}(s)}\\
	&\le \frac{2\eta}{1-\gamma}\KLs{\bmutp}{\bar{\mu}^{(t)}}  + \frac{2\eta}{1-\gamma}\KLs{\nu^{(t+1)}}{\bnutp} + \frac{2\eta^2}{1-\gamma} \normbig{\Q{t}(s) - \Q{t+1}(s)}_\infty.
\end{aligned}
\end{equation*}
Combining the above two inequalities gives
\begin{align*}
	&\innprod{\log \bztp(s)-\log {\zeta}^{(t+1)}(s),  \bztp(s)- {\zeta}^{(t+1)}(s)}\\
		&\le \frac{2\eta}{1-\gamma}\KLs{\bztp}{\bzt} + \frac{2\eta}{1-\gamma}\KLs{\ztp}{\bztp} + \frac{4\eta^2}{1-\gamma} \normbig{\Q{t}(s) - \Q{t+1}(s)}_\infty.
\end{align*}
By a similar argument, when $t \ge 1$:
\begin{align*}
	&\innprod{\log \zt(s) - \log \bzt(s), \bztp(s) - \bzt(s)}\\
	&=\eta \innprod{\Q{t}(s)\bnut(s) - \Q{t-1}(s)\bar{\nu}^{(t-1)}(s), \bmutp(s) - \bmut(s)}\\
	&\qquad - \eta \innprod{\Q{t}(s)^\top\bmut(s) - \Q{t-1}(s)^\top\bar{\mu}^{(t-1)}(s), \bnutp(s) - \bnut(s)}\\
	&\le \frac{2\eta}{1-\gamma}\KLs{\bzt}{\bar{\zeta}^{(t-1)}} + \frac{2\eta}{1-\gamma}\KLs{\bztp}{\bzt} \\
	&\qquad + \eta \big(\normbig{\bmutp(s) - \bmut(s)}_1 + \normbig{\bnutp(s) - \bnut(s)}_1\big) \normbig{\Q{t}(s) - \Q{t-1}(s)}_\infty\\
	&\le \frac{2\eta}{1-\gamma}\KLs{\bzt}{\bar{\zeta}^{(t-1)}} + \frac{2\eta}{1-\gamma}\KLs{\bztp}{\bzt}   + \frac{12\eta^2}{1-\gamma} \normbig{\Q{t}(s) - \Q{t-1}(s)}_\infty.
\end{align*}
Note that the above inequality trivially holds for $t = 0$, since $\log \zeta^{(0)}(s) = \log \bar{\zeta}^{(0)}(s)$.

Putting pieces together, we conclude that
\begin{align*}
	&\KLs{\best{\zeta}}{\ztp} - (1-\eta\tau)\KLs{\best{\zeta}}{\zt} +\Big(1-\eta\tau - \frac{4\eta}{1-\gamma}\Big)\KLs{\bztp}{\bzt} \\
	&\qquad + \eta\tau \KLs{\bztp}{\best{\zeta}} + \Big(1-\frac{2\eta}{1-\gamma}\Big)\KLs{\ztp}{\bztp} + (1-\eta\tau)\KLs{\bzt}{\zt}\\
	&\qquad - \frac{2\eta}{1-\gamma}\KLs{\bzt}{\bar{\zeta}^{(t-1)}}\\
	&\le 2\eta \normbig{\Q{t+1}(s) - \best{Q}(s)}_\infty + \frac{4\eta^2}{1-\gamma} \normbig{\Q{t}(s) - \Q{t+1}(s)}_\infty + \frac{12\eta^2}{1-\gamma}\normbig{\Q{t-1}(s) - \Q{t}(s)}_\infty.
\end{align*}
Averaging state $s$ over the initial state distribution $\rho$ completes the proof.

\subsection{Proof of Lemma \ref{lemma:Laertes}}
\label{sec:pf_lemma_Laertes}

By definition of $Q$, it holds for $t \ge 1$ that
\begin{equation}
	\begin{aligned}
		\big|\Q{t+1}(s,a,b) - \Q{t}(s,a,b)\big| \le \gamma \ex{s'\sim P(\cdot|s,a,b)}{\big|\V{t}(s') - \V{t-1}(s')\big|}.
	\end{aligned}
	\label{eq:Q_diff_2_V_diff}
\end{equation}
Recall the definition of $f_s(Q, \mu, \nu)$ in \eqref{eq:def_fs} as the one-step entropy-regularized game value at state $s$, i.e., 
\begin{equation*}
	f_s(Q, \mu, \nu) = \mu(s)^\top Q(s) \nu(s) + \tau \cH(\mu(s)) - \tau \cH(\nu(s)),
\end{equation*}
which we further simplify the notation by introducing
\begin{equation*}
f_s^{(t)} = 	f_s(\Q{t}, \bmut, \bnut) .
\end{equation*}
By recursively applying the update rule
$\V{t}(s) = (1-\alpha_{t})\V{t-1}(s) + \alpha_{t}f_s^{(t)}$,
we get
\begin{align*}
	\V{t}(s) &= \alpha_{0,t} \V{0} + \sum_{l=1}^t \alpha_{l,t}f_s(\Q{l}, \bar{\mu}^{(l)}, \bar{\nu}^{(l)})= \sum_{l=0}^t \alpha_{l,t}f_s^{(l)}.
\end{align*}
Therefore,
\begin{align} 
		\big|\V{t}(s) - \V{t-1}(s)\big|&= \alpha_t\big| f_s^{(t)}  - \V{t-1}(s) \big| \nonumber\\
		&= \alpha_t \sum_{l=0}^{t-1} \alpha_{l, t-1} \big|f_s^{(t)} - f_s^{(l)}\big| \nonumber \\
		&\le \alpha_t \sum_{l=0}^{t-1} \alpha_{l, t-1} \sum_{j=l}^{t-1}\big|f_s^{(j+1)} - f_s^{(j)}\big|.
	\label{eq:V_diff}
\end{align}
The next lemma enables us to upper bound $\big|f_s^{(t+1)} - f_s^{(t)}\big|$ with $\normbig{\Q{t+1}(s) - \Q{t}(s)}_\infty$ and $\KLs{\bztp}{\bzt}$ as well as their counterparts in the $(t-1)$-th iteration. The proof is postponed to Appendix \ref{sec:pf_lemma_f_diff}.
\begin{lemma}
	\label{lemma:f_diff}
	For any $t \ge 0$, $\eta \le (1-\gamma)/180$, we have
	\begin{equation*}
		\begin{aligned}
			\big|f_s^{(t+1)} - f_s^{(t)}\big|&\le \Big\|\Q{t+1}(s) - \Q{t}(s)\Big\|_\infty + \Big(\frac{3}{\eta} + \frac{4}{1-\gamma}\Big)\KLs{\bztp}{\bzt}  \\
			&\qquad + \frac{12\eta}{1-\gamma}\normbig{\Q{t}(s) - \Q{t-1}(s)}_\infty + \frac{2}{1-\gamma}\KLs{\bzt}{\bar{\zeta}^{(t-1)}}.
		\end{aligned}
	\end{equation*}
\end{lemma}
Plugging the above lemma into \eqref{eq:V_diff}, 
\begin{align*}
		&\big|\V{t}(s) - \V{t-1}(s)\big|\\
		&\le \alpha_t \sum_{l=0}^{t-1} \alpha_{l, t-1} \sum_{j=l}^{t-1}\bigg[\normbig{\Q{j+1}(s) - \Q{j}(s)}_\infty + \Big(\frac{3}{\eta} + \frac{4}{1-\gamma}\Big)\KLs{\bar{\zeta}^{(j+1)}}{\bar{\zeta}^{(j)}}\bigg]\\
		&\qquad +  \alpha_t \sum_{l=0}^{t-1} \alpha_{l, t-1} \sum_{j=l}^{t-1}\bigg[\frac{12\eta}{1-\gamma}\normbig{\Q{j}(s) - \Q{j-1}(s)}_\infty + \frac{2}{1-\gamma}\KLs{\bar{\zeta}^{(j)}}{\bar{\zeta}^{(j-1)}} \bigg]\\
		&\le \alpha_t \sum_{l=0}^{t-1} \alpha_{l, t-1} \sum_{j=l}^{t-1}\bigg[\Big(1+\frac{12\eta}{1-\gamma}\Big)\normbig{\Q{j+1}(s) - \Q{j}(s)}_\infty + \Big(\frac{3}{\eta} + \frac{6}{1-\gamma}\Big)\KLs{\bar{\zeta}^{(j+1)}}{\bar{\zeta}^{(j)}}\bigg]\\
		&\qquad +  \alpha_t \sum_{l=0}^{t-1} \alpha_{l, t-1} \bigg[\frac{12\eta}{1-\gamma}\normbig{\Q{l}(s) - \Q{l-1}(s)}_\infty + \frac{2}{1-\gamma}\KLs{\bar{\zeta}^{(l)}}{\bar{\zeta}^{(l-1)}} \bigg]\\
		&\le \sum_{j=0}^{t-1}\alpha_{j+1}\sum_{l=0}^{j} \alpha_{l, t-1}\bigg[\Big(1+\frac{12\eta}{1-\gamma}\Big)\normbig{\Q{j+1}(s) - \Q{j}(s)}_\infty + \Big(\frac{3}{\eta} + \frac{6}{1-\gamma}\Big)\KLs{\bar{\zeta}^{(j+1)}}{\bar{\zeta}^{(j)}}\bigg]\\
		&\qquad +  \alpha_t \sum_{l=0}^{t-2} \alpha_{l+1, t-1} \bigg[\frac{12\eta}{1-\gamma}\normbig{\Q{l+1}(s) - \Q{l}(s)}_\infty + \frac{2}{1-\gamma}\KLs{\bar{\zeta}^{(l+1)}}{\bar{\zeta}^{(l)}} \bigg],
\end{align*}
where the last step is due to $\alpha_t \le \alpha_j$ for all $j \le t$. 
To continue, by definition of $\alpha_t$ we have $\alpha_t\alpha_{l+1, t-1} \le \alpha_{l+1,t-1}(1-\alpha_t) = \alpha_{l+1,t}$ for $0\le l < t$, and that
\begin{align*}
	\alpha_{j+1}\sum_{l=0}^{j} \alpha_{l, t-1} &= \alpha_{j+1}\sum_{l=0}^{j} \Big(\prod_{i=l+1}^{t-1}(1-\alpha_i) - \prod_{i=l}^{t-1}(1-\alpha_i)\Big)\\
	&= \alpha_{j+1} \prod_{i=j+1}^{t-1}(1-\alpha_i) \\
	&\le \alpha_{j+1} \prod_{i=j+2}^{t}(1-\alpha_i) = \alpha_{j+1,t}.
\end{align*}
Plugging the inequality above into the previous relation gives
\begin{align*}
		&\big|\V{t}(s) - \V{t-1}(s)\big|\\
		&\le \sum_{j=0}^{t-1} \alpha_{j+1, t}\bigg[\Big(1+\frac{12\eta}{1-\gamma}\Big)\normbig{\Q{j+1}(s) - \Q{j}(s)}_\infty + \Big(\frac{3}{\eta} + \frac{6}{1-\gamma}\Big)\KLs{\bar{\zeta}^{(j+1)}}{\bar{\zeta}^{(j)}}\bigg]\\
		&\qquad + \sum_{l=0}^{t-2} \alpha_{l+1, t} \bigg[\frac{12\eta}{1-\gamma}\normbig{\Q{l+1}(s) - \Q{l}(s)}_\infty + \frac{2}{1-\gamma}\KLs{\bar{\zeta}^{(l+1)}}{\bar{\zeta}^{(l)}} \bigg]\\
		&\le \sum_{l=0}^{t-1} \alpha_{l+1, t}\bigg[\Big(1+\frac{24\eta}{1-\gamma}\Big)\normbig{\Q{l+1}(s) - \Q{l}(s)}_\infty + \frac{4}{\eta}\KLs{\bar{\zeta}^{(l+1)}}{\bar{\zeta}^{(l)}}\bigg].
		\label{eq:V_diff_final}
\end{align*}
Plugging the above inequality into \eqref{eq:Q_diff_2_V_diff} leads to
\begin{align*}
	&\big|\Q{t+1}(s,a,b) - \Q{t}(s,a,b)\big| \\
	&\le \gamma \mathop\mathbb{E}\limits_{s'\sim P(\cdot|s,a,b))}\Bigg\{\sum_{l=0}^{t-1} \alpha_{l+1, t}\bigg[\Big(1+\frac{24\eta}{1-\gamma}\Big)\normbig{\Q{l+1}(s') - \Q{l}(s')}_\infty + \frac{4}{\eta}\mathsf{KL}_{s'}(\bar{\zeta}^{(l+1)}\,\|\,\bar{\zeta}^{(l)})\bigg]\Bigg\}.
\end{align*}
When $\eta \le \frac{(1-\gamma)^2}{48\gamma}$, we have $\gamma (1 + \frac{24\eta}{1-\gamma}) \le \frac{1+\gamma}{2}$ and hence that
\begin{align*}
	&\big|\Q{t+1}(s,a,b) - \Q{t}(s,a,b)\big| \\
	&\le \mathop\mathbb{E}\limits_{s'\sim P(\cdot|s,a,b))}\bigg\{\frac{1+\gamma}{2}\sum_{l=0}^{t-1} \alpha_{l+1, t}\Big[\normbig{\Q{l+1}(s') - \Q{l}(s')}_\infty + \frac{4}{\eta}\mathsf{KL}_{s'}(\bar{\zeta}^{(l+1)}\,\|\,\bar{\zeta}^{(l)})\Big]\bigg\}.
\end{align*}
Let $x^{(t+1)}\in \cA^\cS$ and $y^{(t+1)}\in \cB^\cS$ be defined as for any $s\in\cS$:
\[
(x^{(t+1)}(s), y^{(t+1)}(s)) = \arg\min_{(a,b)\in\cA\times\cB}\big|\Q{t+1}(s,a,b) - \Q{t}(s,a,b)\big|.
\]
It follows that $\forall \chi \in \Gamma(\rho)$, we have $\chi P_{x^{(t+1)},y^{(t+1)}} \in \Gamma(\rho)$ and hence
\begin{align}
	&\exlim{s\sim \chi}{\normbig{\Q{t+1}(s) - \Q{t}(s)}_\infty}\nonumber\\
	&= \exlim{\substack{s\sim \chi,\\ a = x^{(t+1)}(s),b = y^{(t+1)}(s)}}{\big|\Q{t+1}(s,a,b) - \Q{t}(s,a,b)\big|}\nonumber\\
	&\le \exlim{s'\sim \chi P_{x^{(t+1)},y^{(t+1)}}}{\frac{1+\gamma}{2}\sum_{l=0}^{t-1} \alpha_{l+1, t}\Big[\normbig{\Q{l+1}(s') - \Q{l}(s')}_\infty + \frac{4}{\eta}\mathsf{KL}_{s'}(\bar{\zeta}^{(l+1)}\,\|\,\bar{\zeta}^{(l)})\Big]}\nonumber\\
	&\le \frac{1+\gamma}{2}\sum_{l=0}^{t-1} \alpha_{l+1, t}\bigg[\normbig{\Q{l+1} - \Q{l}}_{\Gamma(\rho)} + \frac{4}{\eta}\cdot \Big\|\frac{\chi P_{x^{(t+1)},y^{(t+1)}}}{\rho}\Big\|_\infty \KLrho{\bar{\zeta}^{(l+1)}}{\bar{\zeta}^{(l)}}\bigg]\nonumber\\
	&\le \frac{1+\gamma}{2}\sum_{l=0}^{t-1} \alpha_{l+1, t}\bigg[\normbig{\Q{l+1} - \Q{l}}_{\Gamma(\rho)} + \frac{4\mathcal{C}_\rho}{\eta}\KLrho{\bar{\zeta}^{(l+1)}}{\bar{\zeta}^{(l)}}\bigg].
	\label{eq:Gammalizer}
\end{align}
Taking the supremum over $\chi \in \Gamma(\rho)$ completes the proof for $t\geq 1$. To complete the proof, note that when $t = 0$, we have $\normbig{\Q{0} - \Q{1}}_{\Gamma(\rho)} = \normbig{\Q{1}}_{\Gamma(\rho)} \le 2$.

\subsection{Proof of Lemma \ref{lemma:Ophelia}}
\label{sec:pf_lemma_Ophelia}
Note that it suffices to show for $t \ge 0$, $s\in \cS$, $(a,b)\in\cA\times\cB$:
\begin{align}
	&\big|\Q{t+1}(s,a,b) - \best{Q}(s,a,b)\big|\nonumber\\
	&\le \frac{1+\gamma}{2}\cdot\exlim{s'\sim P(s,a,b)}{{\sum_{l=0}^t\alpha_{l,t}}\Big[\normbig{\Q{l}(s') - \best{Q}(s')}_\infty +  \frac{2\eta}{1-\gamma}\normbig{\Q{l}(s') - \Q{l-1}(s')}_\infty\Big]} +  2\alpha_{0,t}.
	\label{eq:Ophelia_lite}
\end{align}
The remaining step follows a similar argument as \eqref{eq:Gammalizer} and is therefore omitted. 

To establish \eqref{eq:Ophelia_lite}, notice that we have for $t \ge 0$, 
\begin{align}
	\Q{t+1}(s,a,b) - \best{Q}(s,a,b) &= \gamma \ex{s'\sim P(\cdot|s,a,b)}{\V{t}(s') - \best{V}(s')}\nonumber\\
	 &= \gamma \ex{s'\sim P(\cdot|s,a,b)}{\sum_{l=0}^{t} \alpha_{l,t}(f_{s'}^{(l)} - f_{s'}^\star)}.
	 \label{eq:Q_gap}
\end{align}
To continue, we start by decomposing $f_s^{(t)} - f_{s}^\star$ as
\begin{align*}
	f_s^{(t)} - f_{s}^\star&= f_s(\Q{t}, \bmut, \bnut) - f_s(\best{Q}, \best{\mu}, \best{\nu})\\
	&= \prn{f_s(\Q{t}, \bmut, \bnut) - f_s(\Q{t}, \bmut, \best{\nu})} + f_s(\Q{t}, \bmut, \best{\nu}) - f_s(\best{Q}, \best{\mu}, \best{\nu})\\
	&\le \prn{f_s(\Q{t}, \bmut, \bnut) - f_s(\Q{t}, \bmut, \best{\nu})} + f_s(\best{Q}, \bmut, \best{\nu})- f_s(\best{Q}, \best{\mu}, \best{\nu})\\
	&\qquad + \normbig{\Q{t}(s) - \best{Q}(s)}_\infty\\
	&\le f_s(\Q{t}, \bmut, \bnut) - f_s(\Q{t}, \bmut, \best{\nu}) + \normbig{\Q{t}(s) - \best{Q}(s)}_\infty.
\end{align*}
We bound the first two terms with the following lemma, whose proof can be found in Appendix \ref{sec:pf_lemma_one_step_regret}.
\begin{lemma}
	\label{lemma:one_step_regret}
	It holds for all $t\ge0$, $s\in \cS$ and $\nu(s) \in \Delta(\cB)$ that
	\begin{align*}
		&f_s(\Q{t}, \bmut, \bnut) - f_s(\Q{t}, \bmut, \nu)\\
		&\le \frac{2\eta}{1-\gamma}\normbig{\Q{t}(s) - \Q{t-1}(s)}_\infty + \frac{2}{1-\gamma}\Big(\KLs{\bmut}{\mutm} + \KLs{\mutm}{\bar{\mu}^{(t-1)}}\Big)\\
		&\qquad   - \frac{1}{\eta} \big(1-\frac{4\eta}{1-\gamma}\big) \KLs{\nut}{\bnut}- \frac{1-\eta\tau}{\eta}\KLs{\bnut}{\nutm} \\
		&\qquad + \frac{1-\eta\tau}{\eta}\KLs{\nu}{\nutm} - \frac{1}{\eta} \KLs{\nu}{\nut}.
\end{align*}
\end{lemma}

Applying Lemma \ref{lemma:one_step_regret} with $\nu(s) = \best{\nu}(s)$ gives
\begin{align}
f_s^{(t)} - f_{s}^\star
	&\le \normbig{\Q{t}(s) - \best{Q}(s)}_\infty +  \frac{2\eta}{1-\gamma}\normbig{\Q{t}(s) - \Q{t-1}(s)}_\infty \nonumber \\
	&\qquad+ \frac{1-\eta\tau}{\eta}\KLs{\best{\nu}}{\nutm} - \frac{1}{\eta} \KLs{\best{\nu}}{\nut} \nonumber\\
	&\qquad   - \frac{1}{\eta} \big(1-\frac{4\eta}{1-\gamma}\big) \KLs{\nut}{\bnut}- \frac{1-\eta\tau}{\eta}\KLs{\bnut}{\nutm} \nonumber\\
	&\qquad +\frac{2}{1-\gamma}\Big(\KLs{\bmut}{\mutm} + \KLs{\mutm}{\bar{\mu}^{(t-1)}}\Big). \label{eq:opt_gap_pos} 
\end{align}
By a similar argument, we can derive
\begin{align} 
	f_s^\star - f_s^{(t)}&\le \normbig{\Q{t}(s) - \best{Q}(s)}_\infty +  \frac{2\eta}{1-\gamma}\normbig{\Q{t}(s) - \Q{t-1}(s)}_\infty \nonumber \\
	&\qquad+ \frac{1-\eta\tau}{\eta}\KLs{\best{\mu}}{\mutm} - \frac{1}{\eta} \KLs{\best{\mu}}{\mut} \nonumber\\
	&\qquad   - \frac{1}{\eta} \big(1-\frac{4\eta}{1-\gamma}\big) \KLs{\mut}{\bmut}- \frac{1-\eta\tau}{\eta}\KLs{\bmut}{\mutm} \nonumber \\
	&\qquad +\frac{2}{1-\gamma}\Big(\KLs{\bnut}{\nutm} + \KLs{\nutm}{\bar{\nu}^{(t-1)}}\Big).
\label{eq:opt_gap_neg}
\end{align}

Combining \eqref{eq:opt_gap_pos} $+\frac{1-\gamma}{4} \cdot$ \eqref{eq:opt_gap_neg} gives
\begin{align} 
		&(1-\frac{1-\gamma}{4})(f_s^{(t)} - f_s^\star) \nonumber \\
		&\le (1+\frac{1-\gamma}{4}) \Big[\normbig{\Q{t}(s) - \best{Q}(s)}_\infty +  \frac{2\eta}{1-\gamma}\normbig{\Q{t}(s) - \Q{t-1}(s)}_\infty\Big] \nonumber \\
		&\qquad + \frac{1-\eta\tau}{\eta} \Big[\KLs{\best{\nu}}{\nutm} + \frac{1-\gamma}{4}\KLs{\best{\mu}}{\mutm}\Big] - \frac{1}{\eta} \Big[\KLs{\best{\nu}}{\nut} + \frac{1-\gamma}{4}\KLs{\best{\mu}}{\mut}\Big] \nonumber\\
		&\qquad + \frac{2}{1-\gamma}\Big[\KLs{\mutm}{\bar{\mu}^{(t-1)}} + \frac{1-\gamma}{4}\KLs{\nutm}{\bar{\nu}^{(t-1)}} \Big] \nonumber\\
		&\qquad - \frac{1}{\eta} \big(1-\frac{4\eta}{1-\gamma}\big)\Big[ \frac{1-\gamma}{4} \KLs{\mut}{\bmut} + \KLs{\nut}{\bnut}\Big] \nonumber\\
		&\qquad +(\frac{2}{1-\gamma} - \frac{1-\eta\tau}{\eta}\cdot\frac{1-\gamma}{4})\KLs{\bmut}{\mutm} + (\frac{2}{1-\gamma} \cdot \frac{1-\gamma}{4} - \frac{1-\eta\tau}{\eta})\KLs{\bnut}{\nutm}.
	\label{eq:opt_gap_pos_mix}
\end{align}
With $0 < \eta \le (1-\gamma)^2/16$,  we have 
\begin{align*}
\frac{2}{1-\gamma} - \frac{1-\eta\tau}{\eta}\cdot\frac{1-\gamma}{4}) & \le 0, \qquad
\frac{2}{1-\gamma} \cdot \frac{1-\gamma}{4} - \frac{1-\eta\tau}{\eta}  \le 0, \\
\frac{1}{\eta} \big(1-\frac{4\eta}{1-\gamma}\big) \cdot \frac{1-\gamma}{4} & \ge \frac{2}{1-\gamma} \cdot \frac{1}{1-\eta\tau}.
\end{align*}
To proceed, we introduce a shorthand notation
\begin{align*}
G^{(t)}(s) &= \frac{1}{\eta}\Big[\KLs{\best{\nu}}{\nut} + \frac{1-\gamma}{4}\KLs{\best{\mu}}{\mut} \Big]\\
&\qquad + \frac{2}{(1-\gamma)(1-\eta\tau)}\Big[\KLs{\mut}{\bar{\mu}^{(t)}} + \KLs{\nut}{\bar{\nu}^{(t)}} \Big].
\end{align*}
We can then write \eqref{eq:opt_gap_pos_mix} as
\begin{align}
	(1-\frac{1-\gamma}{4})(f_s^{(t)} - f_s^\star) &\le (1+\frac{1-\gamma}{4}) \Big[\normbig{\Q{t}(s) - \best{Q}(s)}_\infty +  \frac{2\eta}{1-\gamma}\normbig{\Q{t}(s) - \Q{t-1}(s)}_\infty\Big]\nonumber\\
	&\qquad + (1-\eta\tau)G^{(t-1)}(s) - G^{(t)}(s).
	\label{eq:f_gap}
\end{align}
Note that when $t = 0$, we have
\begin{align}
	f_s^{(0)} - f_s^\star&= \tau\log|\cA| - \tau \log |\cB| - \best{\mu}(s)^\top \best{Q}(s)\best{\nu}(s) -\tau \cH(\best{\mu}(s)) + \tau \cH(\best{\nu}(s))\nonumber\\
	&= \max_{\mu(s)}\min_{\nu(s)} f_s(\Q{0},\mu, \nu) - \max_{\mu(s)}\min_{\nu(s)}f_s(\best{Q},\mu, \nu) \nonumber\\
	&\le \normbig{\Q{0}(s) - \best{Q}(s)}_\infty.
	\label{eq:f_gap_0}
\end{align}

Substitution of \eqref{eq:f_gap} and \eqref{eq:f_gap_0} into \eqref{eq:Q_gap} gives
\begin{align*}
	&\Q{t+1}(s,a,b) - \best{Q}(s,a,b) \\
	&= \gamma \ex{s'\sim P(\cdot|s,a,b)}{\sum_{l=0}^t\alpha_{l,t}(f_{s'}^{(l)} - f_{s'}^\star)}\\
	&\le \gamma \ex{s'\sim P(s,a,b)}{\alpha_{0,t}\normbig{\Q{0}(s') - \best{Q}(s')}_\infty }\\
	&\qquad +\gamma \cdot \frac{1+(1-\gamma)/4}{1-(1-\gamma)/4}\exlim{s'\sim P(s,a,b)}{{\sum_{l=1}^t\alpha_{l,t}}\Big[\normbig{\Q{l}(s') - \best{Q}(s')}_\infty +  \frac{2\eta}{1-\gamma}\normbig{\Q{l}(s') - \Q{l-1}(s')}_\infty\Big]}\\
	&\qquad +  \frac{\gamma}{1-(1-\gamma)/4} \exlim{s'\sim P(s,a,b)}{(1-\eta\tau)\sum_{l=1}^t\alpha_{l,t}G^{(l-1)}(s') - \sum_{l=1}^t\alpha_{l,t}G^{(l)}(s')}.
\end{align*}
Note that 
\begin{align*}
	(1-\eta\tau)\sum_{l=1}^t\alpha_{l,t}G^{(l-1)}(s') - \sum_{l=1}^t\alpha_{l,t}G^{(l)}(s')
	&\le \sum_{l=1}^{t-1} ((1-\eta\tau)\alpha_{l+1,t} - \alpha_{l,t}) G^{(l)}(s') + \alpha_{1,t}G^{(0)}(s')\\
	&\le \alpha_{1,t}G^{(0)}(s') \le 2\alpha_{0,t} \eta\tau G^{(0)}(s') \le 2\alpha_{0,t},
\end{align*}
where the second step is due to
\begin{align}
	(1-\eta\tau)\alpha_{l+1, t}-\alpha_{l,t} &= ((1-\eta\tau)\alpha_{l+1} - \alpha_l (1-\alpha_{l+1}))\prod_{j=l+2}^{t}\alpha_j\nonumber\\
	&\le ((1-\eta\tau)\alpha_{l+1} - \alpha_{l+1} + \alpha_l\alpha_{l+1})\prod_{j=l+2}^{t}\alpha_j\nonumber\\
	&= \alpha_{l+1}(\alpha_l -\eta\tau)\prod_{j=l+2}^{t}\alpha_j \le 0.
	\label{eq:we_love_entropy}
\end{align}
We conclude that
\begin{align*}
	&\Q{t+1}(s,a,b) - \best{Q}(s,a,b) \\
	&\le \gamma \cdot \frac{1+(1-\gamma)/4}{1-(1-\gamma)/4}\exlim{s'\sim P(s,a,b)}{{\sum_{l=0}^t\alpha_{l,t}}\Big[\normbig{\Q{l}(s') - \best{Q}(s')}_\infty +  \frac{2\eta}{1-\gamma}\normbig{\Q{l}(s') - \Q{l-1}(s')}_\infty\Big]}\\
		&\qquad + 2\alpha_{0,t}\\
	&\le \frac{1+\gamma}{2}\cdot\exlim{s'\sim P(s,a,b)}{{\sum_{l=0}^t\alpha_{l,t}}\Big[\normbig{\Q{l}(s') - \best{Q}(s')}_\infty +  \frac{2\eta}{1-\gamma}\normbig{\Q{l}(s') - \Q{l-1}(s')}_\infty\Big]}\\
		&\qquad +  2\alpha_{0,t}.
\end{align*}
The other side of \eqref{eq:Ophelia_lite} can be obtained by computing $\frac{1-\gamma}{4} \cdot$ \eqref{eq:opt_gap_pos} + \eqref{eq:opt_gap_neg} and following a similar argument, and is therefore omitted.
To conclude the proof, we note that
for $t = 0$, we have $\big|\Q{1}(s,a,b) - \best{Q}(s,a,b)\big| \le \gamma \max_{s'\in\cS}|f_{s'}^{(0)} - f_{s'}^\star|\le \frac{2\gamma}{1-\gamma}$.

\subsection{Proof of Lemma \ref{lemma:Claudius}}
\label{sec:pf_lemma_Claudius}

 For $t \ge 1$, let 
\begin{equation*}
	u_{t} = \eta\normbig{\best{Q}(s) - \Q{t}(s)}_{\Gamma(\rho)} + \frac{12\eta^2}{(1-\gamma)^2}\normbig{\Q{t}(s) - \Q{t-1}(s)}_{\Gamma(\rho)}.
\end{equation*}
It follows that
\begin{align*}
	u_{1} &\le \frac{2\gamma \eta}{1-\gamma}+\frac{24\eta^2}{(1-\gamma)^3} \le 1.
\end{align*}
When $t \ge 1$, invoking Lemma \ref{lemma:Laertes} and Lemma \ref{lemma:Ophelia} gives
\begin{equation}
\begin{aligned}
	u_{t+1}&\le \Big(1-\frac{1-\gamma}{2}\Big){\sum_{l=1}^t\alpha_{l,t}}\Big[\eta\normbig{\Q{l} - \best{Q}}_{\Gamma(\rho)} + \Big(\frac{2\eta^2}{1-\gamma} + \frac{12\eta^2}{(1-\gamma)^2}\Big)\normbig{\Q{l} - \Q{l-1}}_{\Gamma(\rho)}\Big]\\
	&\qquad + \frac{48\eta \mathcal{C}_\rho}{(1-\gamma)^2}\sum_{l=1}^{t} \alpha_{l, t}\KLrho{\bar{\zeta}^{(l)}}{\bar{\zeta}^{(l-1)}} + 2\alpha_{0,t}\eta + \alpha_{0,t}\eta\normbig{\Q{0} - \best{Q}}_{\Gamma(\rho)}\\
	&\le \Big(1-\frac{1-\gamma}{3}\Big) {\sum_{l=1}^t\alpha_{l,t}}u_{l} + \frac{48\eta\cC_\rho}{(1-\gamma)^2}\sum_{l=1}^{t} \alpha_{l, t}\KLrho{\bar{\zeta}^{(l)}}{\bar{\zeta}^{(l-1)}} + \frac{4\eta}{1-\gamma}\alpha_{0,t}.
\end{aligned}
\label{eq:u_2_alpha_u_bound}
\end{equation}
Let
\[
	\beta_{l,t} = \alpha_l \prod_{i=l+1}^{t}\Big(1-\frac{1-\gamma}{3}\cdot \alpha_i\Big).
\]
It follows that for $t \ge 0$,
\begin{align}
	&\sum_{l=1}^{t+1}\alpha_{l, t+1} u_l \nonumber\\
	& = (1-\alpha_{t+1})\sum_{l=1}^{t}\alpha_{l, t} u_l + \alpha_{t+1}u_{t+1}\nonumber\\
	&\le \Big(1- \frac{1-\gamma}{3}\cdot\alpha_{t+1}\Big) \sum_{l=1}^{t}\alpha_{l, t}u_l + \alpha_{t+1} \frac{48\eta\cC_\rho}{(1-\gamma)^2} \cdot \sum_{l=1}^{t} \alpha_{l, t}\KLrho{\bar{\zeta}^{(l)}}{\bar{\zeta}^{(l-1)}} + \frac{4\eta}{1-\gamma}\alpha_{t+1}\alpha_{0,t}\nonumber\\
	&\le \prod_{l=2}^{t+1}\Big(1- \frac{1-\gamma}{3}\cdot\alpha_{l}\Big)\alpha_{1,1}u_1 + \frac{48\eta\cC_\rho}{(1-\gamma)^2} \sum_{i=1}^{t}\beta_{i+1,t+1} \sum_{l=1}^{i}\alpha_{l,i}\KLrho{\bar{\zeta}^{(l)}}{\bar{\zeta}^{(l-1)}} +  \frac{4\eta}{1-\gamma}\sum_{i=1}^{t}\alpha_{0,i}\beta_{i+1,t+1}\nonumber\\
	&\le \beta_{1,t+1}u_1 + \frac{48\eta\cC_\rho}{(1-\gamma)^2} \sum_{l=1}^{t}\sum_{i=l}^{t}\alpha_{l,i}\beta_{i+1,t+1}\KLrho{\bar{\zeta}^{(l)}}{\bar{\zeta}^{(l-1)}}  + \frac{4\eta}{1-\gamma} \sum_{i=1}^{t}\alpha_{0,i}\beta_{i+1,t+1}\nonumber\\
	&\le \frac{200\eta\cC_\rho}{(1-\gamma)^2}\sum_{l=1}^{t} \beta_{l, t+1}\KLrho{\bar{\zeta}^{(l)}}{\bar{\zeta}^{(l-1)}}  + \frac{18\eta}{1-\gamma}\beta_{0,t+1},
	\label{eq:alpha_u_bound}
\end{align}
where the last step is due to the following lemma. Similar lemma has appeared in prior works (see i.e., \cite[Lemma 36]{wei2021last}). Our version features a simpler proof, which is postponed to Appendix \ref{sec:pf_lemma_cabbage}.
\begin{lemma}
	\label{lemma:cabbage}
	Let two sequences $\{\delta_i\}, \{\xi_i\}$ be defined as
	\begin{equation*}
		\delta_i = 1 - c_1 \alpha_i, \qquad \text{and} \qquad \xi_i = 1 - c_2 \alpha_i,
	\end{equation*}
	where the constants $c_1, c_2$ satisfy $0 < c_1 < c_2 < \frac{1}{2\alpha_i}$. 
	For $l \le t$, let $\delta_{l, t} = \alpha_l \prod_{i=l+1}^{t}\delta_i$ and $\xi_{l, t} = \alpha_l \prod_{i=l+1}^{t}\xi_i$, where we take $\delta_{l, l} =\xi_{l, l} =\alpha_l$. We have
	\begin{equation*}
		\sum_{i=l}^{t}\xi_{l,i}\delta_{i+1,t}\le \Big(1+\frac{2}{c_2-c_1}\Big)\delta_{l,t}.
	\end{equation*}
\end{lemma}

Substitution of \eqref{eq:alpha_u_bound} into \eqref{eq:u_2_alpha_u_bound} gives
\begin{align*}
	u_{t+1} & \le \Big(1-\frac{1-\gamma}{3}\Big) {\sum_{l=1}^t\alpha_{l,t}}u_{l} + \frac{48\eta}{(1-\gamma)^2}\sum_{l=1}^{t} \alpha_{l, t}\KLrho{\bar{\zeta}^{(l)}}{\bar{\zeta}^{(l-1)}} + \frac{4\eta}{1-\gamma}\alpha_{0,t}\\
	&\le \frac{200\eta \cC_\rho}{(1-\gamma)^2}\sum_{l=1}^{t} \beta_{l, t}\KLrho{\bar{\zeta}^{(l)}}{\bar{\zeta}^{(l-1)}}  + \frac{18\eta}{1-\gamma}\beta_{0,t}+ \frac{48\eta\cC_\rho}{(1-\gamma)^2}\sum_{l=1}^{t} \alpha_{l, t}\KLrho{\bar{\zeta}^{(l)}}{\bar{\zeta}^{(l-1)}} + \frac{4\eta}{1-\gamma}\alpha_{0,t}\\
	&\le \frac{250\eta\cC_\rho}{(1-\gamma)^2}\sum_{l=1}^{t} \beta_{l, t}\KLrho{\bar{\zeta}^{(l)}}{\bar{\zeta}^{(l-1)}}  + \frac{22\eta}{1-\gamma}\beta_{0,t}.
\end{align*}
for $t \ge 1$. It is straightforward to verify that the above inequality holds for $t = 0$ as well.


So we conclude that
\begin{align*}
	&\sum_{l=0}^{t}\lambda_{l+1, t+1} u_{l+1} = \sum_{i=0}^{t}\lambda_{i+1, t+1} u_{i+1} \\
	&\le \sum_{i=0}^{t}\lambda_{i+1, t+1}\Big[\frac{250\eta\cC_\rho}{(1-\gamma)^2}\sum_{l=1}^{i} \beta_{l, i}\KLrho{\bar{\zeta}^{(l)}}{\bar{\zeta}^{(l-1)}}  + \frac{22\eta}{1-\gamma}\beta_{0,i}\Big] \\
	&= \frac{250\eta\cC_\rho}{(1-\gamma)^2}\sum_{l=1}^{t}\sum_{i=l}^{t} \beta_{l, i}\lambda_{i+1,t+1}\KLrho{\bar{\zeta}^{(l)}}{\bar{\zeta}^{(l-1)}}  + \frac{22\eta}{1-\gamma}\sum_{i=0}^{t}\beta_{0,i}\lambda_{i+1, t+1}\\
	&\le \frac{6250\eta\cC_\rho}{(1-\gamma)^3}\sum_{l=1}^{t}\lambda_{l,t+1}\KLrho{\bar{\zeta}^{(l)}}{\bar{\zeta}^{(l-1)}}  + \frac{550\eta}{(1-\gamma)^2}\lambda_{0, t+1}\\
	&= \frac{6250\eta\cC_\rho}{(1-\gamma)^3}\sum_{l=0}^{t-1}\lambda_{l+1,t+1}\KLrho{\bar{\zeta}^{(l+1)}}{\bar{\zeta}^{(l)}}  + \frac{550\eta}{(1-\gamma)^2}\lambda_{0, t+1},
\end{align*}
where the penultimate step invokes Lemma \ref{lemma:cabbage}.


\subsection{Proof of Lemma \ref{lemma:bar_zeta_2_zeta}}
\label{sec:pf_lemma_bar_zeta_2_zeta}
Taking logarithm on the both sides of the update rule \eqref{eq:update_bar}, we get
\begin{equation}
	\begin{cases}
		\log \bmutp(s) - (1-\eta\tau)\log \mut(s) &\overset{\mathbf{1}}{=} \eta \Q{t}(s)\bnut(s)\\
		\log \bar{\nu}^{(t+1)}(s) - (1-\eta\tau)\log \nut(s) &\overset{\mathbf{1}}{=} -\eta \Q{t}(s)^\top\bmut(s)
	\end{cases},
	\label{eq:log_bar_update}
\end{equation}
where we recall the notation in \eqref{eq:equiv_one}.

Subtracting \eqref{eq:log_opt} from \eqref{eq:log_bar_update} and taking inner product with $\bztp(s) - \best{\zeta}(s)$ gives
\begin{align*}
	&\innprod{\log \bztp(s) - (1-\eta\tau)\log \zt(s) - \eta\tau \log \best{\zeta}(s), \bztp(s) - \best{\zeta}(s)}\\
	&= \eta\innprod{\bmutp(s) - \best{\mu}(s), \Q{t}(s)\bnut(s) - \best{Q}(s)\best{\nu}(s)}\\
	&\qquad - \eta \innprod{\bnutp(s) - \best{\nu}(s), \Q{t}(s)^\top\bmut(s) - \best{Q}(s)^\top \best{\mu}(s)}\\
	&\le \eta\innprod{\bmutp(s) - \best{\mu}(s), \Q{t}(s)\big(\bnut(s) - \best{\nu}(s)\big)}\\
	&\qquad - \eta \innprod{\bnutp(s) - \best{\nu}(s), \Q{t}(s)^\top\big(\bmut(s) - \best{\mu}(s)\big)} + 2\eta \normbig{\Q{t}(s) - \best{Q}(s)}_\infty \\
	&\le \eta\innprod{\bmutp(s) - \best{\mu}(s), \Q{t}(s)\big(\bnut(s) - \bnutp(s)\big)}\\
	&\qquad - \eta \innprod{\bnutp(s) - \best{\nu}(s), \Q{t}(s)^\top\big(\bmut(s) - \bmutp(s)\big)} + 2\eta \normbig{\Q{t}(s) - \best{Q}(s)}_\infty \\
	&\le \frac{2\eta}{1-\gamma}\Big(2\KLs{\best{\zeta}}{\bztp} + \KLs{\bztp}{\zt} + \KLs{\zt}{\bzt}\Big) + 2\eta \normbig{\Q{t}(s) - \best{Q}(s)}_\infty.
\end{align*}
LHS can be written as
\begin{align*}
	&\innprod{\log \bztp(s) - (1-\eta\tau)\log \zt(s) - \eta\tau \log \best{\zeta}(s), \bztp(s) - \best{\zeta}(s)}\\
	&=- \innprod{\log \bztp(s) - (1-\eta\tau)\log \zt(s) - \eta\tau \log \best{\zeta}(s), \best{\zeta}(s)}\\
	&\qquad + \innprod{\log \bztp(s) - (1-\eta\tau)\log {\zeta}^{(t)}(s) - \eta\tau \log \best{\zeta}(s), \bztp(s)}\\
	&= \KLs{\best{\zeta}}{\bztp}- (1-\eta\tau)\KLs{\best{\zeta}}{\zt} \\
	&\qquad + (1-\eta\tau)\KLs{\bztp}{\zt} + \eta\tau \KLs{\bztp}{\best{\zeta}}.
\end{align*}
So we conclude that
\begin{align*}
	&\Big(1-\frac{4\eta}{1-\gamma}\Big)\KLs{\best{\zeta}}{\bztp}- (1-\eta\tau)\KLs{\best{\zeta}}{\zt} \\
	&\qquad + \Big(1-\eta\tau - \frac{2\eta}{1-\gamma}\Big)\KLs{\bztp}{\zt} + \eta\tau \KLs{\bztp}{\best{\zeta}}\\
	&\le \frac{2\eta}{1-\gamma}\KLs{\zt}{\bzt} + 2\eta \normbig{\Q{t}(s) - \best{Q}(s)}_\infty.
\end{align*}
With $0 < \eta \le \frac{1-\gamma}{8}$, we have
\begin{align*}
	&\frac{1}{2}\KLs{\best{\zeta}}{\bztp} + \eta\tau \KLs{\bztp}{\best{\zeta}}\\
	&\le (1-\eta\tau)\KLs{\best{\zeta}}{\zt} + \frac{2\eta}{1-\gamma}\KLs{\zt}{\bzt} + 2\eta \normbig{\Q{t}(s) - \best{Q}(s)}_\infty.
\end{align*}

\subsection{Proof of Lemma \ref{lemma:dualgap_markov_2_matrix}}
\label{sec:pf_lemma_dualgap_markov_2_matrix}
	By definition of value function $V_\tau$, we have
	\begin{align*}
		V_\tau^{\mu,\nu}(s) - \best{V}(s) &= \mu(s)^\top Q_{\tau}^{\mu, \nu}(s) \nu(s) + \tau \ent{\mu(s)} - \tau \ent{\nu(s)}\\
		&\qquad - \best{\mu}(s)^\top \best{Q}(s) \best{\nu}(s) - \tau \ent{\best{\mu}(s)} + \tau \ent{\best{\nu}(s)}\\
		&=\mu(s)^\top Q_{\tau}^{\mu, \nu}(s) \nu(s) - \mu(s)^\top \best{Q}(s) \nu(s) + f_s(\best{Q}, \mu, \nu) - f_s(\best{Q}, \best{\mu}, \best{\nu})\\
		&= \gamma \exlim{\substack{a\sim \mu(\cdot|s),b\sim \nu(\cdot|s),\\s'\sim P(\cdot|s,a,b)}}{V_\tau^{\mu, \nu}(s') - \best{V}(s')}+ f_s(\best{Q}, \mu, \nu) - f_s(\best{Q}, \best{\mu}, \best{\nu}).
	\end{align*}
	Applying the relation recursively and averaging $s$ over $\rho$, we arrive at
	\begin{align}
				V_\tau^{\mu,\nu}(\rho) - \best{V}(\rho)&= \frac{1}{1-\gamma} \exlim{s'\sim d_\rho^{\mu, \nu}}{f_{s'}(\best{Q}, \mu, \nu) - f_{s'}(\best{Q}, \best{\mu}, \best{\nu})},
		\label{eq:performance_diff_lemma_game}
	\end{align}
	which is the well-known performance difference lemma applied to the setting of Markov games. It follows that
	\begin{align}
		V_\tau^{\mu_\tau^\dagger(\nu),\nu}(\rho) - \best{V}(\rho)& = \frac{1}{1-\gamma} \exlim{s'\sim d_\rho^{\mu_\tau^\dagger(\nu), \nu}}{f_{s'}(\best{Q}, \mu_\tau^\dagger(\nu), \nu) - f_{s'}(\best{Q}, \best{\mu}, \best{\nu})}\nonumber\\
		&\le \frac{1}{1-\gamma} \exlim{s'\sim d_\rho^{\mu_\tau^\dagger(\nu), \nu}}{f_{s'}(\best{Q}, \mu_\tau^\dagger(\nu), \nu) - f_{s'}(\best{Q}, \mu, \best{\nu})}\nonumber\\
		&\le  \frac{1}{1-\gamma} \exlim{s'\sim d_\rho^{\mu_\tau^\dagger(\nu), \nu}}{\max_{\mu', \nu'}\Big(f_{s'}(\best{Q}, \mu', \nu) - f_{s'}(\best{Q}, \mu, \nu')\Big)}\label{eq:perf_diff_crashed}\\
		&\le \frac{\cC_{\rho,\tau}^\dagger}{1-\gamma} \exlim{s\sim\rho}{\max_{\mu', \nu'}\Big(f_{s}(\best{Q}, \mu', \nu) - f_{s}(\best{Q}, \mu, \nu')\Big)}.\nonumber
	\end{align}
A similar argument gives $\best{V}(\rho) - V_\tau^{\mu,\nu_\tau^\dagger(\mu)}(\rho) \le \frac{\cC_{\rho,\tau}^\dagger}{1-\gamma} \exlim{s\sim\rho}{\max_{\mu', \nu'}\Big(f_{s}(\best{Q}, \mu', \nu) - f_{s}(\best{Q}, \mu, \nu')\Big)}$. Summing the two inequalities proves \eqref{eq:dualgap_markov_2_matrix_rho}. Alternatively, we continue from \eqref{eq:perf_diff_crashed} and show that
\begin{align*}
		V_\tau^{\mu_\tau^\dagger(\nu),\nu}(s) - \best{V}(s)&\le  \frac{1}{1-\gamma} \exlim{s'\sim d_s^{\mu_\tau^\dagger(\nu), \nu}}{\max_{\mu', \nu'}\Big(f_{s'}(\best{Q}, \mu', \nu) - f_{s'}(\best{Q}, \mu, \nu')\Big)} \\
		&\le \frac{\|1/\rho\|_\infty}{1-\gamma} \exlim{s\sim\rho}{\max_{\mu', \nu'}\Big(f_{s}(\best{Q}, \mu', \nu) - f_{s}(\best{Q}, \mu, \nu')\Big)}. 
	\end{align*}
Summing the inequality with the one for $\best{V}(s) - V_\tau^{\mu,\nu_\tau^\dagger(\mu)}(s)$ and taking maximum over $s\in\cS$ completes the proof for \eqref{eq:dualgap_markov_2_matrix_max}.

\section{Proof of key lemmas for the finite-horizon setting}
\subsection{Proof of Lemma \ref{lemma:Hamlet}}
\label{sec:pf_lemma_Hamlet}

Following similar arguments of arriving \eqref{eq:core_step_1}, we have
\begin{align*}
	&\innprod{\log \ztp_h(s) - (1-\eta\tau)\log \zt_h(s) - \eta\tau \log \besth{\zeta}(s), \bztp_h(s) - \besth{\zeta}(s)}\\
	&\le 2\eta \normbig{\Q{t+1}_h(s) - \besth{Q}(s)}_\infty.
\end{align*}
We rewrite the LHS as
\begin{align*}
	&\innprod{\log \ztp_h(s) - (1-\eta\tau)\log \zt_h(s) - \eta\tau \log \besth{\zeta}(s), \bztp_h(s) - \besth{\zeta}(s)}\\
	&=- \innprod{\log \ztp_h(s) - (1-\eta\tau)\log \zt_h(s) - \eta\tau \log \besth{\zeta}(s), \besth{\zeta}(s)}\\
	&\qquad + \innprod{\log \bztp_h(s) - (1-\eta\tau)\log \zt_h(s) - \eta\tau \log \besth{\zeta}(s), \bztp_h(s)}\\
	&\qquad + \innprod{\log \ztp_h(s) - \log \bztp_h(s), \bztp_h(s)}\\
	&= \KLs{\besth{\zeta}}{\ztp_h}- (1-\eta\tau)\KLs{\besth{\zeta}}{\zt_h} \\
	&\qquad + (1-\eta\tau)\KLs{\bztp_h}{\zt_h} + \eta\tau \KLs{\bztp_h}{\besth{\zeta}}\\
	&\qquad + \KLs{\ztp_h}{\bztp_h} - \innprod{\log \bztp_h(s)-\log \ztp_h(s),  \bztp_h(s)- \ztp_h(s)}.
\end{align*}
Rearranging terms gives
\begin{align} \label{eq:core_step_1_epi}
		&\KLs{\besth{\zeta}}{\ztp_h} - (1-\eta\tau)\KLs{\besth{\zeta}}{\zt_h} +(1-\eta\tau)\KLs{\bztp_h}{\zt_h} \nonumber \\
		&\qquad + \eta\tau \KLs{\bztp_h}{\besth{\zeta}} + \KLs{\ztp_h}{\bztp_h} \nonumber\\
		&\qquad - \innprod{\log \bztp_h(s)-\log \ztp_h(s),  \bztp_h(s)- \ztp_h(s)} \nonumber\\
		&\le 2\eta \normbig{\Q{t+1}(s) - \best{Q}(s)}_\infty. 
\end{align}
Note that 
\begin{align}\label{eq:core_step_2_epi} 
		&\innprod{\log \bmutp_h(s)-\log \mutp_h(s),  \bmutp_h(s)- \mutp_h(s)} \nonumber \\
		&= \eta \innprod{\Q{t}_h(s)\bnut_h(s) - \Q{t+1}_h(s)\bnutp_h(s), \bmutp_h(s) - \mutp_h(s)} \nonumber\\
		&\le \eta \normbig{\Q{t}_h(s)\bnut_h(s) - \Q{t+1}_h(s)\bnutp_h(s)}_1\normbig{\bmutp_h(s) - \mutp_h(s)}_1.	
	\end{align}
We bound $\normbig{\Q{t}_h(s)\bnut_h(s) - \Q{t+1}_h(s)\bnutp_h(s)}_1$ as
\begin{align*}
	&\normbig{\Q{t}_h(s)\bnut_h(s) - \Q{t+1}_h(s)\bnutp_h(s)}_1 \\
	&\le \normbig{\Q{t+1}_h(s)\big(\bnut_h(s) - \bnutp_h(s)\big)}_1 + \normbig{\big(\Q{t}_h(s)- \Q{t+1}_h(s)\big)\bnut_h(s)}_1\\
	&\le 2H\normbig{\bnut_h(s) - \bnutp_h(s)}_1 + \normbig{\Q{t}_h(s)- \Q{t+1}_h(s)}_\infty\\
	&\le 2H\normbig{\bnutp_h(s) - \nut_h(s)}_1 + 2H\normbig{\nut_h(s) - \bnut_h(s)}_1 + \normbig{\Q{t}_h(s)- \Q{t+1}_h(s)}_\infty.
\end{align*}
Plugging the above inequality into \eqref{eq:core_step_2_epi} and invoking Young's inequality yields
\begin{align*}
	&\innprod{\log \bmutp_h(s)-\log \mutp_h(s),  \bmutp_h(s)- \mutp_h(s)}\\
	&\le \eta H\Big(\normbig{\bnutp_h(s) - \nut_h(s)}_1^2 + \normbig{\nut_h(s) - \bnut_h(s)}_1^2 + 2\normbig{\bmutp_h(s) - \mutp_h(s)}_1^2\Big)\\
	&\qquad + \eta \normbig{\Q{t}_h(s)- \Q{t+1}_h(s)}_\infty\normbig{\bmutp_h(s) - \mutp_h(s)}_1\\
	&\le 2\eta H\KLs{\bnutp_h}{\nut_h} + 2\eta H\KLs{\nut_h}{\bnut_h} + 4\eta H\KLs{\mutp_h}{\bmutp_h}  + 2\eta^2 H \normbig{\Q{t}_h(s) - \Q{t+1}_h(s)}_\infty,
\end{align*}
where the last step results from Pinsker's inequality and Lemma \ref{lemma:one_step_policy_bound_epi}. Similarly, we have
\begin{align*}
	&\innprod{\log \bnutp_h(s)-\log \nutp_h(s),  \bnutp_h(s)- \nutp_h(s)}\\
	&\le 2\eta H\KLs{\bmutp_h}{\mut_h} + 2\eta H\KLs{\mut_h}{\bmut_h} + 4\eta H\KLs{\nutp_h}{\bnutp_h}   + 2\eta^2 H \normbig{\Q{t}_h(s) - \Q{t+1}_h(s)}_\infty.
\end{align*}
Summing the above two inequalities gives
\begin{align*}
	&\innprod{\log \bztp_h(s)-\log \ztp_h(s),  \bztp_h(s)- \ztp_h(s)}\\
	&\le 2\eta H\KLs{\bztp_h}{\zt_h} + 2\eta H\KLs{\zt_h}{\bzt_h} + 4\eta H\KLs{\ztp_h}{\bztp_h} \\
	&\qquad + 4\eta^2 H \normbig{\Q{t}_h(s) - \Q{t+1}_h(s)}_\infty\\
		&\le 2\eta H\KLs{\bztp_h}{\zt_h} + 2\eta H\KLs{\zt_h}{\bzt_h} + 4\eta H\KLs{\ztp_h}{\bztp_h} \\
	&\qquad + \frac{\eta}{2}\Big( \normbig{\Q{t}_h(s) - \besth{Q}(s)}_\infty + \normbig{\Q{t+1}_h(s) - \besth{Q}(s)}_\infty\Big),
\end{align*}
where the second step invokes triangular inequality and the fact that $\eta \le \frac{1}{8H}$. Plugging the above inequality into \eqref{eq:core_step_1_epi} gives
	\begin{align*}
		&\KLs{\besth{\zeta}}{\ztp_h} - (1-\eta\tau)\KLs{\besth{\zeta}}{\zt_h} +(1-\eta(\tau+2H))\KLs{\bztp_h}{\zt_h} \\
		&\qquad + \eta\tau \KLs{\bztp_h}{\besth{\zeta}} + (1-4\eta H)\KLs{\ztp_h}{\bztp_h} - 2\eta H\KLs{\zt_h}{\bzt_h}\\
		&\le \frac{5\eta}{2} \normbig{\Q{t+1}_h(s) - \best{Q}(s)}_\infty + \frac{\eta}{2} \normbig{\Q{t}_h(s) - \best{Q}(s)}_\infty.
	\end{align*}
With $\eta \le \frac{1}{8H}$, we have $(1-\eta\tau)(1-4\eta H) \ge 2\eta H$ and $1-\eta(\tau + 2H) \ge 0$. It follows that
\begin{align*}
	&\KLs{\besth{\zeta}}{\ztp_h} + (1-4\eta H)\KLs{\ztp_h}{\bztp_h} + \eta\tau \KLs{\bztp_h}{\besth{\zeta}}\\
	&\le (1-\eta\tau)\KLs{\besth{\zeta}}{\zt_h} + 2\eta H\KLs{\zt_h}{\bzt_h}\\
	&\qquad + \frac{5\eta}{2} \normbig{\Q{t+1}_h(s) - \best{Q}(s)}_\infty + \frac{\eta}{2} \normbig{\Q{t}_h(s) - \best{Q}(s)}_\infty\\
	&\le (1-\eta\tau)\Big(\KLs{\besth{\zeta}}{\zt_h} + (1-4\eta H)\KLs{\zt_h}{\bzt_h}\Big)\\
	&\qquad + \frac{5\eta}{2} \normbig{\Q{t+1}_h(s) - \best{Q}(s)}_\infty + \frac{\eta}{2} \normbig{\Q{t}_h(s) - \best{Q}(s)}_\infty.
\end{align*}
Therefore, it holds for $0 \le t_1 < t_2$ that
\begin{align*}
	&\KLs{\besth{\zeta}}{\zeta_h^{(t_2)}} + (1-4\eta H)\KLs{\zeta_h^{(t_2)}}{\bar{\zeta}_h^{(t_2)}} + \eta\tau \KLs{\bar{\zeta}_h^{(t_2)}}{\besth{\zeta}}\\
	&\le (1-\eta\tau)^{t_2-t_1}\Big(\KLs{\besth{\zeta}}{\zeta_h^{(t_1)}} + (1-4\eta H)\KLs{\zeta_h^{(t_1)}}{\bar{\zeta}_h^{t_1}}\Big)\\
	&\qquad +\sum_{t'=t_1+1}^{t_2} (1-\eta\tau)^{t_2-l}\Big[ \frac{5\eta}{2} \normbig{\Q{l}_h(s) - \best{Q}(s)}_\infty + \frac{\eta}{2} \normbig{\Q{l-1}_h(s) - \best{Q}(s)}_\infty\Big]\\
	&\le (1-\eta\tau)^{t_2-t_1}\Big(\KLs{\besth{\zeta}}{\zeta_h^{(t_1)}} + (1-4\eta H)\KLs{\zeta_h^{(t_1)}}{\bar{\zeta}_h^{(t_1)}}\Big)\\
	&\qquad +4\eta\sum_{l=t_1}^{t_2} (1-\eta\tau)^{t_2-l}  \normbig{\Q{l}_h(s) - \best{Q}(s)}_\infty.
\end{align*}

\subsection{Proof of Lemma \ref{lemma:Horatio}}
\label{sec:pf_lemma_Horatio}
For $t_2 > 0$, we have
\begin{align}
	&\Q{t_2}_{h-1}(s,a,b) - Q_{h-1,\tau}^\star(s,a,b) \nonumber\\
	&= \exlim{s'\sim P_{h-1}(\cdot|s,a,b)}{\V{t_2-1}_{h}(s') - \besth{V}(s')}\nonumber\\
	 &= \mathop{\mathbb{E}}_{s'\sim P_{h-1}(\cdot|s,a,b)}\Big[(1-\eta\tau)^{t_2-t_1}\big(\V{t_1-1}_{h}(s') - \besth{V}(s')\big) \nonumber\\
	 &\qquad + \eta\tau\sum_{l=t_1}^{t_2-1}(1-\eta\tau)^{t_2-1-l} \big(f_{s'}(\Q{t_1}, \bar{\mu}_h^{(t_1)}, \bar{\nu}_h^{(t_1)}) - f_{s'}(\besth{Q}, \besth{\mu}, \besth{\nu})\big)\Big]\nonumber\\
	 &\le (1-\eta\tau)^{t_2-t_1} 2H + \exlim{s'\sim P_{h-1}(\cdot|s,a,b)}{\eta\tau\sum_{l=t_1}^{t_2-1}(1-\eta\tau)^{t_2-1-l} \big(f_{s'}(\Q{l}_{h}, \bar{\mu}_h^{(l)}, \bar{\nu}_h^{(l)}) - f_{s'}(\besth{Q}, \besth{\mu}, \besth{\nu})\big)}.
	 \label{eq:Q_gap_epi}
\end{align}
We start by decomposing $f_s^{(t)} - f_{s}^\star$ as
\begin{align*}
	&f_s(\Q{t}_h, \bmut_h, \bnut_h) - f_{s}(\besth{Q}, \besth{\mu}, \besth{\nu})\\
	&= \prn{f_s(\Q{t}_h, \bmut_h, \bnut_h) - f_s(\Q{t}_h, \bmut_h, \besth{\nu})} + f_s(\Q{t}_h, \bmut_h, \besth{\nu}) - f_s(\besth{Q}, \besth{\mu}, \besth{\nu})\\
	&\le \prn{f_s(\Q{t}_h, \bmut_h, \bnut_h) - f_s(\Q{t}_h, \bmut_h, \besth{\nu})} + f_s(\best{Q}, \bmut, \besth{\nu})- f_s(\besth{Q}, \besth{\mu}, \besth{\nu})\\
	&\qquad + \normbig{\Q{t}_h(s) - \besth{Q}(s)}_\infty\\
	&\le f_s(\Q{t}_h, \bmut_h, \bnut_h) - f_s(\Q{t}_h, \bmut_h, \besth{\nu}) + \normbig{\Q{t}_h(s) - \besth{Q}(s)}_\infty.
\end{align*}
Note that Lemma \ref{lemma:one_step_regret} can be applied to the episodic setting by simply replacing $1/(1-\gamma)$ with $H$, which yields
\begin{align} 
	f_s(\Q{t}_h, \bmut_h, \bnut_h) - f_{s}(\besth{Q}, \besth{\mu}, \besth{\nu})
	&\le \normbig{\Q{t}_h(s) - \besth{Q}(s)}_\infty +  2\eta H\normbig{\Q{t}_h(s) - \Q{t-1}_h(s)}_\infty  \nonumber\\
	&\qquad+ \frac{1-\eta\tau}{\eta}\KLs{\besth{\nu}}{\nutm_h} - \frac{1}{\eta} \KLs{\besth{\nu}}{\nut_h} \nonumber\\
	&\qquad   - \frac{1}{\eta} \big(1-4\eta H\big) \KLs{\nut_h}{\bnut_h}- \frac{1-\eta\tau}{\eta}\KLs{\bnut_h}{\nutm_h} \nonumber \\
	&\qquad +2 H\Big(\KLs{\bmut_h}{\mutm_h} + \KLs{\mutm_h}{\bar{\mu}^{(t-1)}_h}\Big).  \label{eq:opt_gap_pos_epi}
\end{align}
By a similar argument,
\begin{align} 
	 f_{s}(\besth{Q}, \besth{\mu}, \besth{\nu}) - f_s(\Q{t}_h, \bmut_h, \bnut_h)  
	&\le \normbig{\Q{t}_h(s) - \besth{Q}(s)}_\infty +  2\eta H\normbig{\Q{t}_h(s) - \Q{t-1}_h(s)}_\infty \nonumber\\
	&\qquad+ \frac{1-\eta\tau}{\eta}\KLs{\besth{\mu}}{\mutm_h} - \frac{1}{\eta} \KLs{\besth{\mu}}{\mut_h} \nonumber\\
	&\qquad   - \frac{1}{\eta} \big(1-4\eta H\big) \KLs{\mut_h}{\bmut_h}- \frac{1-\eta\tau}{\eta}\KLs{\bmut_h}{\mutm_h} \nonumber\\
	&\qquad +2 H\Big(\KLs{\bnut_h}{\nutm_h} + \KLs{\nutm_h}{\bar{\nu}^{(t-1)}_h}\Big).
\label{eq:opt_gap_neg_epi}
\end{align}
Combining \eqref{eq:opt_gap_pos_epi} + $\frac{2}{3}\cdot$ \eqref{eq:opt_gap_neg_epi} gives
\begin{align}
	&\frac{1}{3}\big[f_s(\Q{t}_h, \bmut_h, \bnut_h) - f_{s}(\besth{Q}, \besth{\mu}, \besth{\nu})\big]\nonumber\\
	&\le \frac{5}{3}\big[\normbig{\Q{t}_h(s) - \besth{Q}(s)}_\infty +  2\eta H\normbig{\Q{t}_h(s) - \Q{t-1}_h(s)}_\infty \big]\nonumber\\
	&\qquad + \frac{1-\eta\tau}{\eta}\Big[\KLs{\besth{\nu}}{\nutm_h} + \frac{2}{3}\KLs{\besth{\mu}}{\mutm_h}\Big] - \frac{1}{\eta}\Big[\KLs{\besth{\nu}}{\nut_h} + \frac{2}{3}\KLs{\besth{\mu}}{\mut_h}\Big]\nonumber\\
	&\qquad + 2H\Big[\KLs{\mutm_h}{\bar{\mu}^{(t-1)}_h} + \frac{2}{3}\KLs{\nutm_h}{\bar{\nu}^{(t-1)}_h} \Big]\nonumber\\
	&\qquad - \frac{1}{\eta} \big(1-4\eta H\big)\Big[ \frac{2}{3} \KLs{\mut_h}{\bmut_h} + \KLs{\nut_h}{\bnut_h}\Big]\nonumber\\
	&\qquad +\Big(2H - \frac{1-\eta\tau}{\eta}\cdot\frac{2}{3}\Big)\KLs{\bmut}{\mutm} + \Big(2H \cdot \frac{2}{3} - \frac{1-\eta\tau}{\eta}\Big)\KLs{\bnut}{\nutm}.
	\label{eq:opt_gap_pos_mix_epi}
\end{align}
With $\eta \le \frac{1}{8H}$, we have
\[
2H - \frac{1-\eta\tau}{\eta}\cdot\frac{2}{3} \le 0,\quad 2H \cdot \frac{2}{3} - \frac{1-\eta\tau}{\eta}\le 0, \quad \text{and}\quad \frac{1}{\eta}(1-\eta\tau)(1-4\eta H) \cdot\frac{2}{3}\ge 2 H.
\]
Let 
\begin{align*}
G_h^{(t)}(s) &= \KLs{\besth{\nu}}{\nut_h} + \frac{2}{3}\KLs{\besth{\mu}}{\mut_h} + \frac{2}{3}(1-4\eta H)\Big[\KLs{\mut_h}{\bmut_h} + \KLs{\nut_h}{\bnut_h}\Big].
\end{align*}
We can simplify \eqref{eq:opt_gap_pos_mix_epi} as
\begin{align*}
	&f_s(\Q{t}_h, \bmut_h, \bnut_h) - f_{s}(\besth{Q}, \besth{\mu}, \besth{\nu})\\
	&\le 5\big[\normbig{\Q{t}_h(s) - \besth{Q}(s)}_\infty +  2\eta H\normbig{\Q{t}_h(s) - \Q{t-1}_h(s)}_\infty \big] +\frac{1-\eta\tau}{\eta}G_h^{(t-1)}(s) - \frac{1}{\eta}G_h^{(t)}(s).
\end{align*}
Plugging the above inequality into \eqref{eq:Q_gap_epi} gives
\begin{align*}
	&\Q{t_2}_{h-1}(s,a,b) - Q_{h-1,\tau}^\star(s,a,b) \\
	 &\le (1-\eta\tau)^{t_2-t_1} 2H\\
	 &\qquad + \exlim{s'\sim P_{h-1}(\cdot|s,a,b)}{5\eta\tau\sum_{l=t_1}^{t_2-1}(1-\eta\tau)^{t_2-1-l} \big(\normbig{\Q{l}_h(s') - \besth{Q}(s')}_\infty +  2\eta H\normbig{\Q{l}_h(s') - \Q{l-1}_h(s')}_\infty \big)}\\
	 &\qquad + \exlim{s'\sim P_{h-1}(\cdot|s,a,b)}{\tau (1-\eta\tau)^{t_2-t_1}G_h^{(t_1-1)}(s')}\\
	 &\le (1-\eta\tau)^{t_2-t_1} 2H\\
	 &\qquad + 10\eta\tau\exlim{s'\sim P_{h-1}(\cdot|s,a,b)}{\sum_{l=t_1-1}^{t_2-1}(1-\eta\tau)^{t_2-1-l} \normbig{\Q{l}_h(s') - \besth{Q}(s')}_\infty }\\
	 &\qquad + \tau (1-\eta\tau)^{t_2-t_1}\exlim{s'\sim P_{h-1}(\cdot|s,a,b)}{\mathsf{KL}_{s'}\big(\besth{\zeta}\,\|\,\zeta^{(t_1-1)}_h\big)  + (1-4\eta H)\mathsf{KL}_{s'}\big(\zeta^{(t_1-1)}_h\,\|\,\bar{\zeta}^{(t_1-1)}_h\big)}.
\end{align*}
The other side of Lemma \ref{lemma:Horatio} can be shown with a similar proof and is therefore omitted.

\section{Proof of auxiliary lemmas}

\subsection{Proof of Lemma \ref{lemma:one_step_policy_bound}}
\label{sec:pf_lemma_one_step_policy_bound}
We first single out a set of bounds for $\V{t}$ and $\Q{t}$, which can be obtained by a simple induction:
\begin{equation}	\label{eq:value_bound}
\forall (s,a,b) \in \cS\times\cA\times\cB,\qquad
	\begin{cases}
		-\frac{\tau \log |\cB|}{1-\gamma} \le V^{(t)}(s) \le \frac{1 + \tau \log |\cA|}{1-\gamma}\\
		-\frac{\gamma \tau  \log |\cB|}{1-\gamma}\le Q^{(t)}(s,a,b) \le \frac{1 + \gamma\tau\log |\cA|}{1-\gamma}
	\end{cases}.
\end{equation}
We invoke the following lemma to bound several key quantities that will be helpful in the analysis.
\begin{lemma}[{\cite[Lemma 24]{mei2020global}}]
	Let $\pi, \pi' \in \Delta(\cA)$ such that $\pi(a)\propto \exp(\theta(a))$, $\pi'(a) \propto \exp(\theta'(a))$ for some $\theta, \theta' \in \mathbb{R}^{|\cA|}$. It holds that
	\begin{equation*}
		\normbig{\pi - \pi'}_1 \le \normbig{\theta - \theta'}_\infty.
	\end{equation*}
\end{lemma}
With this lemma in mind, for any $t \ge 0$, it follows that
\begin{align*}
	\normbig{\bmutp(s) - \mutp(s)}_1 &\le \min_{c\in \mathbb{R}}\normbig{\log\bmutp(s) - \log\mutp(s) - c\cdot\one}_\infty\\
	&\le \eta \normbig{\Q{t}(s)\bnut(s) - \Q{t+1}(s)\bnutp(s)}_\infty\\
	&\le \eta\cdot \frac{1+\gamma\tau(\log|\cA| + \log|\cB|)}{1-\gamma}\le \frac{2\eta}{1-\gamma},
\end{align*}
where the second line follows from the update rule \eqref{eq:update_whole}, and the last line follows from \eqref{eq:value_bound}. A similar argument reveals that
\begin{equation*}
	\normbig{\bnutp(s) - \nu^{(t+1)}(s)}_1 \le \frac{2\eta}{1-\gamma},
\end{equation*}
which completes the proof of \eqref{eq:diff_bar_unbar}.

Moving onto the second claim \eqref{eq:diff_bar}, we make note of the fact that when $t \ge 1$,  
\begin{align} 
		\bmutp(a|s) &\propto \mut(a|s)^{1-\eta\tau} \exp(\eta [ \Q{t}(s)\bnut(s)]_a ) \nonumber \\
		&\overset{\mathrm{(i)}}{\propto} \bmut(a|s)^{1-\eta\tau} \exp\Big(\eta \big[\Q{t}(s)\bnut(s)+(1-\eta\tau)(\Q{t}(s)\bnut(s) - \Q{t-1}(s)\bar{\nu}^{(t-1)}(s))\big]_a \Big) \nonumber \\
		&\propto \bmut(a|s) \exp(\eta w^{(t)}(a)),
	\label{eq:mu_update_single_ver}
\end{align}
where
\begin{equation*}
	w^{(t)} = \Q{t}(s)\bnut(s)+(1-\eta\tau)\big(\Q{t}(s)\bnut(s) - \Q{t-1}(s)\bar{\nu}^{(t-1)}(s)\big) - \tau \log \bmut(s).
\end{equation*}
Here, (i) follows from the update rule \eqref{eq:update_whole} as
\begin{align*}
\mut(a|s) & \propto \mutm(a|s)^{1-\eta\tau} \exp(\eta [\Q{t}(s)\bnut(s)]_a) \\
& \propto \mutm(a|s)^{1-\eta\tau} \exp(\eta [\Q{t-1}(s)\bar{\nu}^{(t-1)}(s)]_a) \exp(\eta [\Q{t}(s)\bnut(s)- \Q{t-1}(s)\bar{\nu}^{(t-1)}(s)]_a) \\
& \propto \bmut(a|s) \exp(\eta [\Q{t}(s)\bnut(s)- \Q{t-1}(s)\bar{\nu}^{(t-1)}(s)]_a).
\end{align*}
Moreover, $w^{(t)}$ satisfies 
\begin{align*}
\normbig{w^{(t)}}_\infty 
&\le \normbig{\Q{t}(s)\bnut(s)}_\infty + \normbig{\tau \log \bmut(s)}_\infty + (1-\eta\tau)\normbig{\Q{t}(s)\bnut(s) - \Q{t-1}(s)\bar{\nu}^{(t-1)}(s)}_\infty\\
&\le \frac{2}{1-\gamma} + \frac{2}{1-\gamma} + \frac{2(1-\eta\tau)}{1-\gamma}\le \frac{6}{1-\gamma}.
\end{align*}
Here, the second step is due to \eqref{eq:log_bound}, which we shall prove momentarily. 
Recall that when $t = 0$, we have $\bmutp = \bar{\mu}^{(0)}$. In sum, we have
\begin{equation*}
\forall s\in\cS, t \ge 0,\qquad \normbig{\bmutp(s) - \bmut(s)}_1 \le \frac{6\eta}{1-\gamma},
\end{equation*}
concluding the proof of \eqref{eq:diff_bar}.

It remains to prove the claim \eqref{eq:log_bound}. For simplicity we focus on the bound with $\normbig{\log \mut (s)}_{\infty}$; the other bounds follow similarly.
	It is worth noting that $\mut(s)$ can be always written as $\mut(a|s) \propto \exp(w^{(t)}(a)/\tau)$ for some $w^{(t)} \in \mathbb{R}^{|\cA|}$ satisfying 
	\begin{equation*}
	\forall a\in\cA,\qquad 
		-\frac{\gamma \tau  \log |\cB|}{1-\gamma}\le w^{(t)}(a) \le \frac{1 + \gamma\tau\log |\cA|}{1-\gamma}.
	\end{equation*}
	To see this, note that the claim trivially holds for $t = 0$ with $w^{(0)} = \mathbf{0}$. When the statement holds for some $t \ge 0$, we have
	\begin{align*}
		\mutp(a|s) &\propto \mut(a|s)^{1-\eta\tau} \exp(\eta \Q{t+1}(s)\bnutp(s))\\
		&\propto \exp\big( ((1-\eta\tau) w^{(t)} + \eta\tau \Q{t+1}(s)\bnutp(s))/\tau \big)\\
		&\propto \exp\big( w^{(t+1)}/\tau \big),
	\end{align*}
	with $w^{(t+1)} = (1-\eta\tau) w^{(t)} + \eta\tau \Q{t+1}(s)\bnutp(s)$. We conclude that the claim holds for $t+1$ by recalling \eqref{eq:value_bound}.
	It then follows straightforwardly that
	\begin{equation*}
		\frac{\mut(a_1|s)}{\mut(a_2|s)} = \exp\Big(\frac{w^{(t)}(a_1)-w^{(t)}(a_2)}{\tau}\Big) \le \exp\Big(\frac{1+\gamma\tau(\log|\cA| + \log|\cB|)}{(1-\gamma)\tau}\Big) 
	\end{equation*}
	for any $a_1, a_2 \in \cA$. This allows us to show that
	\begin{equation*}
		\min_{a\in\cA} \mut (a|s) \ge \frac{1}{|\cA|\exp\big(\frac{1+\gamma\tau(\log|\cA| + \log|\cB|)}{(1-\gamma)\tau}\big)} \sum_{a\in\cA} \mut (a|s) =  \frac{1}{|\cA|\exp\big(\frac{1+\gamma\tau(\log|\cA| + \log|\cB|)}{(1-\gamma)\tau}\big)},
	\end{equation*}
	which gives
	\begin{align*}
		\|\log \mut(s)\|_\infty &\le \frac{1+\gamma\tau(\log|\cA| + \log|\cB|)}{(1-\gamma)\tau} + \log |\cA|\le \frac{1}{(1-\gamma)\tau} + \frac{\log|\cA| + \gamma\log|\cB|}{1-\gamma}\\
		&\le \frac{2}{(1-\gamma)\tau}.
	\end{align*}

\subsection{Proof of Lemma \ref{lemma:f_diff}}
\label{sec:pf_lemma_f_diff}
We decompose the term $f_s(\Q{t+1}, \bmutp, \bnutp) - f_s(\Q{t}, \bmut, \bnut)$ as follows:
\begin{align*}
	&f_s(\Q{t+1}, \bmutp, \bnutp) - f_s(\Q{t}, \bmut, \bnut)\\
	&= 	f_s(\Q{t+1}, \bmutp, \bnutp) - f_s(\Q{t}, \bmutp, \bnutp) + 	f_s(\Q{t}, \bmutp, \bnutp) - f_s(\Q{t}, \bmut, \bnut)\\
	&= \bmutp(s)^\top\Big(\Q{t+1}(s) - \Q{t}(s)\Big)\bnutp(s) \\
	&\qquad + f_s(\Q{t}, \bmutp, \bnut) - f_s(\Q{t}, \bmut, \bnut) + f_s(\Q{t}, \bmut, \bnutp) - f_s(\Q{t}, \bmut, \bnut)\\
	&\qquad + \Big[f_s(\Q{t}, \bmutp, \bnutp) + f_s(\Q{t}, \bmut, \bnut)-f_s(\Q{t}, \bmutp, \bnut) - f_s(\Q{t}, \bmut, \bnutp)\Big].
\end{align*}
Note that $\big|\bmutp(s)^\top\big(\Q{t+1}(s) - \Q{t}(s)\big)\bnutp(s)\big| \le \normbig{\Q{t+1}(s) - \Q{t}(s)}_\infty$. For the terms in the bracket, we have
\begin{align*}
	&\Big|\Big[f_s(\Q{t}, \bmutp, \bnutp) + f_s(\Q{t}, \bmut, \bnut)-f_s(\Q{t}, \bmutp, \bnut) - f_s(\Q{t}, \bmut, \bnutp)\Big]\Big|\\
	&= \Big|\big(\bmutp(s) - \bmut(s)\big)^\top \Q{t}(s)\big(\bnutp(s) - \bnut(s)\big)\Big|\\
	&\le \frac{2}{1-\gamma}\KLs{\bztp}{\bzt},
\end{align*}
where the last step invokes Cauchy-Schwarz inequality and Pinsker's inequality (see e.g., \eqref{eq:pinsker_trick_example}).
It remains to bound the two difference terms $\big|f_s(\Q{t}, \bmutp, \bnut) - f_s(\Q{t}, \bmut, \bnut)\big|$ and $\big|f_s(\Q{t}, \bmut, \bnutp) - f_s(\Q{t}, \bmut, \bnut)\big|$.
To proceed, we show that
\begin{align}
	&f_s(\Q{t}, \bmut, \bnut) - f_s(\Q{t}, \bmutp, \bnut)\nonumber\\
	&= \innprod{\bmut(s) - \bnutp(s), \Q{t}(s)^\top\bmut(s)} + \tau \cH(\bmut(s)) - \tau \cH(\bmutp(s))\nonumber\\
	&= \innprod{\bmut(s) - \bmutp(s),  \Q{t}(s)^\top\bnut(s)+ (1-\eta\tau)\big(\Q{t}(s)\bnut(s) - \Q{t-1}(s)\bar{\nu}^{(t-1)}(s)\big)}\nonumber\\
	&\qquad  + \tau \cH(\bmut(s)) - \tau \cH(\bmutp(s))\nonumber\\
	&\qquad - (1-\eta\tau)\innprod{\bmut(s) - \bmutp(s), \Q{t}(s)\bnut(s) - \Q{t-1}(s)\bar{\nu}^{(t-1)}(s)}\nonumber\\
	&=  - \frac{1}{\eta} \KLs{\bmut}{\bmutp}- \frac{1-\eta\tau}{\eta}\KLs{\bmutp}{\bmut}\nonumber\\
	&\qquad - (1-\eta\tau)\innprod{\bmut(s) - \bmutp(s), \Q{t}(s)\bnut(s) - \Q{t-1}(s)\bar{\nu}^{(t-1)}(s)} .
	\label{eq:f_mu_diff_tmp}
\end{align}
Here, the third step results from the special case of the following three-point lemma---which is proven in Appendix~\ref{sec:pf_three_pt}---in view of \eqref{eq:mu_update_single_ver}. 
\begin{lemma}[Regularized three-point lemma]
\label{lemma:strange_three_pt}
Let $x \in \Delta(\cA)$ be defined as
\begin{equation*}
	x(a) \propto y(a)^{1-\eta\tau} \exp(-\eta w(a))
\end{equation*}
for some $w \in \mathbb{R}^{|\cA|}$ and $y\in \Delta(\cA)$. It holds for all $z \in \Delta(\cA)$ that
\begin{equation*}
\frac{\eta}{1-\eta\tau}\Big[\innprod{x-z, w} - \tau \cH(x)+ \tau\cH(z)\Big] = \KL{z}{y} - \frac{1}{1-\eta\tau}\KL{z}{x} - \KL{x}{y}.
\end{equation*}
This immediately implies that
\begin{equation*}
\frac{\eta}{1-\eta\tau}\Big[\innprod{x-y, w} - \tau \cH(x)+ \tau\cH(y)\Big] =  - \frac{1}{1-\eta\tau}\KL{y}{x} - \KL{x}{y}.
\end{equation*}
\end{lemma}

Recall from the earlier discussion (cf. \eqref{eq:mu_update_single_ver}) that $\bmutp(a|s) \propto \bmut(a|s) \exp(\eta w^{(t)}(s))$ for some $w^{(t)} \in \mathbb{R}^{|\cB|}$ satisfying 
\[
\normbig{w^{(t)}}_\infty \le \frac{6}{1-\gamma}.
\]
We can ensure that $\|\eta w^{(t)}\|_\infty \le 1/30$ as long as $\eta^{-1} \ge \frac{180}{1-\gamma}$, and the next lemma guarantees $\KLs{\bmut}{\bmutp} \le 2 \KLs{\bmutp}{\bmut}$ in this case.
\begin{lemma}
\label{lemma:KL_local_equiv}
Let $w \in \mathbb{R}^{|\cA|}$, $\pi, \pi' \in \Delta(\cA)$ satisfy, for each $a \in \cA$, $\pi'(a) \propto \pi(a) \exp(w(a))$ with $\|w\|_\infty \le \frac{1}{30}$. It holds that
\[
\KL{\pi}{\pi'} \le 2 \KL{\pi'}{\pi}.
\]
\end{lemma}
Therefore, we can continue to bound \eqref{eq:f_mu_diff_tmp} by  
\begin{align*}
	&\big|f_s(\Q{t}, \bmutp, \bnut) - f_s(\Q{t}, \bmut, \bnut)\big|\\
	&\le  \frac{1}{\eta} \KLs{\bmut}{\bmutp} + \frac{1-\eta\tau}{\eta}\KLs{\bmutp}{\bmut}\\
	&\qquad + \normbig{\bmutp(s) - \bmut(s)}_1\normbig{\Q{t}(s)\bnut(s) - \Q{t-1}(s)\bar{\nu}^{(t-1)}(s)}_\infty\\
	&\le \frac{3}{\eta}\KLs{\bmutp}{\bmut} + \normbig{\bmutp(s) - \bmut(s)}_1\normbig{\Q{t}(s) - \Q{t-1}(s)}_\infty\\
	&\qquad + \normbig{\Q{t}(s)}_\infty\normbig{\bmutp(s) - \bmut(s)}_1 \normbig{\bnut(s) - \bar{\nu}^{(t-1)}(s)}_1\\
	&\le \Big(\frac{3}{\eta} + \frac{2}{1-\gamma}\Big)\KLs{\bmutp}{\bmut}  + \frac{2}{1-\gamma}\KLs{\bmut}{\bar{\mu}^{(t-1)}} + \frac{6\eta}{1-\gamma}\normbig{\Q{t}(s) - \Q{t-1}(s)}_\infty,
\end{align*}
where the last line uses Lemma~\ref{lemma:one_step_policy_bound}, Cauchy-Schwarz inequality and Pinsker's inequality (see e.g., \eqref{eq:pinsker_trick_example}). One can bound $\big|f_s(\Q{t}, \bmut, \bnut) - f_s(\Q{t}, \bmut, \bnutp)\big|$ with similar arguments. Putting all pieces together, we arrive at
\begin{align*}
	&\big|f_s(\Q{t+1}, \bmutp, \bnutp) - f_s(\Q{t}, \bmut, \bnut)\big|\\
	&\le \Big\|\Q{t+1}(s) - \Q{t}(s)\Big\|_\infty + \Big(\frac{3}{\eta} + \frac{4}{1-\gamma}\Big)\KLs{\bztp}{\bzt}  + \frac{2}{1-\gamma}\KLs{\bzt}{\bar{\zeta}^{(t-1)}} \\
	&\qquad + \frac{12\eta}{1-\gamma} \normbig{\Q{t}(s) - \Q{t-1}(s)}_\infty.
\end{align*}

\subsection{Proof of Lemma \ref{lemma:one_step_regret}}
\label{sec:pf_lemma_one_step_regret}
Note that
\begin{align} 
	&f_s(\Q{t}, \bmut, \bnut) - f_s(\Q{t}, \bmut, \nu) \nonumber \\ 
		& = \innprod{\bnut(s) - \best{\nu}(s), \Q{t}(s)^\top\bmut(s)} - \tau \cH(\bnut(s))+ \tau \cH(\best{\nu}(s))\nonumber \\
		&=  \innprod{\bnut(s) - \nut(s),\Q{t}(s)^\top\bmut(s) - \Q{t-1}(s)^\top\bar{\mu}^{(t-1)}(s)} \nonumber\\
		&\qquad + \innprod{\bnut(s) - \nut(s), \Q{t-1}(s)^\top\bar{\mu}^{(t-1)}(s)}  - \tau \cH(\bnut(s))+ \tau \cH(\nut(s))\nonumber\\
		&\qquad + \innprod{{\nu}^{(t)}(s) - \best{\nu}(s), \Q{t}(s)^\top\bmut(s)}- \tau \cH(\nut(s))+ \tau \cH(\best{\nu}(s)) \nonumber\\
		&=  \innprod{\bnut(s) - \nut(s),\Q{t}(s)^\top\bmut(s) - \Q{t-1}(s)^\top\bar{\mu}^{(t-1)}(s)} \nonumber\\
		&\qquad  + \frac{1-\eta\tau}{\eta}\KLs{\nut}{\nutm} - \frac{1}{\eta} \KLs{\nut}{\bnut}- \frac{1-\eta\tau}{\eta}\KLs{\bnut}{\nutm} \nonumber\\
		&\qquad + \frac{1-\eta\tau}{\eta}\KLs{\best{\nu}}{\nutm} - \frac{1}{\eta} \KLs{\best{\nu}}{\nut}- \frac{1-\eta\tau}{\eta}\KLs{\nut}{\nutm} \nonumber \\
		&\le \normbig{\bnut(s) - \nut(s)}_1\normbig{\Q{t}(s)^\top\bmut(s) - \Q{t-1}(s)^\top\bar{\mu}^{(t-1)}(s)}_\infty \nonumber\\
		&\qquad   - \frac{1}{\eta} \KLs{\nut}{\bnut}- \frac{1-\eta\tau}{\eta}\KLs{\bnut}{\nutm} + \frac{1-\eta\tau}{\eta}\KLs{\best{\nu}}{\nutm} - \frac{1}{\eta} \KLs{\best{\nu}}{\nut}.
	\label{eq:one_step_regret_decomp}
\end{align}
Here, the second step results from Lemma~\ref{lemma:strange_three_pt}. We further bound the first term in \eqref{eq:one_step_regret_decomp} as follows.
\begin{align*}
	&\normbig{\bnut(s) - \nut(s)}_1\normbig{\Q{t}(s)^\top\bmut(s) - \Q{t-1}(s)^\top\bar{\mu}^{(t-1)}(s)}_\infty\\
	&\le \normbig{\bnut(s) - \nut(s)}_1\Big(\normbig{\big(\Q{t}(s) - \Q{t-1}(s)\big)^\top \bar{\mu}^{(t-1)}(s)}_\infty + \normbig{\Q{t}(s) \big(\bmut(s) - \bar{\mu}^{(t-1)}(s)\big)}_\infty\Big)\\
	&\le \normbig{\bnut(s) - \nut(s)}_1\normbig{\Q{t}(s) - \Q{t-1}(s)}_\infty + \frac{2}{1-\gamma}\normbig{\bnut(s) - \nut(s)}_1\normbig{\bmut(s) - \bar{\mu}^{(t-1)}(s)}_1\\
	&\le \frac{2\eta}{1-\gamma}\normbig{\Q{t}(s) - \Q{t-1}(s)}_\infty + \frac{1}{1-\gamma}\Big[2\normbig{\bnut(s) - \nut(s)}_1^2 \\
	&\qquad + \normbig{\bmut(s) - \mutm(s)}_1^2 + \normbig{\mutm(s) - \bar{\mu}^{(t-1)}(s)}_1^2\Big]\\
	&\le \frac{2\eta}{1-\gamma}\normbig{\Q{t}(s) - \Q{t-1}(s)}_\infty + \frac{4}{1-\gamma}\KLs{\nut}{\bnut} \\
	&\qquad + \frac{2}{1-\gamma}\KL{\bmut(s)}{\mutm(s)} + \frac{2}{1-\gamma}\KL{\mutm(s)}{\bar{\mu}^{(t-1)}(s)},
\end{align*}
where the penultimate inequality follows from Lemma \ref{lemma:one_step_policy_bound}, and the last line follows from  Pinsker's inequality. Substitution of the above inequality into \eqref{eq:one_step_regret_decomp} completes the proof.

\subsection{Proof of Lemma \ref{lemma:cabbage}}
\label{sec:pf_lemma_cabbage}
By definition, we have
\begin{align*}
	\delta_{l,t} &= \alpha_l \prod_{i=l+1}^{t}(1 - c_1 \alpha_i)\\
	&= \alpha_l \prod_{i=l+1}^{t}(1 - c_2 \alpha_i + (c_2 - c_1)\alpha_i)\\
	& = \alpha_l(c_2-c_1)\alpha_{l+1}\prod_{i=l+2}^{t}(1 - c_2 \alpha_i + (c_2 - c_1)\alpha_i) +  \alpha_l (1-c_2\alpha_{l+1})\prod_{i=l+2}^{t}(1 - c_2 \alpha_i + (c_2 - c_1)\alpha_i).
\end{align*}
Continuing this expansion recursively, we obtain
\begin{align*}
\delta_{l,t}	&= \alpha_l \sum_{i=l+1}^{t} (c_2 - c_1)\alpha_i \cdot\prod_{j=l+1}^{i}(1-c_2\alpha_j)\cdot \prod_{k=i+1}^{t}(1-c_1\alpha_k) + \alpha_l \prod_{i=l+1}^t (1-c_2\alpha_i)\\
	&= (c_2 - c_1)\sum_{i=l+1}^{t} \xi_{l,i} \delta_{i,t} + \xi_{l,t}.
\end{align*}
Rearranging terms, it follows that
\begin{align*}
	\sum_{i=l}^{t}\xi_{l,i}\delta_{i+1,t} &= \alpha_l \delta_{l+1,t} + \sum_{i=l+1}^{t}\xi_{l,i}\delta_{i+1,t}\\
	& = \frac{\alpha_{l+1}}{1-c_1\alpha_{l+1}} \delta_{l,t} + \sum_{i=l+1}^{t}\xi_{l,i}\delta_{i,t}\cdot\frac{\alpha_{i+1}}{\alpha_{i}(1-c_1\alpha_{i+1})}\\
	&\overset{\mathrm{(i)}}{\le} \delta_{l,t} + 2 \sum_{i=l+1}^{t}\xi_{l,i}\delta_{i,t} = \delta_{l,t} + \frac{2}{c_2-c_1}(\delta_{l,t} - \xi_{l,t})\le \Big(1+\frac{2}{c_2-c_1}\Big)\delta_{l,t},
\end{align*}
where the second line results from the definition of $\delta_{l,t}$ and (i) is due to $\{\alpha_i\}$ being non-increasing and
\[
\alpha_{l+1} \le \eta\tau \le  1/2, \quad 1-c_1 \alpha_l \ge 1/2
\]
for all $l \ge 1$.

\subsection{Proof of Lemma \ref{lemma:strange_three_pt}}
 \label{sec:pf_three_pt}
We have
\begin{align*}
	\KL{z}{y} &= -\cH(z) + \cH(y) - \innprod{z-y, \log y}\\
	&= -\cH(z) + \cH(x) - \innprod{z-x, \log y} -\cH(x) + \cH(y) - \innprod{x-y, \log y}\\
	&= -\cH(z) + \cH(x) - \innprod{z-x, \log x} -\cH(x) + \cH(y) - \innprod{x-y, \log y} - \innprod{z-x, \log y - \log x}\\
	&= \KL{z}{x} + \KL{x}{y} - \frac{\eta}{1-\eta\tau}\innprod{z-x, w + \tau \log x},
\end{align*}
where the last line follows from the update rule. Rearranging terms gives
\begin{align*}
	\frac{\eta}{1-\eta\tau}\innprod{x-z, w } &= \KL{z}{y} -\KL{z}{x} - \KL{x}{y}  + \frac{\eta\tau}{1-\eta\tau}\innprod{z-x, \log x}.
\end{align*}
Adding $ \frac{\eta\tau}{1-\eta\tau}(-\cH(x)+\cH(z))$ to both sides, we are left with
\begin{align*}
	\frac{\eta}{1-\eta\tau}\Big[\innprod{x-z, w} - \tau \cH(x)+ \tau\cH(z)\Big] &= \KL{z}{y} - \KL{z}{x} - \KL{x}{y} \\
	&\qquad - \frac{\eta\tau}{1-\eta\tau} \big(- \cH(z) + \cH(x) - \innprod{z-x, \log x}\big)\\
	 &= \KL{z}{y} - \frac{1}{1-\eta\tau}\KL{z}{x} - \KL{x}{y}.
\end{align*}

\subsection{Proof of Lemma \ref{lemma:KL_local_equiv}}

We begin with a simple sandwich bound of $\log (1+x) $ which will be used later: when $x > -\frac{1}{10}$,
 we have  
	\begin{equation} \label{eq:logx_sandwich}
		x - \Big(\frac{1}{2} + \frac{|x|}{2}\Big)x^2 \le \log (1+x) \le x - \Big(\frac{1}{2} - \frac{|x|}{3}\Big)x^2.
	\end{equation}
We shall prove this at the end of this proof. The following lemma, which is standard (see, e.g., \cite[Lemma 23]{mei2020global}, \cite[Lemma 3]{cen2020fast}), allows us to control $\normbig{\log \pi - \log \pi'}_\infty$, and in turn $\norm{\pi /\pi'}_\infty$.
	\begin{lemma}
		Let $\pi, \pi' \in \Delta(\cA)$ satisfy $\pi(a) \propto \exp(\theta(a))$ and $\pi'(a) \propto \exp(\theta'(a))$ for some $\theta, \theta' \in \mathbb{R}^{|\cA|}$. It holds that
		\begin{equation*}
			\normbig{\log\pi - \log\pi'}_\infty \le 2\normbig{\theta - \theta'}_\infty.
		\end{equation*}
	\end{lemma}

In view of the above lemma, and since $ \normbig{w}_\infty < 1/30$, we have $\forall a\in \cA$:
	\begin{align} \label{eq:pumpkin}
	\Big|\frac{\pi(a)}{\pi'(a)} - 1\Big| &= \Big|\exp\Big(\log\frac{\pi(a)}{\pi'(a)}\Big) - \exp(0)\Big| \le |\log \pi(a) - \log \pi'(a)|\max\Big\{1, \frac{\pi(a)}{\pi'(a)}\Big\} \nonumber \\
	&\le 2\normbig{w}_\infty  \exp(2 \normbig{w}_\infty ) \le 3\normbig{w}_\infty .
	\end{align}
 Therefore, we can bound $\KL{\pi}{\pi'}$ as  
	\begin{align} 
			\KL{\pi}{\pi'} &= \sum_{a\in\cA} \pi(a)\log \frac{\pi(a)}{\pi'(a)} \nonumber\\
			&\overset{\mathrm{(i)}}{\le} \sum_{a\in\cA} \pi(a) \left( \frac{\pi(a)}{\pi'(a)}-1 - \Big(\frac{1}{2} - \normbig{w}_\infty \Big)\Big(\frac{\pi(a)}{\pi'(a)}-1\Big)^2 \right) \nonumber \\
			&\overset{\mathrm{(ii)}}{ =} \sum_{a\in\cA} \big(\pi(a) - \pi'(a)\big)\Big( \frac{\pi(a)}{\pi'(a)}-1\Big) + \sum_{a\in\cA} \pi'(a)\Big( \frac{\pi(a)}{\pi'(a)}-1\Big) - \Big(\frac{1}{2} - \normbig{w}_\infty  \Big)\sum_{a\in\cA}\pi(a)\Big(\frac{\pi(a)}{\pi'(a)}-1\Big)^2 \nonumber\\
			&= \chi^2(\pi;\pi') - \Big(\frac{1}{2} - \normbig{w}_\infty  \Big)\sum_{a\in\cA}\pi(a)\Big(\frac{\pi(a)}{\pi'(a)}-1\Big)^2 \nonumber\\
			&\overset{\mathrm{(iii)}}{\le} \chi^2(\pi;\pi') - \Big(\frac{1}{2} - \normbig{w}_\infty  \Big) \left(1-3\normbig{w}_\infty  \right)\sum_{a\in\cA}\pi'(a)\Big(\frac{\pi(a)}{\pi'(a)}-1\Big)^2 \nonumber \\
			&= \left(1 - \Big(\frac{1}{2} - \normbig{w}_\infty \Big)\left(1-3\normbig{w}_\infty \right) \right)\chi^2(\pi;\pi'),
		\label{eq:KL_chi_equiv_upper}
	\end{align}
where (i) follows from \eqref{eq:logx_sandwich}, (ii) utilizes the fact $\sum_{a\in\cA} (\pi(a) - \pi'(a)) =0$, and (iii) makes use of \eqref{eq:pumpkin}. On the other hand, by similar arguments, we have
		\begin{align} 
			\KL{\pi'}{\pi} &= \sum_{a\in\cA} \pi'(a)\log \frac{\pi'(a)}{\pi(a)} \nonumber \\
			& \ge \sum_{a\in\cA} \pi'(a)\left( \frac{\pi'(a)}{\pi(a)}- 1 - \frac{(1+3\normbig{w}_\infty )}{2}\Big(\frac{\pi'(a)}{\pi(a)}-1\Big)^2 \right) \nonumber\\
			&= \chi^2(\pi';\pi) - \frac{(1+3\normbig{w}_\infty )}{2}\sum_{a\in\cA}\pi'(a)\Big(\frac{\pi'(a)}{\pi(a)}-1\Big)^2 \nonumber\\
			&\ge \chi^2(\pi';\pi) - \frac{(1+3\normbig{w}_\infty )^2}{2}\sum_{a\in\cA}\pi(a)\Big(\frac{\pi'(a)}{\pi(a)}-1\Big)^2 \nonumber\\
			&= \Big(1 - \frac{(1+3\normbig{w}_\infty )^2}{2} \Big)\chi^2(\pi';\pi).
		\label{eq:KL_chi_equiv_lower}
	\end{align}
	By definition of $\chi^2(\pi;\pi')$, we further have
	\begin{align} 
			\chi^2(\pi;\pi') &= \sum_{a\in\cA}\pi'(a)\Big(\frac{\pi(a)}{\pi'(a)}-1\Big)^2  \nonumber\\
			&\le \normbig{\pi/\pi'}_\infty \sum_{a\in\cA}\frac{\big(\pi'(a) - \pi(a)\big)^2}{\pi(a)} \nonumber\\
			&\le (1+3 \normbig{w}_\infty )\chi^2(\pi';\pi),
		\label{eq:chi_equiv}
	\end{align}
where the last line uses \eqref{eq:pumpkin}. Combining \eqref{eq:KL_chi_equiv_upper}, \eqref{eq:KL_chi_equiv_lower} and \eqref{eq:chi_equiv} gives
	\begin{equation*}
		\KL{\pi}{\pi'} \le (1+3\normbig{w}_\infty)\cdot\left[ \frac{1 - \big(1/2 - \normbig{w}_\infty\big)(1-3\normbig{w}_\infty)}{1 - (1+3\normbig{w}_\infty)^2/2} \right] \KL{\pi'}{\pi}.
	\end{equation*}
	It is straightforward to verify that the factor is less than $2$ when $\normbig{w}_\infty \le 1/30$.
 
\paragraph{Proof of \eqref{eq:logx_sandwich}.} For any $x >-1$, it holds that
	\begin{align*}
	\log (1+x) &\le x - \frac{x^2}{2} + \frac{x^3}{3} \\
	&\le x - \frac{x^2}{2} + \frac{|x^3|}{3} = x - \Big(\frac{1}{2} - \frac{|x|}{3}\Big)x^2,
	\end{align*}
	and that
	\begin{align*}
		\log (1+x) &\ge x - \frac{x^2}{2} + \frac{x^3}{3(1+x)^3}\\
		&\ge x - \frac{x^2}{2} - \frac{|x^3|}{3(1+x)^3} = x - \Big(\frac{1}{2} + \frac{|x|}{3(1+x)^3}\Big)x^2.
	\end{align*}
	Therefore, when $x > -\frac{1}{10}$,
 we have $(1+x)^3>\frac{2}{3}$ and thus
	\begin{equation*}
		x - \Big(\frac{1}{2} + \frac{|x|}{2}\Big)x^2 \le \log (1+x) \le x - \Big(\frac{1}{2} - \frac{|x|}{3}\Big)x^2.
	\end{equation*}

\section{Further discussion regarding \citet{wei2021last}}

\label{sec:translation}


This section demonstrates how the last-iterate convergence result in \citet[Theorem 2]{wei2021last} in terms of the Euclidean distance to the set of NEs can be translated to that of the duality gap. Given any policy pair $\zeta = (\mu, \nu)$ and a NE $\zeta^\star = (\mu^\star, \nu^\star)$, we can invoke the performance difference lemma \eqref{eq:performance_diff_lemma_game} and obtain:
\begin{align*}
		V^{\mu,\nu}(\rho) - V^\star(\rho)&= \frac{1}{1-\gamma} \exlim{s'\sim d_\rho^{\mu, \nu}}{\mu(s')^\top Q^\star(s') \nu(s') - \mu^\star(s')^\top Q^\star(s') \nu^\star(s')}\\
		&\le \frac{1}{1-\gamma} \exlim{s'\sim d_\rho^{\mu, \nu}}{\max_{\mu'}\mu'(s')^\top Q^\star(s') \nu(s') - \mu^\star(s')^\top Q^\star(s') \nu^\star(s')}\\
		&= \frac{1}{1-\gamma} \exlim{s'\sim d_\rho^{\mu, \nu}}{\max_{\mu'}\mu'(s')^\top Q^\star(s') \nu(s') - \max_{\mu'}\mu'(s')^\top Q^\star(s') \nu^\star(s')}\\
		&\le  \frac{1}{1-\gamma} \exlim{s'\sim d_\rho^{\mu, \nu}}{\max_{\mu'}\mu'(s')^\top Q^\star(s')\big( \nu(s') - \nu^\star(s')\big)}\\
		&\le \frac{1}{(1-\gamma)^2} \exlim{s'\sim d_\rho^{\mu, \nu}}{\normbig{\nu(s') - \nu^\star(s')}_1}.
\end{align*}
Setting $\mu$ to the best-response policy of $\nu$, i.e., $\mu = \mu^\dagger(\nu) := \arg\max_{\mu} V^{\mu, \nu}(\rho)$, we get
\begin{align*}
	\max_{\mu'}V^{\mu',\nu}(\rho) - V^\star(\rho) &= V^{\mu^\dagger(\nu),\nu}(\rho) - V^\star(\rho)\\
	&\le \frac{1}{(1-\gamma)^2} \exlim{s'\sim d_\rho^{\mu^\dagger(\nu), \nu}}{\normbig{\nu(s') - \nu^\star(s')}_1}\\
	&\le \frac{\normbig{d_\rho^{\mu^\dagger(\nu), \nu}}_\infty}{(1-\gamma)^2} \sum_{s\in \cS}{\normbig{\nu(s) - \nu^\star(s)}_1}.
\end{align*}
Similarly, we have
\begin{align*}
	V^\star(\rho) - \min_{\nu'}V^{\mu,\nu'}(\rho)  \le \frac{\normbig{d_\rho^{\mu, \nu^\dagger(\mu)}}_\infty}{(1-\gamma)^2} \sum_{s \in \cS}{\normbig{\mu(s') - \mu^\star(s')}_1}.
\end{align*}
Taken together, the duality gap can be bounded by the policy pair's $\ell_1$ distance to the NE $(\mu^\star, \nu^\star)$ as
\begin{align*}
	\max_{\mu', \nu'} \Big[ V^{\mu',\nu}(\rho) - V^{\mu,\nu'}(\rho) \Big] &\le \frac{1}{(1-\gamma)^2}  \sum_{s\in \cS}\Big(\normbig{\nu(s') - \nu^\star(s')}_1 + \normbig{\mu(s') - \mu^\star(s')}_1\Big)\\
	&\le \frac{|\cS|^{1/2}(|\cA|+|\cB|)^{1/2}}{(1-\gamma)^2} \bigg[\sum_{s\in \cS} \Big(\normbig{\nu(s') - \nu^\star(s')}_2^2 + \normbig{\mu(s') - \mu^\star(s')}_2^2\Big)\bigg]^{1/2},
\end{align*}
where the second step results from Cauchy-Schwarz inequality. Finally, recall from \citet[Theorem 2]{wei2021last} that it takes at most
\[
\cO\bigg( \frac{|\cS|^2}{\eta^4 c^4 (1-\gamma)^4 \epsilon^2}\bigg)
\]
iterations to ensure
\[
\frac{1}{|\cS|} \sum_{s\in \cS} \Big(\normbig{\nu(s') - \nu^\star(s')}_2^2 + \normbig{\mu(s') - \mu^\star(s')}_2^2\Big) \le \epsilon^2,
\]
with $\eta^2 = \cO((1-\gamma)^5|\cS|^{-1})$.
Putting pieces together and minimizing the bound over $\eta$, this leads to an iteration complexity of
\[
\cO\bigg( \frac{|\cS|^5(|\cA|+|\cB|)^{1/2}}{(1-\gamma)^{16}c^4\epsilon^2}\bigg)
\]
to achieve $\epsilon$-NE in a last-iterate fashion.

\end{document}